\documentclass[11pt]{article}
\usepackage[utf8]{inputenc}
\usepackage[left=2.5cm,right=2.5cm,vmargin=2.5cm]{geometry}

\title{\textbf{Information Limits of Joint Community Detection and Finite Group Synchronization}}
\author{Yifeng Fan\thanks{Y. Fan is with the Department of Electrical and Computer Engineering, University of Illinois at Urbana-Champaign, Champaign, IL, 94404, USA (email: yifengf2@illinois.edu). Z. Zhao is with the Department of Electrical and Computer Engineering, University of Illinois at Urbana-Champaign, Champaign, IL, 94404, USA (email: zhizhenz@illinois.edu)}
	\and Zhizhen Zhao\footnotemark[1]}
\date{}

\usepackage{bookmark}
\usepackage[pdftex]{graphicx}
\usepackage{amssymb}
\usepackage{amsmath}
\usepackage{amsfonts}
\usepackage{subfig}
\usepackage{multirow}
\usepackage{enumerate}
\usepackage{amsthm}
\usepackage{mathrsfs}
\usepackage{enumitem}
\usepackage{booktabs}
\usepackage{tabularx}
\usepackage[linesnumbered, ruled,vlined]{algorithm2e}
\usepackage{array}
\usepackage{bm}
\usepackage{caption}
\usepackage{subfig}
\usepackage{bbm}
\usepackage{float}
\usepackage{color}
\usepackage{hyperref}
\usepackage{mathtools}
\SetKwInput{kwInit}{Initialize}

\newtheorem{theorem}{Theorem}[section]
\newtheorem{lemma}[theorem]{Lemma}

\newtheorem{claim}[theorem]{Claim}
\theoremstyle{definition}

\DeclareMathOperator*{\argmax}{arg\,max}
\DeclareMathOperator*{\argmin}{arg\,min}

\DeclarePairedDelimiter\ceil{\lceil}{\rceil}
\DeclarePairedDelimiter\floor{\lfloor}{\rfloor}

\begin{document}
\def\arXivver{1}

\maketitle

\begin{abstract}
The emerging problem of joint community detection and group synchronization, with applications in signal processing and machine learning, has been extensively studied in recent years. Previous research on this topic has predominantly focused on a statistical model that extends the well-known stochastic block model~(SBM) by incorporating additional group transformations. In its simplest form, the model randomly generates a network of size $n$ that consists of two equal-sized communities, where each node $i$ is associated with an unknown group element $g_i^* \in \mathcal{G}_M$ for some finite group $\mathcal{G}_M$ of order $M$. The connectivity between nodes follows a probability $p$ if they belong to the same community, and a probability $q$ otherwise. Moreover, a group transformation $g_{ij} \in \mathcal{G}_M$ is observed on each edge $(i,j)$, where $g_{ij} = g_i^* - g_j^*$ if nodes $i$ and $j$ are within the same community, and $g_{ij} \sim \text{Uniform}(\mathcal{G}_M)$ otherwise.
The goal of the joint problem is to recover both the underlying communities and group elements. Under this setting, when $p = a\log n /n$ and $q = b\log n /n $ with $a, b > 0$, we establish the following sharp information-theoretic threshold for exact recovery by maximum likelihood estimation~(MLE):
\begingroup\makeatletter\def\f@size{9}\check@mathfonts
$$ (i):\enspace   \frac{a + b}{2} -\sqrt{\frac{ab}{M}} > 1 \quad \text{and} \quad (ii):\enspace   a > 2$$\endgroup
where the exact recovery of communities is possible with high probability only if $(i)$ is satisfied, and the recovery of group elements is achieved with high probability only if both $(i)$ and $(ii)$ are satisfied. Our theory indicates the recovery of communities greatly benefits from the extra group transformations. Also, it demonstrates a significant performance gap exists between the MLE and all the existing approaches, including algorithms based on semidefinite programming and spectral methods.
\end{abstract}

\section{Introduction}
\label{sec:intro}

Community detection and group synchronization are both fundamental problems in signal processing, machine learning, and computer vision. Recently, there has been increasing interest in the joint problem of these two areas~\cite{fan2021joint,fan2019multi, fan2021spectral, chen2021non, bajaj2018smac, wang2023multi}. That is, in the presence of heterogeneous data where each sample is associated with an unknown group element (e.g. the orthogonal groups $\text{O}(d)$) and falls into multiple underlying communities, a network is observed that represents the interactions between data samples including the group transformations. Then, the joint problem aims to simultaneously recover the underlying community structures and the unknown group elements.
A motivating example is the 2D class averaging process in \textit{cryo-electron microscopy single-particle reconstruction}~\cite{frank2006,singer2011viewing,zhao2014rotationally, fan2019representation}, whose goal is to group and rotationally align projections images of a single particle with similar viewing angles, in order to improve their signal-to-noise ratio~(SNR). Another application in computer vision is \textit{simultaneous permutation group synchronization and clustering} on heterogeneous object collections consisting of 2D images or 3D shapes~\cite{bajaj2018smac, huang2012optimization, huang2013consistent}. 

This paper centers on a hybrid statistical model initially proposed in \cite{fan2021joint}, which combines the celebrated stochastic block model~(SBM)~(e.g.~\cite{decelle2011asymptotic, doreian2005generalized, dyer1989solution, fienberg1985statistical, holland1983stochastic, karrer2011stochastic, massoulie2014community, mcsherry2001spectral, mossel2012stochastic, mossel2018proof}) for community detection and the random rewiring model~(e.g.~\cite{singer2011viewing, fan2019multi, fan2019unsupervised, ling2020near, wang2013exact, chen2016information, chen2014near, huang2013consistent}) for group synchronization. Notably, most of the existing methods~\cite{fan2021joint,fan2021spectral, chen2021non, wang2023multi} for solving the joint problem are designed and evaluated based on this model. 
Formally, the model generates a random network~(graph) of size $n$ with $K$ underlying communities, each node $i$ is associated with its community assignment $\kappa_i^* \in \{1, \ldots, K\}$, along with an unknown group element $g_i^* \in \mathcal{G}_M$ for some finite group $\mathcal{G}_M$ of order $M$ (e.g.~the symmetric group $S_M$ and cyclic group $Z_M$). Then, edges are placed randomly and independently between nodes $(i,j)$ with probability $p$ if $i$ and $j$ belong to the same community, and with probability $q$ otherwise. In addition, a group transformation $g_{ij} \in \mathcal{G}_M$ is observed on each edge connection $(i,j)$ in a similar way that $g_{ij} = g_i^*(g_j^*)^{-1}$ if $i$ and $j$ are within the same community, which represents the clean measurement, otherwise $g_{ij} \in \text{Unif}(\mathcal{G}_M)$ that is uniformly drawn from $\mathcal{G}_M$ implying the measurement is completely noisy. In this paper, we focus on the case when $K = 2$ and the two communities are of equal sizes $n/2$, see Figure~\ref{fig:model_illustration} for an illustration.

\begin{figure*}[t!]
	\centering
	\if\arXivver1
	\includegraphics[width = 0.8\textwidth]{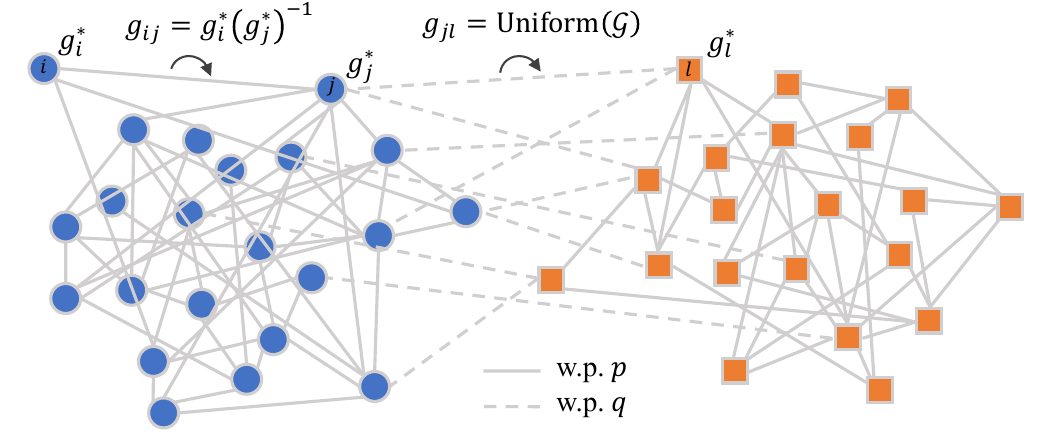}
	\else
	\includegraphics[width = 0.7\textwidth]{figures/data_illustration_v2.pdf}
	\fi
	\vspace{0.2cm}
	\caption{An illustration of our statistical model. We present a network that consists of two communities of equal sizes, shown in circles and squares respectively. Each node $i$ is associated with an unknown group element $g_i^* \in \mathcal{G}_M$ for some finite group $\mathcal{G}_M$ of size $M$.  Each pair of nodes within the same community (resp.~across different communities) are independently connected with probability $p$ (resp.~$q$) as shown in solid (resp.~dash) lines. A group transformation $g_{ij}$ is observed on each edge which satisfies $g_{ij} = g_{i}^*(g_{j}^*)^{-1}$ when nodes $i$ and $j$ are within the same community, otherwise $g_{ij} \sim \text{Unif}(\mathcal{G}_M)$ that is uniformly drawn from $\mathcal{G}_M$. Given the network, our goal is to recover the underlying communities and the associated group elements $\{g_i^*\}_{i = 1}^n$.}
	\label{fig:model_illustration}
\end{figure*}

Given the network observed, our goal is to recover the underlying communities as well as the unknown group elements. Notably, a naive two-stage approach, which first applies classical graph-clustering algorithms~(e.g.~\cite{hajek2016achievinga, frey2007clustering, feige2005spectral, abbe2020entrywise}) followed by synchronization within each identified community, leads to sub-optimal results, since it does not leverage the consistency~(resp.~inconsistency) of group transformations within each community~(resp.~across different communities). To address this limitation, several improved optimization criteria are proposed~\cite{bajaj2018smac, fan2021joint, wang2023multi} incorporating these consistencies. Due to the non-convex and computationally intractable nature of directly solving such optimization programs, convex relaxation techniques such as semidefinite programming~(SDP)~\cite{fan2021joint} and spectral method~\cite{fan2021spectral, bajaj2018smac, chen2021non, wang2023multi} have emerged, which provide approximated solutions with polynomial time complexity. Importantly, such efficient methods have significantly enhanced the recovery accuracy compared to the aforementioned two-stage approach.

Despite the progress made before, various fundamental questions remain unsettled such as:
\begin{enumerate}
	\vspace{0.1cm}
	\item \textit{What is the fundamental limit for achieving exact recovery of the clusters or group elements? }
	\vspace{0.1cm}
	\item \textit{Is the fundamental limit achieved by the existing algorithms?}
	\vspace{0.1cm}
\end{enumerate}
Finding this limit enables us to benchmark the performance of existing algorithms and assess their potential for improvement. 
\textcolor{black}{Notably, the joint problem for recovery can be viewed as a channel decoding problem from an information theory perspective (similar to the argument made in~\cite[Page 4]{abbe2015exact}). Specifically, the input to the channel consists of the $n$ nodes that encodes the clustering labels and group elements, while the channel outputs correspond to the ${n \choose 2}$ pairwise observations -- each representing whether a connection exists between a node pair and, if so, the associated group transformation observation. In this way, the algorithm that maximizes the probability of correctly recovering the communities and the group elements is the Maximum A Posteriori (MAP) decoder. Moreover, since the community assignments are assumed to be uniformly distributed, the MAP decoder coincides with the Maximum Likelihood (ML) decoder.
Therefore, the performance of MLE serves as the \textit{information-theoretic limits} such that if MLE fails in exact recovery with high probability as the network size $n$ goes to infinity, no algorithm, regardless of its complexity, can succeed with high probability.
}

\if\arXivver1
\subsection{Our contribution} 
\else
\subsection{Our contribution}
\fi
In this work, we answer the two questions above by establishing the information-theoretic limits for the exact recovery of (i) communities, (ii) group elements, and both of them, under the statistical model aforementioned. For ease of presentation, we restrict our attention to the simplest scenario of two communities~($K = 2$) of equal size, which allows us to convey the most informative findings without introducing unnecessary complexity. In a nutshell, on the sparse regime such that $p = \frac{a\log n}{n}$ and $q = \frac{b \log n}{n}$ for some constants $a, b > 0$, and as $n$ is sufficiently large, we identify sharp information-theoretic thresholds for the exact recovery of (i) and (ii) above as follows:
\begin{itemize}
	\vspace{0.1cm}
	\item [(i)] Recovering the communities is possible if and only if
	\begin{equation*}
	\if\arXivver1
	\hspace{2.5cm} 
	\fi
	\frac{a + b}{2} -\sqrt{\frac{ab}{M}} > 1
	\label{eq:intro_recover_clus}
	\end{equation*}
	
	\vspace{0.1cm}
	\item [(ii)] Recovering the group elements is possible if and only if
	\begin{equation*}
	\if\arXivver1
	\hspace{2.5cm} 
	\fi
	\frac{a + b}{2} -\sqrt{\frac{ab}{M}} > 1 \quad \text{and} \quad \frac{a}{2} > 1
	\label{eq:intro_recover_rot}
	\end{equation*}
\vspace{0.1cm}
\end{itemize}
Here, several remarks are worthy to be highlighted: 
\begin{enumerate}
    \vspace{0.2cm}
    \item The recovery of communities is the \textit{prerequisite} for recovering the group elements. In other words, one cannot hope to identify the group elements without determining the communities at first.
    \vspace{0.2cm}
    \item When $M = 1$, our model degenerates to the SBM and the condition in (i) becomes $$(\sqrt{a} - \sqrt{b})^2 > 2$$  which agrees with the existing result of information-theoretic limit on the SBM~\cite{abbe2015exact, abbe2015community, massoulie2014community, mossel2018proof}.
    \vspace{0.2cm}
    \item By examining the performance of existing efficient algorithms~\cite{fan2021spectral, bajaj2018smac, chen2021non, wang2023multi, fan2021joint}, our theory demonstrates a significant performance gap between the information-theoretic limit and those approaches~(see Figure~\ref{fig:threshold_clus}), suggesting a considerable room of improvement exists (see Section~\ref{sec:discussion} for a detailed discussion). 
    \vspace{0.2cm}
    \item Our proof strategy is substantially distinct from (and more difficult than) previous works~(e.g.~\cite{chen2016information, abbe2015exact, mossel2018proof, singer2011angular}) which focuses on community detection or group synchronization individually, and the difficulty comes from dealing with the coupling of these two problems. 
\end{enumerate}

\if\arXivver1
\subsection{Organization} 
\else
\subsection{Organization}
\fi 
The rest of this paper is organized as follows: In the remaining of Section~\ref{sec:intro}, we introduce some additional related works and the notations for our analysis. Then, in Section~\ref{sec:preliminary} we formulate the MLE program for recovery on our statistical model and define the evaluation measurement. Next, in Section~\ref{sec:main_results} we present our theory on the information-theoretic limits, and Section~\ref{sec:discussion} is devoted to discussions on the result. We end up with proofs of the theorems in the remaining sections, from Section~\ref{sec:proof_clus_upper_bound} for Theorem~\ref{the:info_upper_bound_1}, to Section~\ref{sec:proof_cond_lower_rot} for Theorem~\ref{the:cond_lower_rot}. We leave the technical details to the Appendix.

\subsection{Related works}

Given its practical importance to various scientific applications, either community detection or group synchronization has been extensively studied over the past decades. Therefore, this section by no means provides a comprehensive review of all the previous works but only highlights the most related ones. 

Community detection is commonly studied on the SBM~(e.g.~\cite{decelle2011asymptotic, doreian2005generalized, dyer1989solution, fienberg1985statistical, holland1983stochastic, karrer2011stochastic, massoulie2014community, mcsherry2001spectral, mossel2012stochastic, mossel2018proof}) where maximum likelihood is equivalent to minimizing the number of edges between different communities. Then, a major line of research~(e.g.~\cite{bui1984graph, dyer1989solution, boppana1987eigenvalues, snijders1997estimation, jerrum1998metropolis, condon2001algorithms, carson2001hill, mcsherry2001spectral, rohe2011spectral, choi2012stochastic}) attempts to find the information-theoretic limits for exact recovery, especially on the case of two communities, until the two seminal papers~\cite{abbe2015exact, mossel2014consistency} first characterize a sharp threshold for the case of two equal-sized communities, which greatly inspires the present work and is further extended to more general scenarios (e.g.~\cite{abbe2015community, hajek2016achievingb}). In terms of algorithms for recovery, as exactly solving the MLE is often NP-hard, different approaches such as semidefinite programming~(SDP)~\cite{abbe2015exact, hajek2016achievinga, hajek2016achievingb, perry2017semidefinite, guedon2016community, amini2018semidefinite, bandeira2018random}, spectral method~\cite{abbe2020entrywise, vu2014simple, yun2014accurate, massoulie2014community, krzakala2013spectral, ng2002spectral, gao2017achieving}, and belief propagation~\cite{decelle2011asymptotic, abbe2015community} are proposed for obtaining an approximated solution to the MLE. However, many of them have been proven optimal such that they can achieve the information-theoretic limits at least on two equal-sized communities~\cite{hajek2016achievinga, hajek2016achievingb, abbe2020entrywise, yun2014accurate, lei2015consistency, perry2017semidefinite, gao2017achieving}. 

On the side of group synchronization, the goal is to recover the underlying group elements from a set of noisy pairwise measurements. A large body of literature focuses on the aforementioned random rewiring model~(e.g.~\cite{singer2011viewing, fan2019multi, fan2019unsupervised, ling2020near, wang2013exact, chen2016information, chen2014near, huang2013consistent}) which independently replaces each pairwise clean measurement with some random element. 
\textcolor{black}{The fundamental information-theoretic limits for recovery are generally studied in several prior works including~\cite{chen2016information, perry2018message, yang2024asymptotic} regardless of the specific choice of group. Moreover, \cite{javanmard2016phase} focuses on SO(2) (angular synchronization) and 
 U(1) group and provides asymptotic thresholds of recovery for several popular estimators including minimum MSE (Bayes Optimal Estimator) and MLE. Regarding the algorithms, similar to the development of community detection, various approaches are developed for recovery depending on the specific group of interest, including semidefinite relaxation~\cite{singer2011angular, huang2013consistent},  spectral methods~\cite{singer2011angular, arrigoni2016spectral, chaudhury2015global, pachauri2013solving, shen2016normalized, gao2019multi}, and (approximate)~message passing~\cite{shi2020message, perry2018message} (which is conjectured to achieve the information-theoretic limits), along with many theoretical investigations on their performances~\cite{singer2011angular, zhong2018near, ling2020near, huang2013consistent, ling2020solving, gao2021optimal, zhang2022exact}. }

The joint problem of community detection and group synchronization is an emerging research topic~(e.g.~\cite{fan2019multi,fan2019unsupervised, bajaj2018smac, lederman2019representation}) motivated by recent scientific applications such as the cryo-electron microscopy single-particle reconstruction~\cite{frank2006,singer2011viewing,zhao2014rotationally, fan2019representation} aforementioned. In particular, the statistical model adopted in this paper was initially proposed in~\cite{fan2021joint}, which also studied exact recovery by SDP with the performance guarantee provided. Recently, several spectral method-based approaches~\cite{fan2021spectral, chen2021non, wang2023multi} are designed for efficient recovery. As there is a lack of indicators on the performance of these algorithms, the information-theoretic limit obtained in this work fills in the blanks. \textcolor{black}{A detailed discussion on the existing algorithms and the comparison between their performances with the information-theoretic limit, is provided in Section~\ref{sec:discussion}.}

\begin{table*}[t!]
	\centering
	\small
	\caption{A summary of notations for the model. The notations with $^*$ indicate the ground truth parameters.}
	\begin{tabular}{l l|l} \toprule
		\multicolumn{2}{c|}{Definition} & \multicolumn{1}{c}{Description} \\ \midrule
		$\bm{\kappa}^* := \{\kappa_1^*, \ldots, \kappa_n^*\}$,  &$\bm{\kappa} := \{\kappa_1, \ldots, \kappa_n\}$  & The set of true (estimated) community assignment \\
		$\bm{g}^* := \{g_1^*, \ldots, g_n^*\}$, &$\bm{g} := \{g_1, \ldots, g_n\}$ & The set of true (estimated) group elements \\
		$S_k^* := \{i | \kappa_i^* = k\}$, &$S_k := \{i | \kappa_i = k\}$ & The set of nodes in the $k$-th true (estimated) community \\
		$n_k^* := |S_k^*|$, &$n_k := |S_k|$ & The size of the $k$-th true (estimated) community\\
		$\bm{x}^* := \{\bm{\kappa}^*, \bm{g}^*\}$, &$\bm{x} := \{\bm{\kappa}, \bm{g}\}$ & The true (estimated) parameters\\ \midrule
		\multicolumn{2}{c|}{$\mathcal{E}$} & The set of edge connections \\
		\multicolumn{2}{c|}{$\mathfrak{G} := \{g_{ij} | (i, j) \in \mathcal{E}\}$} & The set of group transformations \\
		\multicolumn{2}{c|}{$\bm{y} := \{\mathcal{E}, \mathfrak{G}\}$} & The observation of the network\\
		\multicolumn{2}{c|}{$\mathcal{E}(S)$} & The set of edges within the node set $S$ \\
		\multicolumn{2}{c|}{$\mathcal{E}(S_1, S_2)$} & The set of edges across the node sets $S_1$ and $S_2$ \\
		\bottomrule
	\end{tabular}
	\label{tab:notations}
\end{table*}

\subsection{Notations}
\label{sec:notations}

Here, we define a series of notations for our statistical model that will be frequently used. Given the network of size $n$ with $K$ underlying communities, we denote $\bm{\kappa}^* := \{\kappa_1^*, \ldots, \kappa_n^*\}$ as the set of true community assignment, and $\bm{g}^* := \{g_1^*, \ldots, g_n^*\}$ as the set of true group elements. Similarly, the estimated community assignments and group elements are denoted as $\bm{\kappa} := \{\kappa_1, \ldots, \kappa_n\}$ and $\bm{g} := \{g_1, \ldots, g_n\}$, respectively. We use $n_k^*$ and $S_k^*$ to denote the size of the $k$-th community and the set of nodes belonging to it, respectively, i.e.,$S_k^* := \{i | \kappa_i = k\}$ and $n_k^* := |S_k^*|$ for $k = 1,\ldots K$, and we define $S_k := \{i | \kappa_i = k\}$ and $n_k := |S_k|$ analogous to $S_k^*$ and $n_k^*$. 
Then, we denote $\bm{x}^* := \{\bm{\kappa}^*, \bm{g}^*\}$ and $\bm{x} := \{\bm{\kappa}, \bm{g}\}$ as the true and estimated parameters. By denoting $\mathcal{E}$ as the set of edge connections and $\mathfrak{G} := \{g_{ij} | (i, j) \in \mathcal{E}\}$ as the set of group transformations observed on each edge, we define $\bm{y} := \{\mathcal{E}, \mathfrak{G}\}$ as the whole observation of the network. Lastly, given any set of nodes $S$, we define $\mathcal{E}(S) = \{(i,j) | i,j \in S, (i,j) \in \mathcal{E}\}$ as the set of edge connections with $S$. For any two sets of nodes $S_1$ and $S_2$, we denote $\mathcal{E}(S_1, S_2) = \{(i, j) |i \in S_1, j \in S_2, (i,j) \in \mathcal{E} )\}$ as the set of edges connecting $S_1$ and $S_2$. 
We summarize these notations in Table~\ref{tab:notations}.

Besides, we use the following standard notations for our analysis: for two non-negative functions $f(n)$ and $g(n)$, $f(n) = O(g(n))$ means there exists an absolute positive constant $C$ such that $f(n) \leq Cg(n)$ for all sufficiently large $n$; $f(n) = o(g(n))$ indicates for every positive constant $C$, the inequality $f(n) \leq Cg(n)$ holds for all sufficiently large $n$; $f(n) = \omega(g(n))$ denotes for every positive constant $C$, the inequality $f(n) \geq Cg(n)$ for all sufficiently large $n$; and $f(n) = \Theta(g(n))$ represents there exists two absolute positive constants $C_1, C_2$ such that $C_1g(n) \leq f(n) \leq C_2g(n)$.

\section{Preliminary}
\label{sec:preliminary}

\subsection{Maximum likelihood estimation}
\label{sec:problem_formulation}
Given an observation $\bm{y} = \{\mathcal{E}, \mathfrak{G} \}$ of the network, our goal is to recover the true parameters $\bm{x}^* = \{\bm{\kappa}^*, \bm{g}^*\}$  including both the community assignment $\bm{\kappa}^*$ and the group elements $\bm{g}^*$. Recall that in our statistical model, each pair of nodes are independently connected by following
\begin{equation*}
	\mathcal{P}((i,j) \in \mathcal{E}) = 
	\begin{cases}
	p, &\quad \kappa_i^* = \kappa_j^* \\[2pt]
	q, &\quad \textrm{otherwise}
	\end{cases}
\end{equation*}
Also, the group transformation $g_{ij}$ observed on each edges $(i,j) \in \mathcal{E}$ follows
\begin{equation*}
	g_{ij} = 
	\begin{cases}
		g_i^*\left(g_j^*\right)^{-1}, &\quad \kappa_i^* = \kappa_j^* \\[2pt]
		\text{Uniform}(\mathcal{G}_M), &\quad \textrm{otherwise}
	\end{cases}
\end{equation*}
Given the above, this work focuses on the maximum likelihood estimator~(MLE) that maximizes $\mathcal{P}(\bm{y}| \bm{x})$. Formally, we have
\begin{claim}
The maximum likelihood estimator~(MLE) for estimating $\bm{x}^* = \{\bm{\kappa}^*, \bm{g}^*\}$ from $\bm{y} = \{\mathcal{E}, \mathfrak{G} \}$ is  
\begin{equation}
	\boxed{
		\begin{aligned}
			\psi_{\textup{MLE}}(\bm{y}) &= 
			\begin{cases}
				\displaystyle \argmax_{\bm{\kappa}, \;\bm{g}} \;  |\mathcal{E}_{\textup{inner}}|,  &\; \frac{Mp(1-q)}{q(1-p)} > 1\\
				\displaystyle \argmin_{\bm{\kappa}, \;\bm{g}} \;  |\mathcal{E}_{\textup{inner}}|,  &\; \textup{otherwise}
			\end{cases}
			\\[5pt]
			\textrm{s.t.} & \quad n_k = n_k^*, \quad k = 1,\ldots, K \\[2pt]
			&\quad g_{ij} =  g_i g_j^{-1}, \quad \forall (i,j) \in \mathcal{E}_{\textup{inner}}
		\end{aligned}
		\label{eq:MLE_rewrite}
	}
\end{equation}
where 
\begin{equation}
	\mathcal{E}_{\textup{inner}} := \{(i,j) \;|\; (i,j) \in \mathcal{E}, \; i < j, \; \kappa_i = \kappa_j \}
	\label{eq:def_E_innner}
\end{equation}
denotes the set of edges within the same community. 
\end{claim}

The derivation of \eqref{eq:MLE_rewrite} is deferred to Section~\ref{sec:proof_mle_derivation}. Here, we assume the cluster sizes $\{n_k^*\}_{k = 1}^K$ are given and therefore the constraint $n_k = n_k^*$ ensures our recovery is valid.  

According to \eqref{eq:MLE_rewrite}, finding the MLE is equivalent to performing the following two tasks simultaneously:
\begin{enumerate}
	 \vspace{0.2cm}
	\item (\textit{For clustering}) Partitioning the data network into $K$ disjoint clusters such that the number of in-cluster edges is \textit{maximized} or \textit{minimized}.
	\vspace{0.2cm}
	\item (\textit{For synchronization}) On each in-cluster edge, the observed group transformation $g_{ij}$ should satisfy the \textit{consistency} $g_{ij} = g_i g_j^{-1}$.
	 \vspace{0.2cm}
\end{enumerate}
Here, as a special case when $M = 1$ such that no group synchronization is needed, \eqref{eq:MLE_rewrite} essentially becomes the MLE on the SBM:
\begin{equation}
	\begin{aligned}
		\psi_{\text{MLE}}(\bm{y}) &= 
		\begin{cases}
			\displaystyle \argmax_{\bm{\kappa}} \;  |\mathcal{E}_{\text{inner}}|,  &\; p > q\\
			\displaystyle \argmin_{\bm{\kappa}} \;  |\mathcal{E}_{\text{inner}}|,  &\; \text{otherwise}
		\end{cases}
		\\[5pt]
		\textrm{s.t.} & \quad n_k = n_k^*, \quad k = 1,\ldots, K \\[2pt]
	\end{aligned}
	\label{eq:MLE_SBM}
\end{equation}
which simply maximizes or minimizes the number of in-cluster edges. As a result, \eqref{eq:MLE_rewrite} can be treated as an extension of \eqref{eq:MLE_SBM} with additional constraints on the group elements included.

In this work, although we focus on the scenario when $\mathcal{G}_M$ is a finite group with order $M$, notice that by letting $M \rightarrow \infty$ that goes to infinity, \eqref{eq:MLE_rewrite} becomes
\begin{equation}
		\begin{aligned}
			\psi_{\textup{MLE}}(\bm{y}) &= \argmax_{\bm{\kappa}, \;\bm{g}} \;  |\mathcal{E}_{\textup{inner}}|
			\\[5pt]
			\textrm{s.t.} & \quad n_k = n_k^*, \quad k = 1,\ldots, K \\[2pt]
			&\quad g_{ij} =  g_i g_j^{-1}, \quad \forall (i,j) \in \mathcal{E}_{\textup{inner}}
		\end{aligned}
		\label{eq:MLE_infinite}
\end{equation}
which is the MLE for the case of an infinite group, e.g.~the orthogonal group $\textrm{O}(d)$ and the rotation group $\textrm{SO}(d)$.  In this scenario, we always aim to maximize the number of in-cluster edges regardless of the model parameters. 

From an algorithm perspective, the program~\eqref{eq:MLE_rewrite} is obviously non-convex and is thus computationally intractable to be exactly solved. Moreover, even solving \eqref{eq:MLE_SBM} is well-known to be NP-hard, including the famous \textit{minimum bisection} problem~(see e.g.~\cite{garey1974some}) as a special case when the network consists of two equal-sized communities.  This challenge gives rise to several approximate but efficient algorithms proposed in \cite{fan2021spectral, chen2021non, wang2023multi, fan2022joint}.

\subsection{Evaluation measurement}
This paper centers on exact information recovery that the true parameter $\bm{\kappa}^*$ and $\bm{g}^*$ are precisely recovered by \eqref{eq:MLE}. For evaluating the recovery of clusters, given any cluster  assignments $\bm{\kappa}$, we introduce the following \textit{zero-one distance} between $\bm{\kappa}^*$ and $\bm{\kappa}$:
\begin{equation}
	\textrm{dist}_{\text{c}}(\bm{\kappa}, \bm{\kappa}^*) := 1 - \max_{\pi \in \Pi_K} \; \mathbbm{1} \big(\kappa^*_i = \pi(\kappa_i); \; i = 1,\ldots, n\big)
	\label{eq:cluster_dist}
\end{equation}
where $\Pi_K$ denotes the permutation group on $\{1, \ldots, K\}$. Notably, $\Pi_K$ is introduced to eliminate the ambiguity on the permutation of cluster labels. That is, any permutation on the labels does not affect the underlying partition, and thus the evaluation metric should be invariant to permutations. The same idea applies to the metric for group elements:  one can only identify them within each community up to some global offset $\bar{g} \in \mathcal{G}_M$, as there is no basis to distinguish $g_i$ with $g_i\bar{g}$ given the pairwise transformations. Therefore, similar to \eqref{eq:cluster_dist}, we define 
\begin{equation}
	\text{dist}_{\text{g}}(\bm{g}, \bm{g}^*) := \sum_{k = 1}^K \left[1 - \max_{\bar{g}_k \in \mathcal{G}_M} \mathbbm{1}\left(g_i^* = g_{i} \bar{g}_k; \; \forall i \in S_k^*\right)\right]
	\label{eq:rotation_dist_k}
\end{equation}
as the distance between $\bm{g}$ and $\bm{g}^*$. Here, we compute the distance on each cluster by maximizing over its offset $\bar{g}_k$, and then sum over all clusters. As a result, \eqref{eq:rotation_dist_k} is zero only if all the group elements are exactly recovered. 

With all the metrics in place, we define, for any recovery procedure $\psi(\bm{y})$ which takes some observation $\bm{y} = \{\mathcal{E},\mathfrak{G}\}$ as input and output the estimation $\bm{x} = \{\bm{\kappa}, \bm{g}\}$, the error probability of $\bm{\kappa}$ and $\bm{g}$ as 
\begin{equation}
	\begin{aligned}
	\mathcal{P}_{e, \text{c}}(\psi) &= \max_{x^{*} = \{\bm{\kappa}^*, \bm{g}^*\}} \mathcal{P}\left(\textrm{dist}_\text{c}\left(\bm{\kappa}, \bm{\kappa}^*\right) \neq 0 \mid \bm{x}^*\right), \\
	\mathcal{P}_{e, \text{g}}(\psi) &= \max_{x^{*} = \{\bm{\kappa}^*, \bm{g}^*\}} \mathcal{P}\left(\textrm{dist}_\text{g}\left(\bm{g}, \bm{g}^*\right) \neq 0 \mid \bm{x}^*\right).
	\end{aligned}
	\label{eq:error_prob}
\end{equation}
In other words, $\mathcal{P}_{e, \text{c}}(\psi)$ (resp.~$\mathcal{P}_{e, \text{g}}(\psi)$) stands for the maximum probability of giving a wrong estimation $\bm{\kappa}$ (resp.~$\bm{g}$) by $\psi$ over all possible ground truth $\bm{x}^*$. 

We remark that in our model setting, the error probability does not depend on the specific choice of $\bm{x}^*$ as there is no preference. Therefore, one should expect both $\mathcal{P}\left(\textrm{dist}_\text{c}\left(\bm{\kappa}, \bm{\kappa}^*\right) \neq 0 \mid \bm{x}^*\right)$  and $\mathcal{P}\left(\textrm{dist}_\text{g}\left(\bm{g}, \bm{g}^*\right) \neq 0 \mid \bm{x}^*\right)$ are invariant to $\bm{x}^*$

\section{Main results}
\label{sec:main_results}
This section presents the information-theoretic limits for exact recovery on the statistical model defined in Section~\ref{sec:intro}. \textcolor{black}{We focus on the scenario of having two communities ($K=2$) of equal sizes ($n_1 = n_2$), and the order of group $M$ is assumed to be a fixed integer as $n$ grows.} Notably, even though we aim to recover both the (i) communities and (ii) group elements simultaneously, there might exist scenarios where the recovery of one is possible but impossible for the another. Because of this, we study the conditions for recovering (i) and (ii) separately. Specifically, for each of them, we provide
\begin{itemize}
	\vspace{0.1cm}
	\item Information-theoretic \textit{upper} bound,  above which the MLE~\eqref{eq:MLE_rewrite} achieves exact recovery with high probability~(w.h.p.)
	\vspace{0.1cm}
	\item Information-theoretic \textit{lower} bound,  below which the MLE~\eqref{eq:MLE_rewrite} fails to {\color{black} exactly} recover with a probability bounded away from zero
\end{itemize}

\subsection{Exact recovery of communities}
\label{sec:exact_recovery_clus}
\begin{theorem}[Upper bound]
Under the setting of two equal-sized clusters. Let $p = \frac{a\log n}{n}$ and $q = \frac{b \log n}{n}$ for $a, b \geq 0$. For any $c > 0$, if 
	\begin{equation}
		\frac{a + b}{2} - (1 + o(1))\sqrt{\frac{ab}{M}} > 1 + c,
		\label{eq:cond_upper_clus}
	\end{equation}
	then the MLE~\eqref{eq:MLE_rewrite} exactly recovers the communities w.p.~at least $1 - n^{-c}$, i.e.,$\mathcal{P}_{e, \mathrm{c}}(\psi_{\textup{MLE}}) \leq n^{-c}$.
	\label{the:info_upper_bound_1}
\end{theorem}

\begin{theorem}[Lower bound]
Under the same setting of Theorem~\ref{the:info_upper_bound_1}, if 
\begin{equation}
	\frac{a + b}{2} - \sqrt{\frac{ab}{M}} < 1,
	\label{eq:cond_lower_clus}
\end{equation}
then the MLE~\eqref{eq:MLE_rewrite} fails to exactly recover the communities  w.p.~at least $\frac{3}{5}$, i.e., $\mathcal{P}_{e, \mathrm{c}}(\psi_{\textup{MLE}}) \geq \frac{3}{5}$.
	\label{the:cond_lower_clus}
\end{theorem}

\begin{figure*}[t!]
	\centering
	\subfloat[$M = 1$]{\includegraphics[width=0.33\textwidth]{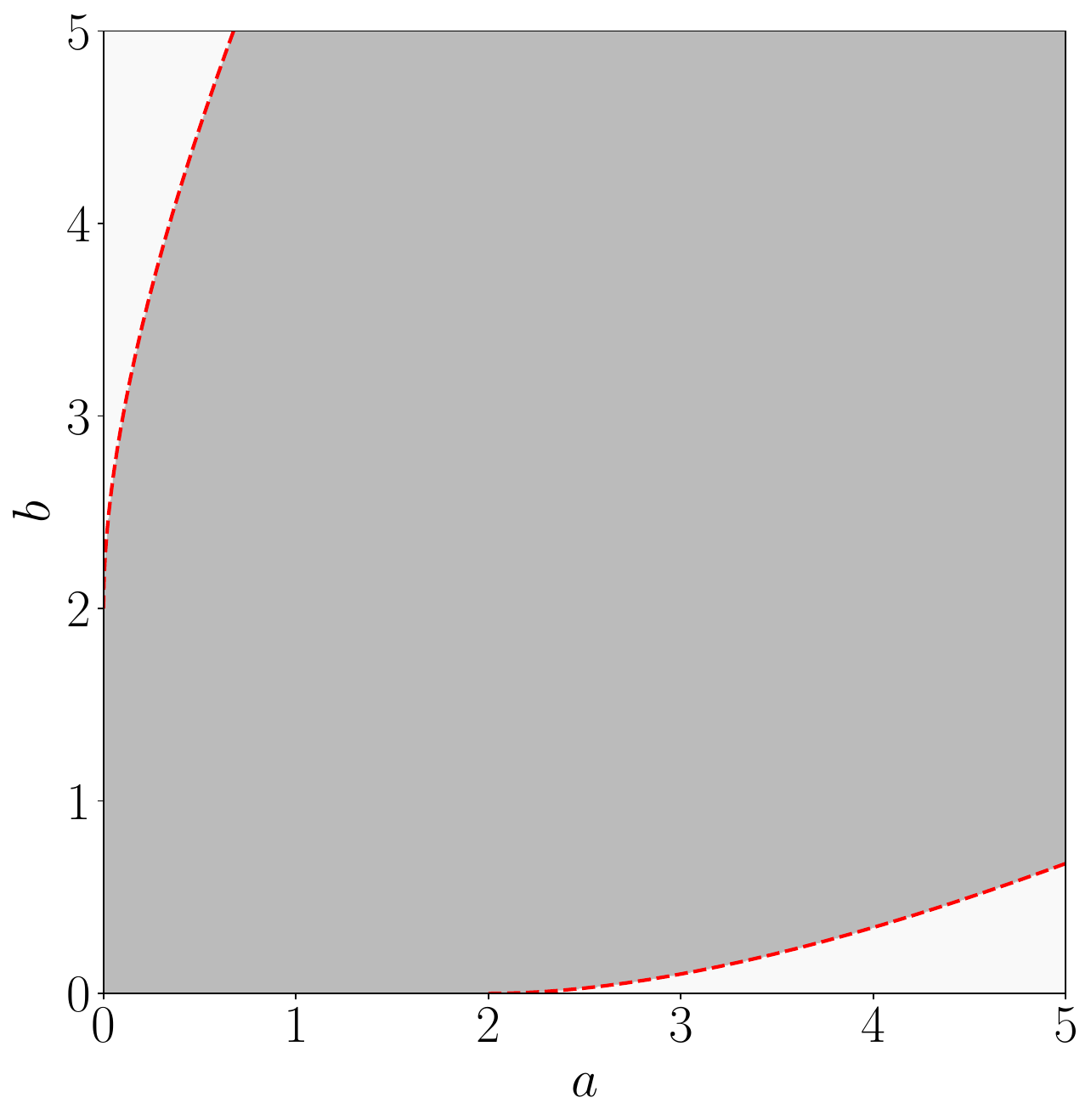}} \hskip 15pt
	\subfloat[$M = 2$]{\includegraphics[width=0.33\textwidth]{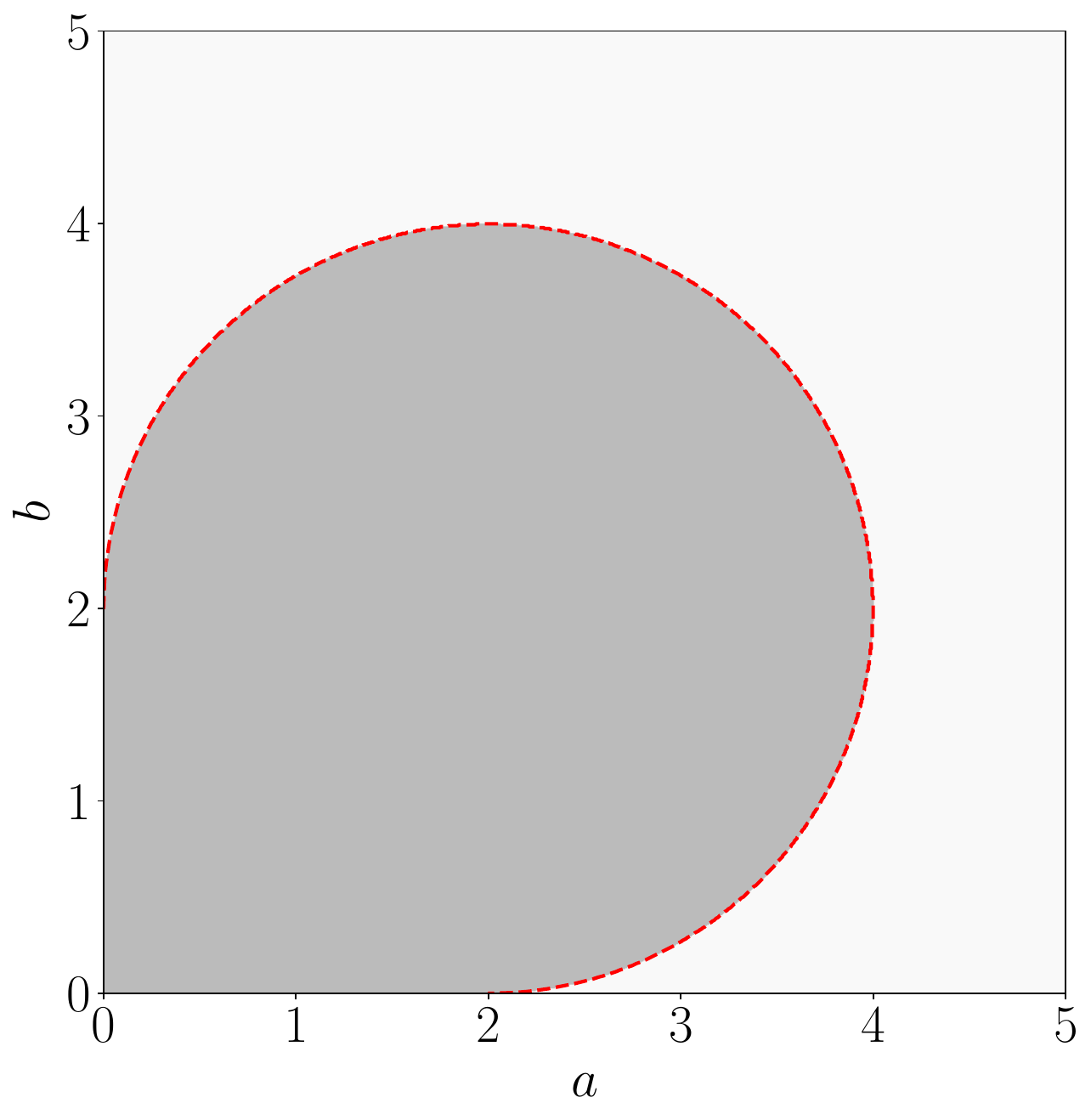}}\\
	\subfloat[$M = 5$]{\includegraphics[width=0.33\textwidth]{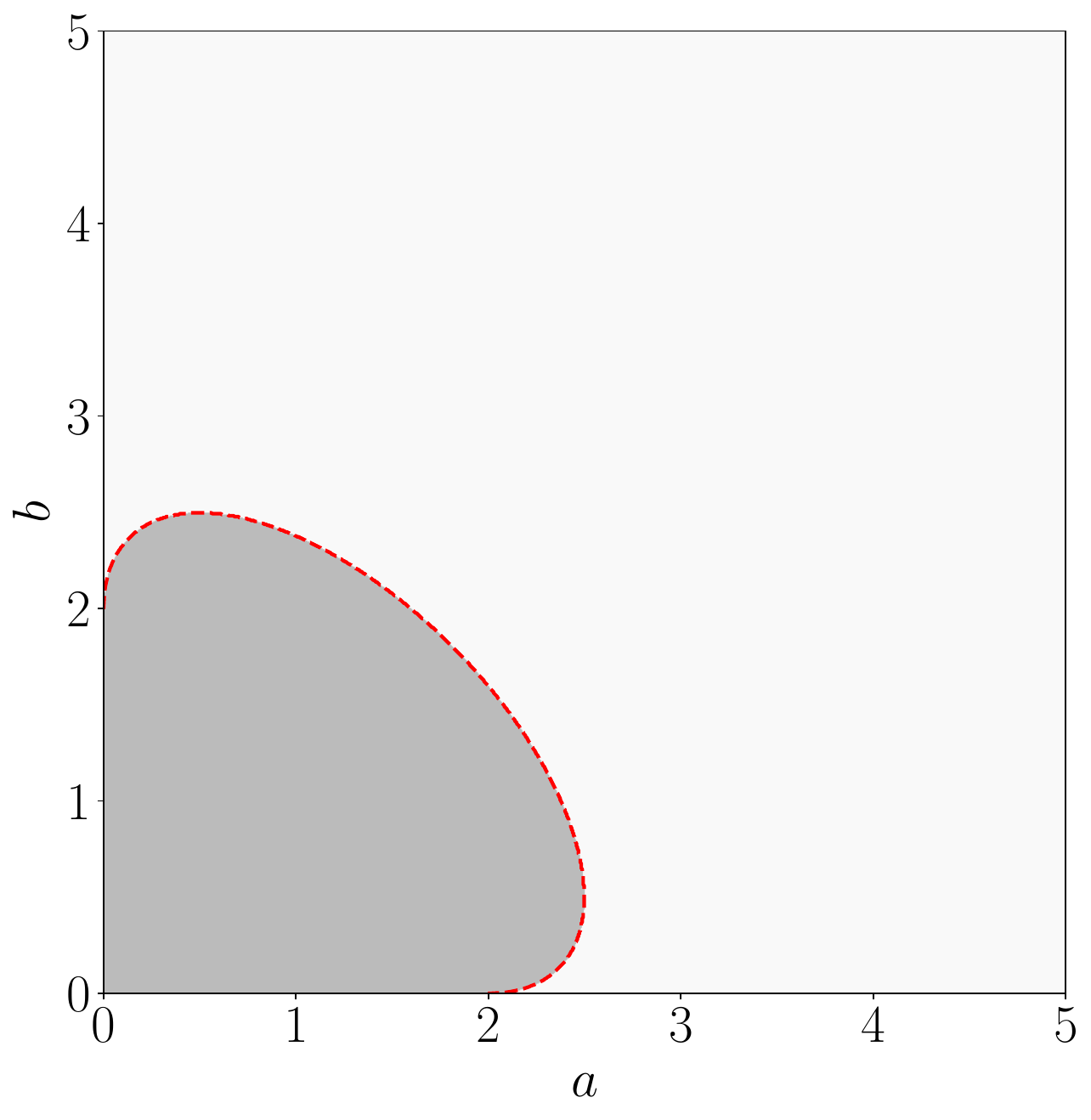}} \hskip 15pt
	\subfloat[$M = 50$]{\includegraphics[width=0.33\textwidth]{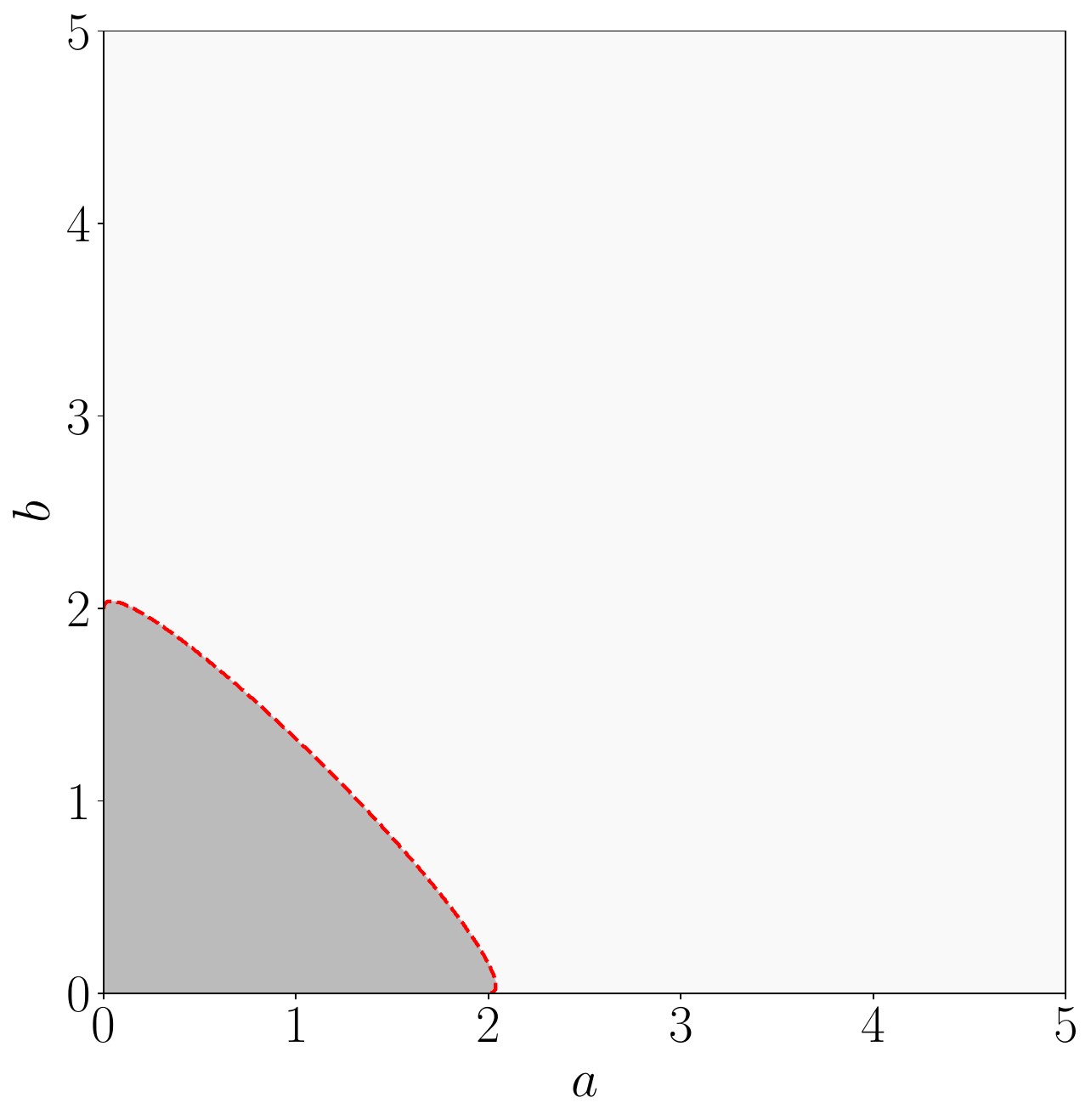}
		\label{fig:phase_tran_clus_d_SDP}	
	}
	\caption{An illustration of the information-theoretic limits for exact recovery of the communities. We plot the phase transition threshold~\eqref{eq:phase_transition_threshold} obtained from Theorems~\ref{the:info_upper_bound_1} and \ref{the:cond_lower_clus} with different choices of $M$ in red dash curves. The gray area represents the region where the MLE~\eqref{eq:MLE_rewrite} fails to recover the communities with a certain probability bounded away from zero, indicating the exact recovery is impossible, and the remaining area represents the region where \eqref{eq:MLE_rewrite} achieves exact cluster recovery with high probability.}
	\label{fig:threshold_clus}
\end{figure*}

\if\arXivver1
\vspace{0.2cm}
\paragraph{Phase transition}
\else
\vspace{0.2cm}
\subsubsection{Phase transition}
\fi
Theorems~\ref{the:info_upper_bound_1} and \ref{the:cond_lower_clus} reveal the presence of a phase transition phenomenon for the exact recovery of communities. That is, as the network size $n$ grows, both the upper bound  \eqref{eq:cond_upper_clus} and the lower bound \eqref{eq:cond_lower_clus} converge to the same threshold given by
\begin{equation}
	\frac{a + b}{2} - \sqrt{\frac{ab}{M}} = 1,
	\label{eq:phase_transition_threshold}
\end{equation}
where exact recovery is possible above and impossible below, indicating the optimality of our analysis. 

To provide further insights, we plot the threshold \eqref{eq:phase_transition_threshold} in Figure~\ref{fig:threshold_clus} for different choice of $M$. As a special case when $M = 1$, our model simplifies to the SBM without synchronization needed, and the threshold becomes 
\begin{equation}
	(\sqrt{a} - \sqrt{b})^2 = 2,
	\label{eq:phase_transition_threshold_SBM}
\end{equation}
which agrees with the existing information-theoretic limit on the SBM proved in \cite{abbe2015exact, mossel2014consistency}.  Moreover, when $M > 1$, the threshold~\eqref{eq:phase_transition_threshold} exhibits a distinct behavior compared to \eqref{eq:phase_transition_threshold_SBM}, where exact recovery is possible as long as either $a$ or $b$, that indicates the density of in-cluster or across clusters connections, is sufficiently large. To understand this, recall the MLE~\eqref{eq:MLE_rewrite} explicitly checks the consistency $g_{ij} = g_ig_j^{-1}$ on each in-cluster edge and then rule out unsatisfied hypotheses. As a result, even edges that span different communities and possess ``noisy'' group transformations are helpful to recovery by serving as negative labels. On the contrary, such edges becomes irrelevant and provide no information in the absence of synchronization when $M = 1$.

In addition, one can observe that as $M$ increases, the gray area in Figure~\ref{fig:threshold_clus} that represents the impossible region for recovery shrinks rapidly. This implies that the consistency check $g_{ij} = g_ig_j^{-1}$ is progressively effective for recovering the communities as $M$ grows. To interpret this fact, recall the noisy group transformation across communities is uniformly drawn from $\mathcal{G}_M$. As a result, the probability for a noisy observation to satisfy the consistency is exactly $M^{-1}$, indicating the increasing difficulty for a wrong hypothesis to pass the check as $M$ is large. Moreover, the threshold~\eqref{eq:phase_transition_threshold} would eventually converge to 
\begin{equation}
	\frac{a + b}{2} = 1
\label{eq:threshold_when_M_is_zero}
\end{equation}
when $M$ is large\footnote{Notice that in our analysis, we only consider the finite groups and thus $M$ must be a constant as opposed to $n$ that grows to infinity. Therefore, it is somewhat not rigorous to let $M \rightarrow \infty$ and achieve the threshold \eqref{eq:threshold_when_M_is_zero}. Here, we present this merely for the purpose of a better interpretation of our theory.}. Then, as a basic result derived from \cite{erdHos1960evolution}, \eqref{eq:threshold_when_M_is_zero} corresponds to the threshold for a random graph being connected with high probability, where the graph is generated from the SBM with parameters $p = \frac{a\log n}{n}$ and $q = \frac{b\log n}{n}$. As a result, we conjecture that in the extreme case when $M \rightarrow \infty$, Theorem~\ref{the:info_upper_bound_1} still holds and thus the exact recovery of clusters becomes possible as long as the whole network is connected.

\if\arXivver1
\paragraph{Proof strategy}
\else
\vspace{0.1cm}
\subsubsection{Proof strategy}
\fi
The proofs of Theorems~\ref{the:info_upper_bound_1} and \ref{the:cond_lower_clus} are deferred to Section~\ref{sec:proof_clus_upper_bound} and Section~\ref{sec:proof_clus_lower_bound}, respectively. 
Notably, our proof technique for the upper bound significantly differs from those traditional approaches employed in community detection or group synchronization~(e.g.~\cite[Theorem 2]{abbe2015exact} for SBM and \cite[Theorem 1]{chen2016information} for synchronization) in the way of handling different hypotheses. To be concrete, since the MLE fails when there exists a \textit{wrong} hypothesis $\bm{x} = \{\bm{\kappa}, \bm{g}\}$ with a higher likelihood than the ground truth $\bm{x}^* = \{\bm{\kappa}^*, \bm{g}^*\}$, i.e.,
\begin{equation*}
	\mathcal{P}(\bm{y}|\bm{x}) > \mathcal{P}(\bm{y} | \bm{x}^*) \quad \text{and} \quad \textrm{dist}_{\text{c}}(\bm{\kappa}, \bm{\kappa}^*) > 0, 
\end{equation*}
a common path towards an upper bound of $\mathcal{P}_{e, \mathrm{c}}(\psi_{\textup{MLE}})$ is by first applying the union bound over all wrong hypothesis then conducting a sharp analysis on each one individually. That is 
\begin{align}
	\mathcal{P}_{e, \mathrm{c}}(\psi_{\textup{MLE}}) &= \mathcal{P}\left(\exists \bm{x}: \frac{\mathcal{P}(\bm{y}|\bm{x})}{\mathcal{P}(\bm{y}|\bm{x}^*)} \geq 1, \; \textrm{dist}_{\text{c}}(\bm{\kappa}, \bm{\kappa}^*) > 0\right) \label{eq:upper_bound_P_union_new}\\
	&\leq \sum_{\bm{x}: \textrm{dist}_{\text{c}}(\bm{\kappa}, \bm{\kappa}^*) > 0} \mathcal{P}\left(\frac{\mathcal{P}(\bm{y}|\bm{x})}{\mathcal{P}(\bm{y}|\bm{x}^*)} \geq 1\right).
	\label{eq:upper_bound_P_union}
\end{align}
This proof strategy offers a simple but sharp analysis in the context of community detection or synchronization e.g.~\cite{abbe2015exact, chen2016information}, where cluster memberships or group elements are considered separately.
However, in our situation when these two components are coupled, \eqref{eq:upper_bound_P_union} fails to provide a tight result due to the exponentially large space of the wrong hypothesis. To illustrate this, consider that for any hypothesis of the cluster assignment $\bm{\kappa}$, there are $M^n$  combinations of the associated group elements $\bm{g}$. Therefore, \eqref{eq:upper_bound_P_union} explodes exponentially as $n \rightarrow \infty$. To address this issue, we instead analyze the original probability \eqref{eq:upper_bound_P_union_new} that involves different $\bm{g}$ as a whole, leading to a substantially more complicated proof scheme. 

In addition, our proof relies on several sharp analyses related to random graph theory~(see e.g.,~Theorem~\ref{the:bound_Z_n_3} in Appendix~\ref{sec:erdos_renyi_intro}), which serves as one of the main technical challenges. We leave some discussion in Section~\ref{sec:discuss_er_graph}. 

{\color{black}
\vspace{0.2cm}
\paragraph{An intuitive justification for Theorem~\ref{the:info_upper_bound_1}} 
Here, we provide an intuitive justification for the upper bound (Theorem~\ref{the:info_upper_bound_1}). Notably, the analysis here is intentionally informal as it involves a few heuristic assumptions and approximations. To begin with, consider the regime where $Mp(1-q) > q(1-p)$. In this case, the MLE~\eqref{eq:MLE_rewrite} maximizes inner-cluster edges while maintaining group transformation consistency. For a specific node $i$, let us define $N^{(i)}_{\text{inner}} \sim \text{Binom}\left(\frac{n}{2}-1, p\right)$ and $N^{(i)}_{\text{inter}} \sim \text{Binom}\left(\frac{n}{2}, q\right)$ as the number of inner and inter-cluster connections with the node $i$, respectively. Then, the MLE~\eqref{eq:MLE_rewrite} fails if there exists a node $i$, such that
\begin{equation*}
\underbrace{N^{(i)}_{\text{inter}} \cdot \mathbbm{1}\left(g_{ij} = g_ig^{-1}_j, \;\forall j \text{ connected to } i, \text{ and } \kappa_j \neq \kappa_i \right)
}_{ =: \tilde{N}^{(i)}_{\text{inter}}} \geq N^{(i)}_{\text{inner}}, 
\end{equation*}
for some group elements $\{g_j\}_{j = 1}^n$. 
For simplicity, we have ignored the enforcement of cluster sizes (otherwise we need to find a pair of such nodes). To proceed, we estimate the probability mass function~(PMF) of $\tilde{N}^{(i)}_{\text{inter}}$ as, for $k \geq 1$,
\begin{align}
\mathcal{P}(\tilde{N}^{(i)}_{\text{inter}} = k) &= \mathcal{P}\left(g_{ij} = g_ig^{-1}_j, \;\forall j \text{ connected to } i, \text{ and } \kappa_j \neq \kappa_i \mid N^{(i)}_{\text{inter}} = k \right) \cdot \mathcal{P}\left(N^{(i)}_{\text{inter}} = k\right) \nonumber \\ 
&\overset{(a)}{\approx} \frac{1}{M^k} \cdot {n/2 \choose k} q^k(1-q)^{\frac{n}{2} - k} \approx \frac{e^{-\frac{nq}{2}}}{k!} \left(\frac{nq}{2M}\right)^k.
\label{eq:ninter_tilde}
\end{align}
In the last step of~\eqref{eq:ninter_tilde}, we apply the Poisson approximation that $\text{Binom}\left(\frac{n}{2}, q\right) \approx \text{Poisson}\left(\frac{nq}{2}\right)$. Notably, in step $(a)$ we approximate that the probability for $k$ inter-cluster edges to satisfy the group transformation consistency is $1 / M^k$, which aligns with our model setting where each edge is drawn uniformly at random from $M$ group elements. However, we emphasize that this assumption requires careful justification that is provided in Section~\ref{sec:prood_upper_bound_clus_cycle_analysis}, which plays a critical role in our proof. Given the PMF~\eqref{eq:ninter_tilde}, the moment generating function~(MGF) of $\tilde{N}^{(i)}_{\text{inter}}$ can be approximated as 
\begin{align*}
\mathbb{E}\left[e^{t\tilde{N}^{(i)}_{\text{inter}}}\right] &= \underbrace{1 - \sum_{k = 1}^\infty \mathcal{P}(\tilde{N}^{(i)}_{\text{inter}} = k)}_{\mathcal{P}(\tilde{N}^{(i)}_{\text{inter}} = 0)}  + \sum_{k = 1}^\infty \mathcal{P}(\tilde{N}^{(i)}_{\text{inter}} = k) e^{tk} \approx 1 - \sum_{k = 1}^\infty \frac{e^{-\frac{nq}{2}}}{k!} \left(\frac{nq}{2M}\right)^k \cdot (1 - e^{tk}) \\
&= 1 - \exp\left(-\left(1 - \frac{1}{M}\right)\frac{nq}{2}\right) + \exp\left(-\left(1 - \frac{e^{t}}{M}\right)\frac{nq}{2}\right) .
\end{align*}
Applying the Chernoff bound, we have, for any $t > 0$
\begin{equation}
\mathcal{P}\left(\tilde{N}^{(i)}_{\text{inter}} \geq N^{(i)}_{\text{inner}}\right) \leq \mathbb{E}\left[e^{t(\tilde{N}^{(i)}_{\text{inter}} -N^{(i)}_{\text{inner}})}\right].
\label{eq:prob_n_i_n_o}
\end{equation}
As $\tilde{N}^{(i)}_{\text{inter}}$ and $N^{(i)}_{\text{inner}}$ are independent, we can further bound the RHS of \eqref{eq:prob_n_i_n_o} as
\begin{align}
\log \left(\mathbb{E}\left[e^{t( \tilde{N}^{(i)}_{\text{inter}}-N^{(i)}_{\text{inner}} )}\right]\right) &{=} \log\left(\mathbb{E}\left[e^{t\tilde{N}^{(i)}_{\text{inter}}}\right]\right) +  \log\left(\mathbb{E}\left[e^{-tN^{(i)}_{\text{inner}}}\right]\right) \nonumber \\
&\overset{(a)}{\approx} \log\left(1 - \exp\left(-\left(1 - \frac{1}{M}\right)\frac{nq}{2}\right) + \exp\left(-\left(1 - \frac{e^{t}}{M}\right)\frac{nq}{2}\right)\right) + \frac{np}{2}\left(e^{-t} - 1\right) \nonumber\\
&\overset{(b)}{\approx} \frac{n}{2}\left(q\left(\frac{e^{t}}{M} - 1\right) + p\left(e^{-t} -1\right)\right), \label{eq:exp_n_i_n_o}
\end{align}
where $(a)$ applies the Poisson approximation on $N^{(i)}_{\text{inner}}$ similar to $\tilde{N}^{(i)}_{\text{inter}}$ in \eqref{eq:ninter_tilde}, and $(b)$ holds due to $e^{-t} < M$ for $t > 0$, which will indeed be the case when optimizing with regard to $t$ in \eqref{eq:optimize_t_proof_sketch},
followed by using the approximations $e^x \approx 1 + x$ and $\log(1+x) \approx x$. Now, as $t$ can be chosen arbitrary, by taking 
\begin{equation}
t = \frac{1}{2}\log\left(\frac{Mp}{q}\right) = \frac{1}{2}\log\left(\frac{Ma}{b}\right),
\label{eq:optimize_t_proof_sketch}
\end{equation}
to minimize the RHS of \eqref{eq:exp_n_i_n_o}, we obtain
\begin{equation*}
    \log \left(\mathbb{E}\left[e^{t(\tilde{N}^{(i)}_{\text{inter}}-N^{(i)}_{\text{inner}} )}\right]\right) \leq - \left(\frac{a + b}{2} - \sqrt{\frac{ab}{M}}\right) \log n.
\end{equation*}
By plugging this into \eqref{eq:prob_n_i_n_o}, we have
\begin{equation*}
    \mathcal{P}\left(\tilde{N}^{(i)}_{\text{inter}}> N^{(i)}_{\text{inner}}\right) \leq \exp\left(-\left(\frac{a + b}{2} - \sqrt{\frac{ab}{M}}\right)\log n\right).
\end{equation*}
Applying the union bound over all the $n$ nodes yields
\begin{equation*}
    \mathcal{P}\left(\text{MLE~\eqref{eq:MLE_rewrite} fails}\right) \leq \sum_{i = 1}^n \mathcal{P}\left(\tilde{N}^{(i)}_{\text{inter}} >N^{(i)}_{\text{inner}}\right) \leq \exp\left(-\left(\frac{a + b}{2} - \sqrt{\frac{ab}{M}} -1\right)\log n\right).
\end{equation*}
Therefore, the failure probability goes to $0$ when $\frac{a + b}{2} - \sqrt{\frac{ab}{M}} > 1$. 

We remark that the above analysis is based on several approximations and assumptions. For instance, we do not enforce the cluster size constraint, and the consistency analysis is entirely omitted in~\eqref{eq:ninter_tilde}. Nevertheless, this simplified argument yields the same parameter threshold as the more rigorous analysis presented in Section~\ref{sec:proof_clus_upper_bound}. As such, it serves as a warm-up proof and provides intuitive justification for the theorem.
}

\subsection{Exact recovery of group elements}
\begin{theorem}[Upper bound]
Under the same setting of Theorem~\ref{the:info_upper_bound_1}. For any $c > 0$, if  
\begin{equation}
	\frac{a + b}{2} - (1 + o(1))\sqrt{\frac{ab}{M}} > 1 + c \quad \textrm{and} \quad \frac{a}{2} > 1 + c 
	\label{eq:cond_upper_rot}
\end{equation}
then the MLE~\eqref{eq:MLE_rewrite} exactly recovers the group elements w.p.~at least $1 - n^{-c + o(1)}$, i.e.,$\mathcal{P}_{e, \mathrm{g}}(\psi_{\textup{MLE}})  \leq n^{-c + o(1)}$.
\label{the:cond_upper_rot}
\end{theorem}

\begin{theorem}[Lower bound]
Under the same setting of Theorem~\ref{the:info_upper_bound_1}. If 
\begin{equation}
	\frac{a + b}{2} - \sqrt{\frac{ab}{M}} < 1 \quad \textrm{or} \quad \frac{a}{2} < 1  
	\label{eq:cond_lower_rot}
\end{equation}
then the MLE~\eqref{eq:MLE_rewrite} fails to {\color{black} exactly} recover the group elements w.p.~at least $0.45$, i.e.,$\mathcal{P}_{e, \mathrm{g}}(\psi_{\textup{MLE}}) \geq 0.45$.
\label{the:cond_lower_rot}
\end{theorem}

The proofs of Theorems~\ref{the:cond_upper_rot} and \ref{the:cond_lower_rot} are deferred to Section~\ref{sec:proof_cond_upper_rot} and Section~\ref{sec:proof_cond_lower_rot} respectively, which are simply extended from the ones of Theorems~\ref{the:info_upper_bound_1} and \ref{the:cond_lower_clus}. Similar to the threshold~\eqref{eq:phase_transition_threshold}, Theorems~\ref{the:cond_upper_rot} and \ref{the:cond_lower_rot} suggest a phase transition exists for the exact recovery of group elements, given as 
\begin{equation}
    \frac{a + b}{2} - \sqrt{\frac{ab}{M}} = 1 \quad \text{and} \quad a = 2.
    \label{eq:phase_transition_group_element}
\end{equation}
Here, one can see that \eqref{eq:phase_transition_group_element} is an extension of \eqref{eq:phase_transition_threshold} with the additional condition $a = 2$, implying that the exact recovery of communities is a \textit{prerequisite} for recovering the group elements. In other words, one cannot hope to recover the group elements before recovering the communities at first.  As a result, \eqref{eq:phase_transition_group_element} can be also viewed as the threshold for recovering both communities and group elements. 

Besides, the condition $a = 2$ can be interpreted from a random graph perspective: recall the nodes within each community of size $\frac{n}{2}$ are connected with the same probability $p = \frac{a\log n}{n}$, essentially following the Erdős–Rényi model. Then, it is well known that $a = 2$ represents the sharp threshold for the community being connected. Therefore, this implies that the exact recovery of group elements is possible only if each community is connected, which makes sense as one cannot hope the synchronize successfully over disconnected components due to the lack of information. 

In summary, the recovery of group elements is possible only if (1) the two communities are recovered correctly and (2) each community is connected.

\section{Discussions}
\label{sec:discussion}

\if\arXivver1
\subsection{Existing approaches and their limits}
\else
\subsection{Algorithm design}
\fi

From an algorithm perspective, since directly solving the MLE~\eqref{eq:MLE_rewrite} is non-convex and computationally intractable, finding an efficient algorithm with a competitive performance becomes necessary in practice. 
In the context of community detection, efficient approaches based on semidefinite programming~(SDP)~\cite{abbe2015exact, hajek2016achievinga, hajek2016achievingb, amini2018semidefinite, perry2017semidefinite, li2018convex} and spectral method~\cite{abbe2020entrywise, yun2014accurate, lei2015consistency} have been proposed. In particular, they are known to achieve the information-theoretic limit~\eqref{eq:phase_transition_threshold_SBM} for exact recovery on the SBM~\cite{abbe2015exact, abbe2020entrywise, hajek2016achievinga, perry2017semidefinite}. In light of this, similar algorithms~\cite{fan2021joint, fan2021spectral, chen2021non, wang2023multi} have been applied to the joint problem, and it is natural to expect that these methods also work well all the way down to the information-theoretic limit derived in this paper. 

\begin{figure}[t!]
	\centering
	\includegraphics[width = 0.33\textwidth]{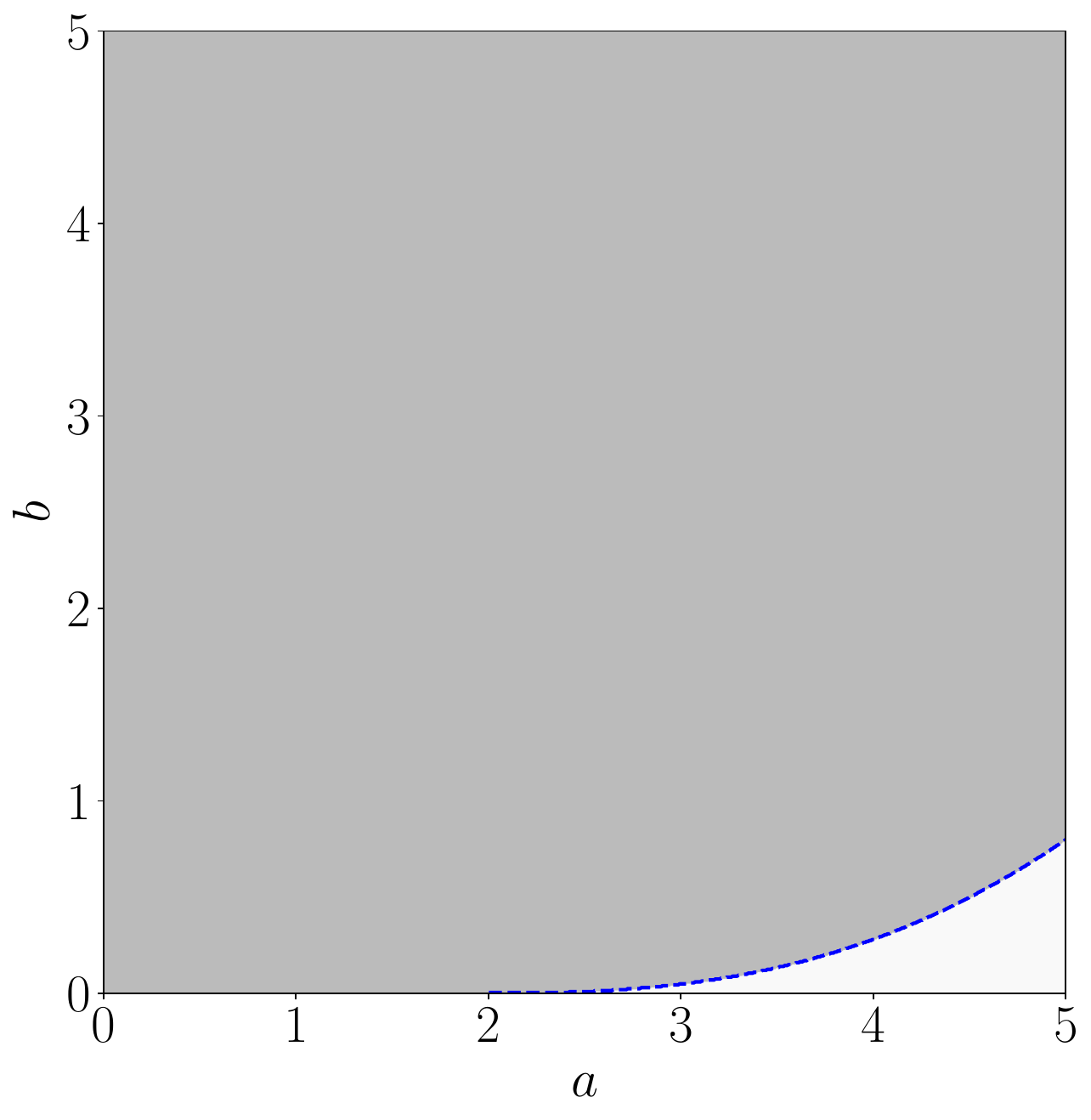}
	\caption{The sharp threshold in blue dash curve for exactly recovering both communities and group elements, by the SDP proposed in~\cite{fan2021joint}. The gray area represents the region where the SDP fails in the exact recovery with high probability.}
	\label{fig:phase_transition_SDP}
\end{figure}

Unfortunately, our result suggests that none of these methods can reach the information-theoretic limit, indicating that there is a performance gap compared to MLE~\eqref{eq:MLE_rewrite}. For instance, the SDP developed in \cite{fan2021joint} for solving the joint problem exhibits a sharp threshold for the exact recovery of both communities and group elements, which is given as
\begin{equation}
	a - \sqrt{2b}\log\left(\frac{ea}{\sqrt{2b}}\right) = 2, 	\label{eq:phase_transition_SDP}
\end{equation}
as shown in Figure~\ref{fig:phase_transition_SDP}. Notably, the statistical model studied in \cite{fan2021joint} is slightly different to ours, where \cite{fan2021joint} focuses on the compact Lie group $\text{SO}(2)$, as opposed to the finite group considered in this work. Therefore, approximately, one may compare the threshold~\eqref{eq:phase_transition_SDP} by the SDP to the information-theoretic limit~\eqref{eq:phase_transition_group_element} with a large $M$. As a result, the performance gap is clearly demonstrated by comparing Figure~\ref{fig:phase_transition_SDP} with Figure~\ref{fig:threshold_clus}. 

{\color{black}
For spectral method-based approaches~\cite{fan2021spectral, chen2021non, wang2023multi, bajaj2018smac}, it is not surprising to observe similar performance gaps exist as well. For example, \cite{chen2021non} proposes using generalized power method and provides the corresponding condition for exact recovery~\cite[Theorem 2]{chen2021non}. That is, under the setting of compact Lie group $\text{SO}(2)$ and two equal-sized clusters, the sufficient condition for exact recovery is:
\begin{equation}
    \sqrt{2b} < a \quad \text{and} \quad  a - \sqrt{2b}\log\left(\frac{ea}{\sqrt{2b}}\right) > 2, 
\end{equation}
which is almost identical to \eqref{eq:phase_transition_SDP} by SDP with the extra requirement $\sqrt{2b} < a$, indicating a performance gap similar to the SDP~\cite{fan2021joint} exists. In a same spirit of methodology, \cite{fan2021spectral} employs QR-decomposition for recovery and develops the condition for exact recovery~\cite[Theorem 5.2]{fan2021spectral} under the same setting (two equal-sized clusters and compact Lie group $\text{SO}(2)$) as\footnote{
The original condition provided in \cite[Theorem 5.2]{fan2021spectral} is written as
\begin{equation*}
    \frac{\sqrt{\left(p(1-p) + q\right)\log(2n)}}{p\sqrt{n}} \leq c_0
\end{equation*}
and \eqref{eq:condition_fan_spectral} is obtained by plugging the scaling $p = \frac{a \log n}{n}$ and $q = \frac{b\log n}{n}$.}
\begin{equation}
    \frac{\sqrt{a + b}}{a} \leq c_0, 
    \label{eq:condition_fan_spectral}
\end{equation}
for some unknown constant $c_0 > 0$.
While the performance gap may not be explicitly evident from \eqref{eq:condition_fan_spectral} due to the unknown value of $c_0$, it is clearly illustrated in the experimental results shown in \cite[Figure 2]{fan2021spectral}. Besides, the performance of another spectral method-based solution \cite{bajaj2018smac} is examined in \cite[Figure 5]{fan2021spectral}, revealing a phase transition threshold similar to~\cite{fan2021spectral}. Lastly, an improved approach based on~\cite{fan2021spectral} by aggregating multiple irreducible representations of group transformations is presented in~\cite{wang2023multi}. Despite this refinement, the method exhibits a similar performance trend (see e.g.~\cite[Figure 3]{wang2023multi}) as~\cite{fan2021spectral}, and thus, the performance gap for algorithms with polynomial time complexity persists.}

{\color{black}
\subsection{Understanding the performance gap}
\label{sec:discuss_understand_gap}
To further understand the performance gap, note that all the aforementioned convex relaxation-based approaches exhibits a common pattern: recovery becomes more challenging as $q$ (the probability of connection across communities) or $b$ increases. This behavior contrasts with the information-theoretic limits displayed in Figure~\ref{fig:threshold_clus}, where recovery can actually become easier as $q$ increases. To interpret this, as discussed in Section~\ref{sec:exact_recovery_clus}, inter-cluster edges in community detection typically carry no useful information but act purely as noise. However, in the joint problem setting, where additional group transformations are observed, inter-cluster edges can aid recovery by helping to eliminate incorrect hypothesis that violate group transformation consistency. Therefore, the performance gap arises from the inability of these methods to fully exploit the potential advantages offered by such ``noisy'' edges.}

In addition, from the perspective of problem formulation, it is evident that all existing approaches deviate from the original MLE~\eqref{eq:MLE_rewrite}. Instead, they are designed to solve alternative optimization problems that are more amenable to (convex) relaxation techniques. As an example, the authors in \cite{fan2021joint, fan2021spectral} consider the orthogonal group $\textrm{O}(d)$ and design the following program for recovery:
\begin{equation}
	\max_{\{\kappa_i\}_{i = 1}^n,\; \{\bm{O}_i\}_{i = 1}^n} \; \sum_{(i,j) \in \mathcal{E}, \; \kappa_i = \kappa_j} \left\langle \bm{O}_{ij}, \; \bm{O}_i \bm{O}_j^\top\right\rangle.
	\label{eq:program_SDP}
\end{equation}
Here, $\bm{O}_i$ and $\bm{O}_{ij} \in \mathbb{R}^{d \times d}$ are orthogonal matrices that represent $g_i$ and $g_{ij}$, respectively. As a result, \eqref{eq:program_SDP} attempts to recover by maximizing the total consistency between $\bm{O}_{ij}$ and $\bm{O}_i\bm{O}_j^\top$ for all in-cluster edges. Obviously, a major difference between $\eqref{eq:program_SDP}$ and the corresponding MLE~\eqref{eq:MLE_infinite} is that the strict constraint $g_{ij} = g_ig_j^{-1}$ or equivalently $\bm{O}_{ij} = \bm{O}_i \bm{O}_j^\top$ given in \eqref{eq:MLE_infinite} has been weakened\footnote{Equivalently, the constraint $g_{ij} = g_ig_j^{-1}$ can be transferred to the objective function in \eqref{eq:MLE_infinite} and becomes a binary-valued penalty $f(g_i, g_j)$ such that $f(g_i, g_j) = \begin{cases}
    -\infty &\quad g_{ij}, \neq g_ig_j^{-1} \\
    0, &\quad \text{otherwise}
\end{cases}$. } to maximizing the inner product $\langle \bm{O}_{ij}, \; \bm{O}_i \bm{O}_j^\top\rangle$, which instead tolerates a certain amount of error between $\bm{O}_{ij}$ and $\bm{O}_i \bm{O}_j^\top$. As a result, the relaxed program~\eqref{eq:program_SDP} may not fully leverage the consistency requirement, resulting in the performance gap.

{\color{black}
To bridge the performance gap between existing algorithms and the information-theoretic limits, our earlier discussion suggests that a more refined approach could be developed by carefully verifying cycle consistency -- similar to our proof technique used in Section~\ref{sec:prood_upper_bound_clus_cycle_analysis} for Theorem~\ref{the:info_upper_bound_1}. However, such an algorithm would likely entail significantly higher computational cost compared to existing convex relaxation-based methods. For instance, similar idea of checking cycle consistency has been recently applied to group synchronization in~\cite{lerman2021robust}, where the proposed algorithm is of complexity of $O(Ln^L)$ (see \cite[Section 4.2.7]{lerman2021robust}), with $L$ denoting the cycle length being checked. To manage the computational cost, it is desirable to use the shortest cycle in the graph (e.g. $L = 3$). However, this becomes problematic in sparse graphs (which is often the case in practice) where nodes typically participate only in long cycles with $L \gg 3$. For example, in an Erd\H{o}s–R{\'e}nyi graph $G(n, p)$ with $p = \Theta(\frac{\log n}{n})$, one can show that for any node, the expected number of cycles with length $L$ passing through it is\footnote{Here, ${n-1 \choose L-1}$ determines the set of nodes involved in the cycle, $\frac{(L-1)!}{2}$ stands for the order of nodes traversing along the cycle, and $p^L$ ensures the cycle are indeed connected.} 
$${n-1 \choose L-1} \cdot \frac{(L-1)!}{2} \cdot p^L = O\left(\frac{(\log n)^L}{n}\right).$$
Therefore, a typical cycle length would be $L = \omega \big(\frac{\log n}{\log \log n }\big)$, indicating the scarcity of short cycles with $L = \Theta(1)$ in the sparse regime.
As a result, the necessity of checking long cycles makes the algorithm inefficient and less favorable comparing to the existing convex relaxation based approaches.

Moreover, the algorithm that checks cycle-consistency relies on the assumption that the measurements $g_{ij}$ are perfect for intra-cluster connections, whereas in practice these are typically noisy. As a result, this fact further limits the practicality of cycle-consistency checks. In contrast, all the existing methods tolerate moderate noise as long as it stays below a certain threshold. Given the above, designing an efficient and robust algorithm that can handle both sparse structures and noisy measurements remains an open and challenging problem.
}

\if\arXivver1
\subsection{Giant component of Erdős–Rényi graphs}
\else
\subsection{Giant component of Erdős–Rényi graphs}
\fi
\label{sec:discuss_er_graph}
An important technical ingredient of our theory is analyzing the giant component of Erdős–Rényi graphs~\cite{erdHos1960evolution}. As aforementioned, it is well-known that an Erdős–Rényi graph $G(n, p)$ of size $n$ and connection probability $p = \frac{a\log n}{n}$ for some constant $a$ is connected with high probability if and only if $a > 1$. Furthermore, when $a < 1$ such as the graph is disconnected, a giant component of size close to $n$ is expected to exist with high probability. 
\textcolor{black}{To be specific, let $Z_n$ denote the size of the largest connected component in the graph. By following standard techniques commonly used in the study of random graphs (e.g., \cite[Theorem 2.14]{frieze2016introduction})\footnote{\cite[Theorem 2.14]{frieze2016introduction} focuses on the regime where $p = \Theta(1/n)$, whereas our interest lies in the denser regime $p = \Theta(\log n / n)$. Therefore, we cannot directly apply \cite[Theorem 2.14]{frieze2016introduction}, but we adopt the same proof techniques to derive the bound in \eqref{eq:Z_n_bound_3_main}. In particular, we first show that there is no component of intermediate size (i.e., between $\beta_1 \log n$ and $\beta_2 n$ for some constants $\beta_1, \beta_2 > 0$), and then we upper bound the total number of vertices in small components (those of size less than $\beta_1 \log n$.)}, we can derive the following result regarding $Z_n$,}
\begin{equation}
	\mathcal{P}\left(Z_n > \left(1 -  n^{-a + \epsilon + o(1)}\right)n\right) > 1 - n^{-\epsilon}, 
	\label{eq:Z_n_bound_3_main}
\end{equation}
which holds for any $\epsilon < a$. However, \eqref{eq:Z_n_bound_3_main} is too loose to be directly applied in our analysis in Section~\ref{sec:proof_clus_upper_bound}. To address this, we improve upon \eqref{eq:Z_n_bound_3_main} by establishing the following sharper result:
\begin{equation}
    \mathcal{P}\left(Z_n \geq n - \Theta\left(\frac{n}{\log\log n}\right)\right) \geq 1 - n^{-\Theta(\log n)},
    \label{eq:large_component_main}
\end{equation} 
which is essential for deriving the sharp threshold in~\eqref{eq:phase_transition_threshold}; see Theorem~\ref{the:bound_Z_n_3} in Appendix~\ref{sec:erdos_renyi_intro} for details. Notably, the bound in \eqref{eq:large_component_main} cannot be obtained via conventional moment methods such as Markov's or Chebyshev's inequalities, which are commonly used in random graph theory (e.g., see~\cite{bandeira2018random}). Instead, we establish \eqref{eq:large_component_main} through a delicate combinatorial analysis, which constitutes one of the main technical challenges of our proof.

\if\arXivver1
\subsection{Future works}
\else
\subsection{Future works}
\fi
In summary, our analysis presented in this work can be generalized in several aspects, listed as follows:
\begin{enumerate}
\vspace{0.2cm}
\item Our theory given in Section~\ref{sec:main_results} is based on the assumption that $\mathcal{G}_M$ is a finite group and $M$ stays as a constant as $n$ grows, extending the result to infinite groups e.g.~the orthogonal group $\textrm{O}(d)$ is left as a future work.
\vspace{0.2cm}
\item \textcolor{black}{Although we restrict our attention on the simplest scenario of two equal-sized communities, our proof framework can be naturally extended to more generalized cases of multiple clusters with different sizes. In such scenarios, we expect that the same proof techniques -- particularly our key technical ingredients such as the cycle consistency analysis in Section~\ref{sec:prood_upper_bound_clus_cycle_analysis} -- can still be applied, though with additional care to account for the presence of multiple communities. As $K$ increases, exact recovery is expected to become more challenging due to a higher number of inter-cluster edges.}
\vspace{0.2cm}
\item \textcolor{black}{
Our theory assumes the group transformation $g_{ij}$ is perfectly observed on connections within the same community. However, it is more common in practice where $g_{ij}$ is perturbed by noise. 
Therefore, in order to extend our analysis to account for such noise, additional care must be taken in the cycle consistency analysis (e.g., Section~\ref{sec:prood_upper_bound_clus_cycle_analysis} for the upper bound). Specifically, one should not expect exact cycle consistency, and the error in group alignment may accumulate in long cycles. The presence of inconsistency may no longer reliably indicate the existence of inter-cluster edges within the cycle.
}

\vspace{0.2cm}
\item \textcolor{black}{
Another promising direction for future work is to establish the theoretical guarantees for weak (partial) recovery, which is important in many practical scenarios where partial recovery may already suffice for downstream applications. However, the proof techniques required for weak recovery may differ significantly from those used in our current analysis for exact recovery. Some insights may be drawn from existing studies on weak recovery in community detection (e.g.,~\cite{abbe2017community, decelle2011asymptotic, mossel2012stochastic}), which show that weak recovery is achievable in the sparse regime $p = a/n$, $q = b/n$ for some constants $a, b$, provided that giant components exist within communities.
}

\vspace{0.2cm}
\item Based on our discussion in Section~\ref{sec:discuss_understand_gap}, finding an efficient algorithm that achieves the information-theoretic limit remains an open problem. Our results suggest focusing on the recovery of communities, which is a prerequisite to accurately recover group elements. Moreover, developing an efficient strategy to check the cycle consistency could be the key to significantly improving the performance of existing approaches.
\end{enumerate}

\section{Proof Preliminaries}
\label{sec:four_subsets}

\begin{figure*}[t!]
    \centering
    \includegraphics[width = 0.7\textwidth]{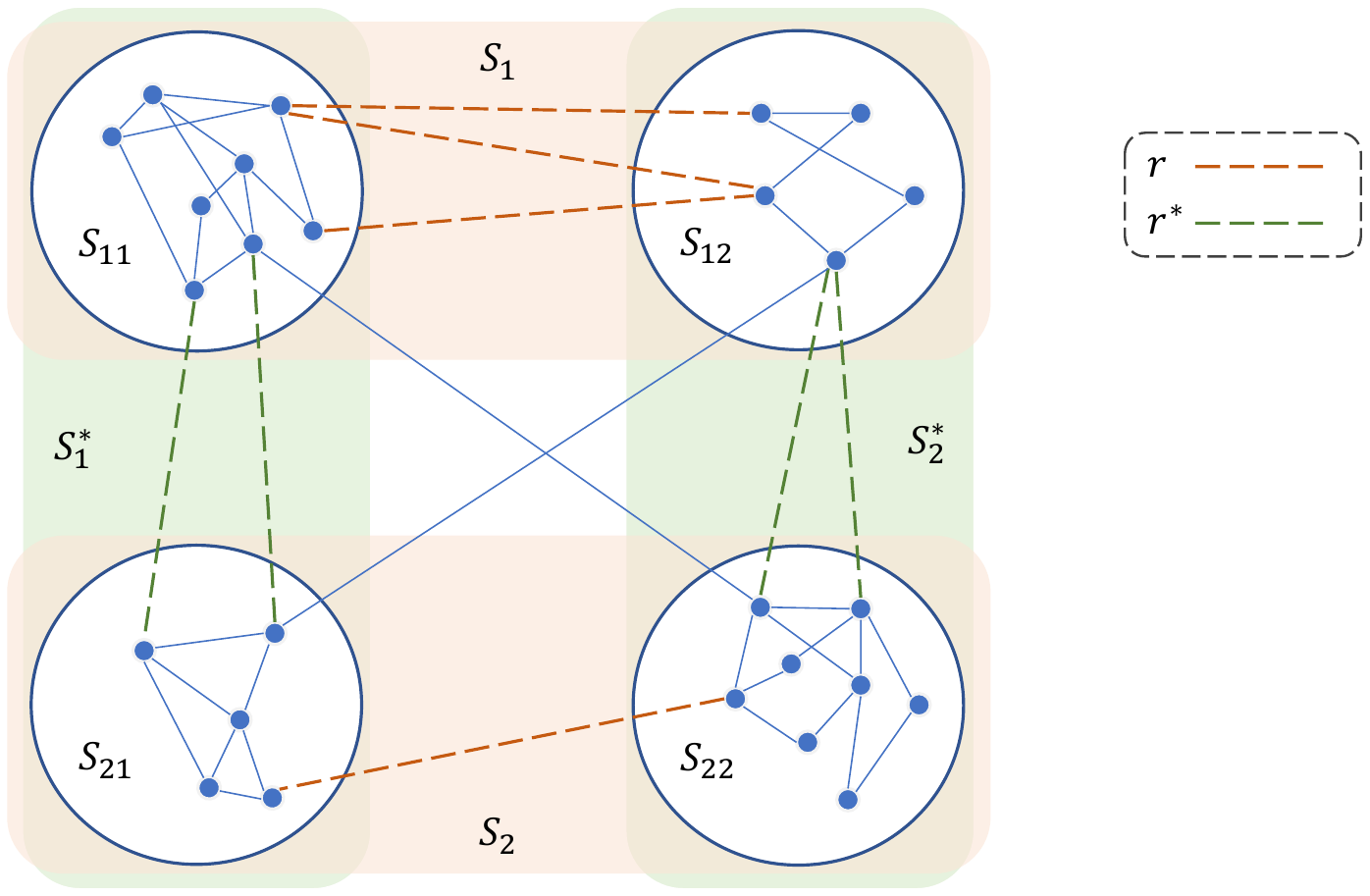}
    \vspace{-0.8cm}
    \caption{An illustration of the four subsets $S_{11}$, $S_{12}$, $S_{21}$, and $S_{22}$ defined in \eqref{eq:four_subsets}. %
    We use green and orange dash lines to indicate the set of edges corresponding to $r$ and $r^*$ defined in \eqref{eq:def_r_in_out}, respectively.}
    \label{fig:four_subsets}
\end{figure*}

We first introduce several important ingredients that support our proof.

\if\arXivver1
\subsection{Four subsets}
\else
\vspace{0.2cm}
\subsubsection{Four subsets}
\fi
Under the setting of two equal-sized communities, given the ground truth of the community assignment $\bm{\kappa}^*$ and any alternative hypothesis $\bm{\kappa}$, we define the following four subsets of nodes, which greatly facilitate our analysis and thus will be frequently used in our proof:
\begin{equation}
	\begin{aligned}
		S_{11} &:= S_1 \cap S_1^* = \left\{ i \mid \kappa_i = 1, \; \kappa_i^* = 1 \right\}, \\
		S_{12} &:= S_1 \cap S_2^* = \left\{ i \mid \kappa_i = 1, \; \kappa_i^* = 2 \right\},\\
		S_{21} &:= S_2 \cap S_1^* = \left\{ i \mid \kappa_i = 2, \; \kappa_i^* = 1 \right\},  \\
	    S_{22} &:= S_2 \cap S_2^* = \left\{ i \mid \kappa_i = 2, \; \kappa_i^* = 2 \right\}.
	\end{aligned}
	\label{eq:four_subsets}
\end{equation}
For instance, \textcolor{black}{$S_{11}$ denotes the set of nodes that are correctly identified as belonging to cluster 1, while $S_{12}$ stands for the set of nodes that are incorrectly assigned to cluster 1 but actually belongs to cluster 2.} See Figure~\ref{fig:four_subsets} for an illustration. Notably, the four subsets are mutually disjoint and their union covers all nodes in the network. Also, one can immediately observe the following equations: 
\begin{equation}
	|S_{11}| = |S_{22}|, \quad\text{and} \quad |S_{21}| = |S_{12}| = \frac{n}{2} - |S_{11}|,
\end{equation} 
which hold for any valid hypothesis $\bm{\kappa}$ \textcolor{black}{under the assumption that the two communities are of equal size}. 

\if\arXivver1
\subsection{Derivation of the MLE~\eqref{eq:MLE_rewrite}}
\else
\vspace{0.2cm}
\subsubsection{Derivation of the MLE~\eqref{eq:MLE_rewrite}}
\fi
\label{sec:proof_mle_derivation}
Our theory is built based on the MLE program~\eqref{eq:MLE_rewrite} where we start by expanding the likelihood $\mathcal{P}( \mathcal{E}, \mathfrak{G} | \bm{\kappa}, \bm{g})$ as
\begin{equation}
	\begin{aligned}
		\mathcal{P}( \mathcal{E}, \mathfrak{G} | \bm{\kappa}, \bm{g}) &= \mathcal{P}( \mathcal{E} | \bm{\kappa}, \bm{g}) \mathcal{P}( \mathfrak{G} | \bm{\kappa}, \bm{g}, \mathcal{E}) \\
		&= \mathcal{P}( \mathcal{E} | \bm{\kappa}) \mathcal{P}( \mathfrak{G} | \bm{\kappa}, \bm{g}, \mathcal{E}),
	\end{aligned}
	\label{eq:likelihood}
\end{equation}
which is the product of the likelihood of edge connections $\mathcal{P}( \mathcal{E} | \bm{\kappa})$ and the likelihood of group transformations given the edges, i.e.,~$\mathcal{P}( \mathfrak{G} | \bm{\kappa}, \bm{g}, \mathcal{E})$. 

According to the statistical model, $\mathcal{P}( \mathcal{E} | \bm{\kappa})$ is given as  
\begin{equation}
	\begin{aligned}
		\mathcal{P}( \mathcal{E} | \bm{\kappa}) 
		&= \prod_{\substack{(i,j) \in \mathcal{E} \\[2pt] \kappa_i = \kappa_j}} p \prod_{\substack{(i,j) \notin \mathcal{E}, \\[2pt] \kappa_i = \kappa_j}} (1-p) \prod_{\substack{(i,j) \in \mathcal{E}, \\[2pt] \kappa_i \neq \kappa_j}} q \prod_{\substack{(i,j) \notin \mathcal{E}, \\[2pt] \kappa_i \neq \kappa_j}} (1-q) \\
		&\propto \left(\frac{p}{1-p}\right)^{\left| \mathcal{E}_{\text{inner}}\right|}  \left(\frac{q}{1-q}\right)^{\left| \mathcal{E}_{\text{inter}}\right|}\\
		&= \left(\frac{p}{1-p}\right)^{\left| \mathcal{E}_{\text{inner}}\right|}  \left(\frac{q}{1-q}\right)^{\left|\mathcal{E}\right| - \left| \mathcal{E}_{\text{inner}}\right|}\\
		&\propto \left(\frac{p(1-q)}{q(1-p)}\right)^{|\mathcal{E}_{\text{inner}}|},
	\end{aligned}
	\label{eq:likelihood_edge}
\end{equation}
\textcolor{black}{
where ${\mathcal{E}_{\text{inter}}} := \{(i,j) \;|\; (i,j) \in \mathcal{E}, \; i < j, \; \kappa_i \neq \kappa_j \}$ denotes the set of edges across different communities}, as opposed to $\mathcal{E}_{\text{inner}}$ defined in \eqref{eq:def_E_innner} for the set of edges within the same community, and we have used the fact that $|\mathcal{E}| = |\mathcal{E}_{\text{inner}}| + |\mathcal{E}_{\text{inter}}|$. 

For $\mathcal{P}( \mathfrak{G} | \bm{g}, \bm{\kappa}, \mathcal{E})$, we have 
\begin{equation}
	\begin{aligned}
		\mathcal{P}( \mathfrak{G} | \bm{g}, \bm{\kappa}, \mathcal{E}) &= \prod_{\substack{(i,j) \in \mathcal{E} \\[2pt] \kappa_i = \kappa_j}} \mathbbm{1}\left(g_{ij} = g_ig_j^{-1}\right) \prod_{\substack{(i,j) \in \mathcal{E} \\[2pt] \kappa_i \neq \kappa_j}} M^{-1} \\
		&\propto \mathbbm{1}\left(g_{ij} = g_ig_j^{-1}; \; \forall (i,j) \in \mathcal{E}_{\text{inner}}\right) \cdot M^{-|\mathcal{E}_{\text{inter}}|} \\
		&\propto \mathbbm{1}\left(g_{ij} = g_ig_j^{-1}; \; \forall (i,j) \in \mathcal{E}_{\text{inner}}\right) \cdot M^{|\mathcal{E}_{\text{inner}}|},
	\end{aligned}
	\label{eq:likelihood_rotation}
\end{equation}
where $\mathbbm{1}(\cdot)$ denotes the indicator function. 

Given the above, by plugging \eqref{eq:likelihood_edge} and \eqref{eq:likelihood_rotation} into \eqref{eq:likelihood},  the MLE problem can be written as 
\if\arXivver1
\begin{equation}
	\begin{aligned}
		\psi_{\text{MLE}}(\mathcal{E}, \mathfrak{G}) &=  \argmax_{\bm{\kappa}, \; \bm{g}} \quad \mathcal{P}( \mathcal{E}, \mathfrak{G} | \bm{\kappa}, \bm{g})\\
		&= \argmax_{\bm{\kappa}, \; \bm{g}} \quad  \left(\frac{Mp(1-q)}{q(1-p)}\right)^{|\mathcal{E}_{\text{inner}}|} \cdot \mathbbm{1}\left(g_{ij} = g_ig_j^{-1}; \; \forall (i,j) \in \mathcal{E}_{\text{inner}}\right) \\[5pt]
		\textrm{s.t.} & \quad n_k = n_k^*, \quad k = 1, \ldots, K,
	\end{aligned}
	\label{eq:MLE}
\end{equation}
\else
\begin{equation}
	\begin{aligned}
		\psi_{\text{MLE}}(\mathcal{E}, \mathfrak{G}) &=  \argmax_{\bm{\kappa}, \; \bm{g}} \quad \mathcal{P}( \mathcal{E}, \mathfrak{G} | \bm{\kappa}, \bm{g})\\
		&= \argmax_{\bm{\kappa}, \; \bm{g}} \quad  \left(\frac{Mp(1-q)}{q(1-p)}\right)^{|\mathcal{E}_{\text{inner}}|} \\
		&\; \cdot \mathbbm{1}\left(g_{ij} = g_ig_j^{-1}; \; \forall (i,j) \in \mathcal{E}_{\text{inner}}\right) \\[5pt]
		\textrm{s.t.} & \quad n_k = n_k^*, \quad k = 1, \ldots, K
	\end{aligned}
	\label{eq:MLE}
\end{equation}
\fi
where we apply the assumption that the size of each community is provided, \textcolor{black}{ and note that \eqref{eq:MLE} applies to arbitrary community settings}. At first glance, solving \eqref{eq:MLE} requires the prior knowledge of the model parameters $p, q$, and $M$. However, by splitting \eqref{eq:MLE} into the two scenarios $Mp(1-q) > q(1-p)$ or $Mp(1-q) \leq q(1-p)$ and making the indicator function as constraints, \eqref{eq:MLE} becomes the aforementioned MLE forumlation \eqref{eq:MLE_rewrite}, which does not involve any model parameter specifically. 

\if\arXivver1
\subsection{Likelihood ratio}
\else
\vspace{0.2cm}
\subsubsection{Likelihood ratio}
\fi
Given the true parameter $\bm{x}^* = \{\bm{\kappa}^*, \bm{g}^*\}$ and a hypothesis $\bm{x} = \{\bm{\kappa}, \bm{g}\}$, our proof is based on analyzing the likelihood ratio defined as 
\begin{equation*}
	\frac{\mathcal{P}(\bm{y}|\bm{x})}{\mathcal{P}(\bm{y}|\bm{x}^*)} = \frac{\mathcal{P}(\mathcal{E}|\bm{\kappa})}{\mathcal{P}(\mathcal{E}|\bm{\kappa}^*)} \cdot \frac{\mathcal{P}( \mathfrak{G} | \bm{\kappa}, \bm{g}, \mathcal{E})}{\mathcal{P}( \mathfrak{G} | \bm{\kappa}^*, \bm{g}^*, \mathcal{E})}.
\end{equation*}
Then, by using the likelihood derived in \eqref{eq:likelihood_edge}, we obtain 
\begin{equation}
	\frac{\mathcal{P}(\mathcal{E}|\bm{\kappa})}{\mathcal{P}(\mathcal{E}|\bm{\kappa}^*)} = \left(\frac{p(1-q)}{q(1-p)}\right)^{|\mathcal{E}_{\text{inner}}| - |\mathcal{E}_{\text{inner}}^*|}, 
	\label{eq:ratio_kappa}
\end{equation}
where $\mathcal{E}_{\text{inner}}^*$ denotes the set of in-cluster edges according to $\bm{\kappa}^*$, analogous to $\mathcal{E}_{\text{inner}}$ defined in \eqref{eq:def_E_innner} for the hypothesis $\bm{\kappa}$. Furthermore, \textcolor{black}{under the setting of two equal-sized clusters and recall the four subsets defined in \eqref{eq:four_subsets}, we have
\begin{equation}
\begin{aligned}
    |\mathcal{E}_{\text{inner}}| &= |\mathcal{E}(S_{11})| + |\mathcal{E}(S_{12})| + |\mathcal{E}(S_{21})|+ |\mathcal{E}(S_{22})| + |\mathcal{E}(S_{11}, S_{12})| + |\mathcal{E}(S_{21}, S_{22})|, \\
    |\mathcal{E}^*_{\text{inner}}| &= |\mathcal{E}(S_{11})| + |\mathcal{E}(S_{12})| + |\mathcal{E}(S_{21})|+ |\mathcal{E}(S_{22})| + |\mathcal{E}(S_{11}, S_{21})| + |\mathcal{E}(S_{12}, S_{22})|.
\end{aligned}
\nonumber
\end{equation}
We suggest readers resorting to Figure~\ref{fig:four_subsets} for a better understanding. Then, $|\mathcal{E}_{\text{inner}}| - |\mathcal{E}_{\text{inner}}^*|$ in \eqref{eq:ratio_kappa} can be expressed as
\if\arXivver1
\begin{equation}
	\begin{aligned}
		|\mathcal{E}_{\text{inner}}| - |\mathcal{E}_{\text{inner}}^*| &= \underbrace{|\mathcal{E}(S_{11}, S_{12})| + |\mathcal{E}(S_{21}, S_{22})|}_{=: r} - \underbrace{\left(|\mathcal{E}(S_{11}, S_{21})| + |\mathcal{E}(S_{12}, S_{22})|\right)}_{=: r^*} \\
		&=: r - r^*,
	\end{aligned}
\label{eq:E_inner_minus_start}
\end{equation}
\else
\begin{equation}
	\begin{aligned}
		|\mathcal{E}_{\text{inner}}| - |\mathcal{E}_{\text{inner}}^*| &= |\mathcal{E}(S_{11}, S_{12})| + |\mathcal{E}(S_{21}, S_{22})| \\
		&\quad - |\mathcal{E}(S_{11}, S_{21})| - |\mathcal{E}(S_{12}, S_{22})| \\
		&=: r - r^*,
	\end{aligned}
	\label{eq:E_inner_minus_start}
\end{equation}
\fi
}
where $r$ and $r^*$ are defined as 
\if\arXivver1
\begin{equation}
	\begin{aligned}
		r := |\mathcal{E}(S_{11}, S_{12})| + |\mathcal{E}(S_{21}, S_{22})|, \quad
		r^* := |\mathcal{E}(S_{11}, S_{21})| + |\mathcal{E}(S_{12}, S_{22})|.
	\end{aligned}
	\label{eq:def_r_in_out}
\end{equation}
\else
\begin{equation}
	\begin{aligned}
		r &:= |\mathcal{E}(S_{11}, S_{12})| + |\mathcal{E}(S_{21}, S_{22})|, \\
		r^* &:= |\mathcal{E}(S_{11}, S_{21})| + |\mathcal{E}(S_{12}, S_{22})|.
	\end{aligned}
	\label{eq:def_r_in_out}
\end{equation}
\fi
One can interpret $r$ (resp. $r^*$) as the number of connections that are within the same community in $\bm{\kappa}$ (resp. $\bm{\kappa}^*$) but across different ones in $\bm{\kappa}^*$ (resp. $\bm{\kappa}$). See Figure~\ref{fig:four_subsets} for an illustration. 

\textcolor{black}{Similarly, for the other part of the likelihood ratio, by using \eqref{eq:likelihood_rotation}, we obtain
\begin{equation*}
\begin{aligned}
	\frac{\mathcal{P}( \mathfrak{G} | \bm{\kappa}, \bm{g}, \mathcal{E})}{\mathcal{P}( \mathfrak{G} | \bm{\kappa}^*, \bm{g}^*, \mathcal{E})} &= \frac{\mathbbm{1}\left(g_{ij} = g_ig_j^{-1}; \; \forall (i,j) \in \mathcal{E}_{\text{inner}}\right)}{\mathbbm{1}\left(g_{ij} = g_i^*(g_j^*)^{-1}; \; \forall (i,j) \in \mathcal{E}_{\text{inner}}^*\right)} \cdot M^{\left(|\mathcal{E}_{\text{inner}}| - |\mathcal{E}_{\text{inner}}^*| \right)} \\
    &= \frac{\mathbbm{1}\left(g_{ij} = g_ig_j^{-1}; \; \forall (i,j) \in \mathcal{E}_{\text{inner}}\right)}{\mathbbm{1}\left(g_{ij} = g_i^*(g_j^*)^{-1}; \; \forall (i,j) \in \mathcal{E}_{\text{inner}}^*\right)} \cdot M^{r - r^*}, 
\end{aligned}
\end{equation*}
where we have applied \eqref{eq:E_inner_minus_start} in the last step.}

Now, by combining the results above, we can express the likelihood ratio as
\if\arXivver1
\begin{equation}
	\begin{aligned}
		\frac{\mathcal{P}(\bm{y}|\bm{x})}{\mathcal{P}(\bm{y}|\bm{x}^*)} 
		&= \left(\frac{Mp(1-q)}{q(1-p)}\right)^{r - r^*} \cdot 	\frac{\mathbbm{1}\left(g_{ij} = g_ig_j^{-1}; \; \forall (i,j) \in \mathcal{E}_{\text{inner}}\right)}{\mathbbm{1}(g_{ij} = g_i^*(g_j^*)^{-1}; \; \forall (i,j) \in \mathcal{E}_{\text{inner}}^*)}, 
	\end{aligned}
	\label{eq:likelihood_ratio}
\end{equation}
\else
\begin{equation}
	\begin{aligned}
	\frac{\mathcal{P}(\bm{y}|\bm{x})}{\mathcal{P}(\bm{y}|\bm{x}^*)} 
	&= \left(\frac{Mp(1-q)}{q(1-p)}\right)^{r - r^*} \\
	&\quad \cdot \frac{\mathbbm{1}\left(g_{ij} = g_ig_j^{-1}; \; \forall (i,j) \in \mathcal{E}_{\text{inner}}\right)}{\mathbbm{1}(g_{ij} = g_i^*(g_j^*)^{-1}; \; \forall (i,j) \in \mathcal{E}_{\text{inner}}^*)} , 
	\end{aligned}
\label{eq:likelihood_ratio}
\end{equation}
\fi
which will be served as the starting point of our analysis.

\section{Proof of Theorem~\ref{the:info_upper_bound_1}}
\label{sec:proof_clus_upper_bound}

For ease of presentation, without loss of generality, we assume all the ground truth group elements  $\bm{g}^* = \{g_i^*\}_{i = 1}^n$ satisfies
{\color{black}
\begin{equation}
	g_i^* = I, \quad \forall i \in \{1,\ldots, n\},
	\label{eq:assumption_theta}
\end{equation}
where $I$ stands for the identity element of group $\mathcal{G}_M$, i.e., $I \cdot g = g$.} %
Also, we denote 
\begin{equation*}
	\mathcal{P}_{\bm{x}^*}(\cdot) := \mathcal{P}(\cdot | \bm{x}^*), \quad \mathcal{P}_{\bm{x}}(\cdot) := \mathcal{P}(\cdot | \bm{x}), 
\end{equation*}
as the probability measure conditional on the ground truth $\bm{x}^*$ and the hypothesis $\bm{x}$, respectively. 

For readability, we split our proof into four major steps.

\if\arXivver1
\subsection{Step 1: A union bound}
\else
\vspace{0.2cm}
\subsubsection{Step 1: A union bound}
\fi
To begin with, given an observation $\bm{y}$, the MLE~\eqref{eq:MLE_rewrite} fails to exactly recover $\bm{\kappa}^*$ if and only if there exists an alternative hypothesis $\bm{x} = \{\bm{\kappa}, \bm{g}\}$ such that 
\begin{equation*}
	\text{dist}_{\mathrm{c}}(\bm{\kappa}, \bm{\kappa}^*) > 0 \quad \text{and} \quad \mathcal{P}_{\bm{x}}(\bm{y}) \geq \mathcal{P}_{\bm{x}^*}(\bm{y}).
	\label{eq:wrong_hypo_cond}
\end{equation*}
In other words, $\bm{x}$ is a wrong hypothesis with a likelihood higher than or equal to the ground truth $\bm{x}^*$. As a result, the error probability defined in \eqref{eq:error_prob} can be written as 
\if\arXivver1
\begin{equation}
	\begin{aligned}
		\mathcal{P}_{e, \text{c}}(\psi_{\text{MLE}}) &= \mathcal{P}_{\bm{x}^*}\big( \exists \bm{x} = \{\bm{\kappa}, \bm{g}\}\!: \text{dist}_{\mathrm{c}}(\bm{\kappa}, \bm{\kappa}^*) > 0, \; \mathcal{P}_{\bm{x}}(\bm{y}) \geq  \mathcal{P}_{\bm{x}^*}(\bm{y}) \big).
	\end{aligned}
	\label{eq:total_error_prob}
\end{equation}
\else
\begin{equation}
	\begin{aligned}
		&\mathcal{P}_{e, \text{c}}(\psi_{\text{MLE}}) \\
		&= \mathcal{P}_{\bm{x}^*}\left( \exists \bm{x}\!: \text{dist}_{\mathrm{c}}(\bm{\kappa}, \bm{\kappa}^*) > 0, \; \frac{\mathcal{P}_{\bm{x}}(\bm{y})}{\mathcal{P}_{\bm{x}^*}(\bm{y})} \geq 1 \right).
	\end{aligned}
	\label{eq:total_error_prob}
\end{equation}
\fi

\textcolor{black}{To proceed, recall the definition of the two subsets $S_{12}$ and $S_{21}$ in \eqref{eq:four_subsets}, which denotes the set of nodes assigned with different labels between $\bm{\kappa}$ and $\bm{\kappa}^*$. 
Let us further define 
\begin{equation}
\Delta := |S_{12}| = |S_{21}|
\label{eq:delta_definition}
\end{equation}
as their sizes, then for any wrong hypothesis $\bm{\kappa}$ with $\text{dist}_{\mathrm{c}}(\bm{\kappa}, \bm{\kappa}^*) > 0$, it must satisfy $1 \leq \Delta \leq \frac{n}{2}-1$. When $\Delta = \frac{n}{2}$, exact recovery can be achieved by swapping the labels of $S_1$ and $S_2$. In this way, for a fixed $\Delta$, we define
\begin{align*}
	\mathcal{H}_{\bm{\kappa}, \Delta} &:= \{\bm{\kappa}\!: |S_{12}| = |S_{21}| = \Delta \},
\end{align*}
as the set of all satisfying hypotheses (of the cluster memberships), which includes $\binom{n/2}{\Delta}^2$ 
number of unique $\bm{\kappa}$, i.e., $$|\mathcal{H}_{\bm{\kappa},\Delta}| = \binom{n/2}{\Delta}^2.$$
Furthermore, let $\bm{\kappa}_\Delta \in \mathcal{H}_{\bm{\kappa}, \Delta}$ be any fixed hypothesis in $\mathcal{H}_{\bm{\kappa}, \Delta}$, the error probability in \eqref{eq:total_error_prob} can be upper bounded as}
\if\arXivver1
\begin{align}
	\mathcal{P}_{e, \text{c}}(\psi_{\text{MLE}}) &\overset{(a)}{\leq} \sum_{\Delta = 1}^{n/2 - 1} \mathcal{P}_{\bm{x}^*} \left(\exists \bm{x} = \{\bm{\kappa}, \bm{g}\}:  \bm{\kappa} \in \mathcal{H}_{\bm{\kappa}, \Delta}, \; \mathcal{P}_{\bm{x}}(\bm{y}) \geq  \mathcal{P}_{\bm{x}^*}(\bm{y}) \right)\nonumber \\
	&\overset{(b)}{\leq} \sum_{\Delta = 1}^{n/2 - 1} \sum_{\bm{\kappa} \in \mathcal{H}_{\bm{\kappa}, \Delta }} \mathcal{P}_{\bm{x}^*} \left(\exists \bm{x} = \{\bm{\kappa}, \bm{g}\}:   \mathcal{P}_{\bm{x}}(\bm{y}) \geq  \mathcal{P}_{\bm{x}^*}(\bm{y}) \right)\nonumber \\
	&\color{black}{\overset{(c)}{=} \sum_{\Delta = 1}^{n/2 - 1} \binom{n/2}{\Delta}^2  \underbrace{\mathcal{P}_{\bm{x}^*} \left(\exists \bm{x} = \{\bm{\kappa}_{\Delta}, \bm{g}\}: \mathcal{P}_{\bm{x}}(\bm{y}) \geq  \mathcal{P}_{\bm{x}^*}(\bm{y}) \right)}_{=: \mathcal{P}_{\Delta}}}\nonumber \\
	&\overset{(d)}{\leq} 2 \sum_{\Delta = 1}^{\ceil*{n/4}} \binom{n/2}{\Delta}^2 \mathcal{P}_{\Delta}.
	\label{eq:bound_error_MLE_delta}
\end{align}
\else
\begin{align}
	&\mathcal{P}_{e, \text{c}}(\psi_{\text{MLE}}) \nonumber \\
	&\overset{(a)}{\leq} \sum_{\Delta = 1}^{n/2 - 1} \mathcal{P}_{\bm{x}^*} \left(\exists \bm{x} = \{\bm{\kappa}, \bm{g}\}:  \bm{\kappa} \in \mathcal{H}_{\bm{\kappa}, \Delta}, \; \mathcal{P}_{\bm{x}}(\bm{y}) \geq  \mathcal{P}_{\bm{x}^*}(\bm{y}) \right)\nonumber \\
	&\overset{(b)}{\leq} \sum_{\Delta = 1}^{n/2 - 1} \sum_{\bm{\kappa} \in \mathcal{H}_{\bm{\kappa}, \Delta }} \mathcal{P}_{\bm{x}^*} \left(\exists \bm{x} = \{\bm{\kappa}, \bm{g}\}:   \mathcal{P}_{\bm{x}}(\bm{y}) \geq  \mathcal{P}_{\bm{x}^*}(\bm{y}) \right)\nonumber \\
	&\overset{(c)}{=} \sum_{\Delta = 1}^{n/2 - 1} {n/2 \choose \Delta}^2  \underbrace{\mathcal{P}_{\bm{x}^*} \left(\exists \bm{x} \in \mathcal{H}_{\bm{\kappa}, \Delta}: \mathcal{P}_{\bm{x}}(\bm{y}) \geq  \mathcal{P}_{\bm{x}^*}(\bm{y}) \right)}_{=: \mathcal{P}_{\Delta}}\nonumber \\
	&\overset{(d)}{\leq} 2 \sum_{\Delta = 1}^{\ceil*{n/4}} {n/2 \choose \Delta}^2 \mathcal{P}_{\Delta}.
	\label{eq:bound_error_MLE_delta}
\end{align}
\fi

Here, $(a)$ holds by applying the union bound over different $\Delta$;  $(b)$ comes by further applying this over the set $\mathcal{H}_{\bm{\kappa}, \Delta}$; $(c)$ stems from the fact that any hypothesis $\bm{\kappa} \in \mathcal{H}_{\bm{\kappa}, \Delta}$ should contribute the same amount of error probability under a fixed $\Delta$ and therefore we focus on $\bm{\kappa} = \bm{\kappa}_\Delta$ as one of them; $(d)$ uses the symmetry between $\Delta$ and $\frac{n}{2} - \Delta$ by swapping the labels in $\bm{\kappa}$.

\if\arXivver1
\subsection{Step 2: Bound on each $\mathcal{P}_\Delta$}
\else
\vspace{0.2cm}
\subsubsection{Step 2: Bound on each $\mathcal{P}_\Delta$}
\fi
Our next job is to tightly bound $\mathcal{P}_{\Delta}$ for each $1 \leq \Delta \leq \ceil*{\frac{n}{4}}$. To this end, by using the likelihood ratio obtained in \eqref{eq:likelihood_ratio} and plugging in the scaling $p = \frac{a\log n}{n}, q = \frac{b\log n}{n}$, we have
\textcolor{black}{
\if\arXivver1
\begin{align}
	\mathcal{P}_{\Delta} &= \mathcal{P}_{\bm{x}^*}\left( \exists \bm{x} = \{\bm{\kappa}_{\Delta}, \bm{g}\}: \frac{\mathcal{P}_{\bm{x}}(\bm{y})}{\mathcal{P}_{\bm{x}^*}(\bm{y})} \geq 1 \right) \nonumber \\ 
    &=  \mathcal{P}_{\bm{x}^*}\left( \exists \bm{x} = \{\bm{\kappa}_{\Delta}, \bm{g}\}: \left(\frac{Ma(1-q)}{b(1-p)}\right)^{r - r^*} \cdot 	\frac{\mathbbm{1}(g_{ij} = g_ig_j^{-1}; \; \forall (i,j) \in \mathcal{E}_{\text{inner}})}{\mathbbm{1}(g_{ij} = g_i^*(g_j^*)^{-1}; \; \forall (i,j) \in \mathcal{E}_{\text{inner}}^*)} \geq  1 \right) \nonumber \\
	&= \mathcal{P}_{\bm{x}^*}\left( \exists \bm{x} = \{\bm{\kappa}_{\Delta}, \bm{g}\}: \left(\frac{Ma(1-q)}{b(1-p)}\right)^{r - r^*} \cdot 	\mathbbm{1}(g_{ij} = g_ig_j^{-1}; \; \forall (i,j) \in \mathcal{E}_{\text{inner}}) \geq 1 \right),  
	\label{eq:p_delta_bound}
\end{align}
\else
\begin{align}
	\mathcal{P}_{\Delta} &= \mathcal{P}_{\bm{x}^*}\left( \exists \bm{x} \in \mathcal{H}_{\bm{\kappa}, \Delta}: \frac{\mathcal{P}_{\bm{x}}(\bm{y})}{\mathcal{P}_{\bm{x}^*}(\bm{y})} \geq 1 \right) \nonumber \\ 
	&\leq  \mathcal{P}_{\bm{x}^*}\Bigg( \exists \bm{x} \in \mathcal{H}_{\bm{\kappa}, \Delta}: \left(\frac{Ma(1-q)}{b(1-p)}\right)^{r - r^*} \nonumber  \\
	&\quad\cdot 	\frac{\mathbbm{1}\left(g_{ij} = g_ig_j^{-1}; \; \forall (i,j) \in \mathcal{E}_{\text{inner}}\right)}{\mathbbm{1}\left(g_{ij} = g_i^*(g_j^*)^{-1}; \; \forall (i,j) \in \mathcal{E}_{\text{inner}}^*\right)} >  1 \Bigg) \nonumber \\
	&= \mathcal{P}_{\bm{x}^*}\Bigg( \exists \bm{x} \in \mathcal{H}_{\bm{\kappa}, \Delta}: \left(\frac{Ma(1-q)}{b(1-p)}\right)^{r - r^*} \nonumber \\
	&\quad \cdot 	\mathbbm{1}\left(g_{ij} = g_ig_j^{-1}; \; \forall (i,j) \in \mathcal{E}_{\text{inner}}\right) > 1 \Bigg), 
	\label{eq:p_delta_bound}
\end{align}
\fi
where the last step uses the fact that $$\mathcal{P}_{\bm{x}^*}\left(g_{ij} = g_i^*(g_j^*)^{-1}; \; \forall (i,j) \in \mathcal{E}^*_{\text{inner}}\right) = 1, $$
based on our model setting that any group transformation $g_{ij}$ is noiseless when $(i,j) \in \mathcal{E}^*_{\text{inner}}$.
}

\textcolor{black}{
To move forward, \eqref{eq:p_delta_bound} motivates us to define the following two events in the probability space:
\begin{equation}
	\begin{aligned}
	E_{1, \bm{x}}&:\quad  \left(\frac{Ma(1-q)}{b(1-p)}\right)^{r - r^*} \geq  1, \quad \text{for } \bm{x}, \\
	E_{2, \bm{x}}&:\quad g_{ij} = g_ig_j^{-1}, \;  \forall (i,j) \in \mathcal{E}_{\text{inner}}, \quad \text{for } \bm{x}.
\end{aligned}
\label{eq:event_1_2_def}
\end{equation}
In this way,  \eqref{eq:p_delta_bound} essentially measures the probability of having any $\bm{x} = \{\bm{\kappa}_{\Delta}, \bm{g}\}$ that makes the joint event $E_{1, \bm{x}} \cap E_{2, \bm{x}}$ happen. Therefore, \eqref{eq:p_delta_bound} can be rewritten as 
\begin{equation}
    \mathcal{P}_{\Delta} = \mathcal{P}_{\bm{x}^*}\left(\bigcup\limits_{\bm{x} : \bm{\kappa} = \bm{\kappa}_{\Delta}} \left(E_{1, \bm{x}} \bigcap E_{2, \bm{x}} \right)\right).
    \label{eq:p_delta_bound_step2}
\end{equation}
Furthermore, notice that for any $\bm{x} = \{\bm{\kappa}_{\Delta}, \bm{g}\}$ where the cluster membership $\bm{\kappa} = \bm{\kappa}_{\Delta}$ is fixed, the event $E_{1, \bm{x}}$ is actually invariant to the choice of $\bm{g}$ as it only involves edge connections (recall $r$ and $r^*$ defined in \eqref{eq:def_r_in_out}). Therefore, we can define 
\begin{equation*}
    E_1 :\quad  \left(\frac{Ma(1-q)}{b(1-p)}\right)^{r - r^*} \geq  1, \quad \text{for } \bm{\kappa} = \bm{\kappa}_\Delta, 
\end{equation*}
that does not depend on $\bm{x}$ and further express $\mathcal{P}_{\Delta}$ as 
\if\arXivver1
\begin{align}
	\mathcal{P}_{\Delta} &= \mathcal{P}_{\bm{x}^*}\left( \bigcup\limits_{\bm{x} : \bm{\kappa} = \bm{\kappa}_{\Delta}} \left(E_{1, \bm{x}} \bigcap E_{2, \bm{x}} \right)\right) = \mathcal{P}_{\bm{x}^*}\left( \bigcup\limits_{\bm{x} : \bm{\kappa} = \bm{\kappa}_{\Delta}} \left(E_{1} \bigcap E_{2, \bm{x}} \right)\right) = \mathcal{P}_{\bm{x}^*}\left( E_{1} \bigcap  \underbrace{\left(\bigcup\limits_{\bm{x} : \bm{\kappa} = \bm{\kappa}_{\Delta}} E_{2, \bm{x}} \right)}_{=: E_2}\right) \nonumber \\
	&=  \mathcal{P}_{\bm{x}^*}\left( E_1 \bigcap  E_2\right) , 
	\label{eq:p_delta_bound_2}
\end{align}
\else
\begin{align}
	\mathcal{P}_{\Delta} &\leq \mathcal{P}_{\bm{x}^*}\left( \bigcup\limits_{\bm{x} \in \mathcal{H}_{\bm{\kappa}, \Delta}} \left(E_{1} \bigcap E_{\bm{x},2} \right)\right) \nonumber \\
	&= \mathcal{P}_{\bm{x}^*}\left( E_1 \bigcap  \underbrace{\left(\bigcup\limits_{\bm{x} \in \mathcal{H}_{\bm{\kappa}, \Delta}} E_{\bm{x},2} \right)}_{=: E_2}\right) \nonumber \\
	&=  \mathcal{P}_{\bm{x}^*}\left( E_1 \bigcap  E_2\right) , 
	\label{eq:p_delta_bound_2}
\end{align}
\fi
where we denote $E_2$ as the union of $E_{2, \bm{x}}$.}

To proceed, recall the definitions of $r$ and $r^*$ in \eqref{eq:def_r_in_out}, we further split $r$ into $r = r_1 + r_2$ such that   
\begin{equation}
	r_1 := |\mathcal{E}(S_{11}, S_{12})|, \quad r_2 := |\mathcal{E}(S_{21}, S_{22})|
\label{eq:r_1_r_2_definition}
\end{equation}
Under the ground truth $\bm{x}^*$, the three random variables $r_1$, $r_2$, and $r^*$ are binomial and independent as they involve mutually disjoint sets of edges (one can refer to Figure~\ref{fig:four_subsets} for visualization). 
Therefore, by conditioning on $r_1$, $r_2$, and $r^*$, \eqref{eq:p_delta_bound_2} satisfies
\if\arXivver1
\begin{align}
	\mathcal{P}_{\Delta} &= \sum_{r_1}\sum_{r_2}\sum_{r^*} \mathcal{P}_{\bm{x}^*}(r_1)\mathcal{P}_{\bm{x}^*}(r_2)\mathcal{P}_{\bm{x}^*}(r^*) \mathcal{P}_{\bm{x}^*}\left(E_1 \bigcap E_2 \mid r_1, r_2, r^* \right) \nonumber \\ 
	&=  \sum_{r_1}\sum_{r_2}\sum_{r^*} \mathcal{P}_{\bm{x}^*}(r_1)\mathcal{P}_{\bm{x}^*}(r_2)\mathcal{P}_{\bm{x}^*}(r^*) \mathcal{P}_{\bm{x}^*} \left( \left(\left(\frac{Ma(1-q)}{b(1-p)}\right)^{r_1 + r_2 - r^*} > 1\right)  \bigcap E_2 \mid r_1, r_2, r^* \right) \nonumber \\
	&=  \sum_{r_1}\sum_{r_2}\sum_{r^*} \mathcal{P}_{\bm{x}^*}(r_1)\mathcal{P}_{\bm{x}^*}(r_2)\mathcal{P}_{\bm{x}^*}(r^*) \cdot  \mathbbm{1}\left(\left(\frac{Ma(1-q)}{b(1-p)}\right)^{r_1 + r_2 - r^*} > 1 \right) \cdot \mathcal{P}_{\bm{x}^*}\left(E_2 \mid r_1, r_2, r^*\right) ,
	\label{eq:p_delta_bound_3}
\end{align}
\else
\begin{align}
	\mathcal{P}_{\Delta} \nonumber 
	&\leq \sum_{r_1}\sum_{r_2}\sum_{r^*} \bigg[\mathcal{P}_{\bm{x}^*}(r_1)\mathcal{P}_{\bm{x}^*}(r_2)\mathcal{P}_{\bm{x}^*}(r^*) \nonumber \\ 
	&\quad \cdot \mathcal{P}_{\bm{x}^*}\left(E_1 \bigcap E_2 \mid r_1, r_2, r^* \right)\bigg] \nonumber \\ 
	&=  \sum_{r_1}\sum_{r_2}\sum_{r^*} \Bigg[\mathcal{P}_{\bm{x}^*}(r_1)\mathcal{P}_{\bm{x}^*}(r_2)\mathcal{P}_{\bm{x}^*}(r^*) \nonumber \\
	&\quad \cdot \mathcal{P}_{\bm{x}^*} \left( \left(\left(\frac{Ma(1-q)}{b(1-p)}\right)^{r_1 + r_2 - r^*} > 1\right)  \bigcap E_2 \mid r_1, r_2, r^* \right) \Bigg] \nonumber \\
	&=  \sum_{r_1}\sum_{r_2}\sum_{r^*} \Bigg[\mathcal{P}_{\bm{x}^*}(r_1)\mathcal{P}_{\bm{x}^*}(r_2)\mathcal{P}_{\bm{x}^*}(r^*) \nonumber \\
	&\quad \cdot  \mathbbm{1}\left(\left(\frac{Ma(1-q)}{b(1-p)}\right)^{r_1 + r_2 - r^*} > 1 \right) \cdot \mathcal{P}_{\bm{x}^*}\left(E_2 \mid r_1, r_2, r^*\right) \Bigg], 
	\label{eq:p_delta_bound_3}
\end{align}
\fi
where the last step comes from the fact that given $r_1$, $r_2$, and $r^*$, the probability $\mathcal{P}_{\bm{x}^*} (E_1 \mid r_1, r_2, r^* )$ is either 0 or 1.  Therefore, we represent it as an indicator function.  Furthermore, by introducing an extra factor $\alpha > 0$, one can see that it satisfies 
\begin{align*}
	\mathbbm{1}\left(\left(\frac{Ma(1-q)}{b(1-p)}\right)^{r_1 + r_2 - r^*} > 1 \right)  &= \mathbbm{1}\left(\left(\frac{Ma(1-q)}{b(1-p)}\right)^{\alpha(r_1 + r_2 - r^*)} > 1 \right)  \\
	&\leq  \left(\frac{Ma(1-q)}{b(1-p)}\right)^{\alpha(r_1 + r_2 - r^*)}. 
\end{align*}
Plugging this into \eqref{eq:p_delta_bound_3} yields
\if\arXivver1
\begin{equation}
	\mathcal{P}_{\Delta} \leq \min_{\alpha > 0} \left(  \sum_{r_1}\sum_{r_2}\sum_{r^*} \mathcal{P}_{\bm{x}^*}(r_1)\mathcal{P}_{\bm{x}^*}(r_2)\mathcal{P}_{\bm{x}^*}(r^*) \cdot  \left(\frac{Ma(1-q)}{b(1-p)}\right)^{\alpha(r_1 + r_2 - r^*)} \cdot \mathcal{P}_{\bm{x}^*}\left(E_2 \mid r_1, r_2,  r^*\right) \right).
	\label{eq:bound_P_Delta_middle}
\end{equation}
\else
\begin{equation}
	\begin{aligned}
		\mathcal{P}_{\Delta} &\leq \min_{\alpha > 0} \Bigg[  \sum_{r_1}\sum_{r_2}\sum_{r^*} \Bigg(\mathcal{P}_{\bm{x}^*}(r_1)\mathcal{P}_{\bm{x}^*}(r_2)\mathcal{P}_{\bm{x}^*}(r^*) \\
		&\quad \cdot \left(\frac{Ma(1-q)}{b(1-p)}\right)^{\alpha(r_1 + r_2 - r^*)} \cdot \mathcal{P}_{\bm{x}^*}\left(E_2 \mid r_1, r_2,  r^*\right)\Bigg) \Bigg].
	\end{aligned}
	\label{eq:bound_P_Delta_middle}
\end{equation}
\fi
Notably, $\alpha$ is %
 crucial for obtaining the sharp result,  %
as we will minimize over all $\alpha > 0$ later on.

\if\arXivver1
\subsection{Step 3: Cycle consistency analysis}
\else
\vspace{0.2cm}
\subsubsection{Step 3: Cycle analysis}
\fi
\label{sec:prood_upper_bound_clus_cycle_analysis}
Now, we need to upper bound $\mathcal{P}_{\bm{x}^*}\left(E_2 \mid r_1, r_2, r^*\right)$.  Importantly, a naive way by applying the union bound over all possible hypotheses $\bm{x} = \{\bm{\kappa}_\Delta, \bm{g}\}$ such as
\begin{align*}
	\mathcal{P}_{\bm{x}^*}\left(E_2 \mid r_1, r_2, r^*\right) = \mathcal{P}_{\bm{x}^*}\left(\bigcup\limits_{\bm{x} : \bm{\kappa} = \bm{\kappa}_{\Delta}} E_{2, \bm{x}} \mid r_1, r_2, r^*\right) \leq \sum_{\bm{x} : \bm{\kappa} = \bm{\kappa}_{\Delta}} \mathcal{P}_{\bm{x}^*}\left(E_{2, \bm{x}}  \mid r_1,  r_2, r^* \right)
\end{align*}
fails to give a sharp result, since the summation over $\bm{x} : \bm{\kappa} = \bm{\kappa}_{\Delta}$ includes $M^n$ terms %
as each node has $M$ choices of its group elements. As a result, the upper bound increases exponentially as $n$ becomes large. 

Instead, we have to bound $\mathcal{P}_{\bm{x}^*}\left(E_2 \mid r_1, r_2, r^*\right)$ directly as a whole. Let us restate the event $E_2$ defined in \eqref{eq:p_delta_bound_2} as 
\begin{equation}
	E_2: \quad \exists \bm{x} = \{\bm{\kappa}_{\Delta}, \bm{g}\}: \; g_{ij} = g_ig_j^{-1}, \; \forall (i,j) \in \mathcal{E}_{\text{inner}}.
	\label{eq:E_2_def}
\end{equation}
Here, recall $\mathcal{E}_{\text{inner}}$ denotes the set of edges within the communities $S_1$ and $S_2$, i.e., $\mathcal{E}_{\text{inner}} = \mathcal{E}(S_1) \cup \mathcal{E}(S_2)$, then we can consider $S_1$ and $S_2$ separately by further defining the following two events:
\begin{equation}
\begin{aligned}
E_{2,S_1}: &\quad \exists \bm{x} = \{\bm{\kappa}_{\Delta}, \bm{g}\}: g_{ij} = g_ig_j^{-1}, \; \forall (i,j) \in \mathcal{E}(S_1), \\
E_{2,S_2}: &\quad \exists \bm{x} = \{\bm{\kappa}_{\Delta}, \bm{g}\}: g_{ij} = g_ig_j^{-1}, \; \forall (i,j) \in \mathcal{E}(S_2), 
\end{aligned}
\label{eq:def_E_2_S_12}
\end{equation}
and one can see that $E_2 = E_{2, S_1} \cap E_{2, S_2}$. In this way, we have  
	\if\arXivver1
\begin{align}
	\mathcal{P}_{\bm{x}^*}\left(E_2 \mid r_1, r_2, r^*\right) &\overset{(a)}{=} \mathcal{P}_{\bm{x}^*}\left(E_2 \mid r_1, r_2 \right) = \mathcal{P}_{\bm{x}^*}\left(E_{2, S_1} \cap E_{2, S_2} \mid r_1, r_2\right) \nonumber \\
	&\overset{(b)}{=}  \mathcal{P}_{\bm{x}^*}\left(E_{2, S_1}\mid r_1\right) \cdot \mathcal{P}_{\bm{x}^*}\left(E_{2, S_2}\mid r_2\right) 
	\label{eq:bound_E_2_r_1_2}
\end{align}
\else
\begin{align}
	\mathcal{P}_{\bm{x}^*}\left(E_2 \mid r_1, r_2, r^*\right) &\overset{(a)}{=} \mathcal{P}_{\bm{x}^*}\left(E_2 \mid r_1, r_2 \right) \nonumber \\
	&= \mathcal{P}_{\bm{x}^*}\left(E_{2, S_1} \cap E_{2, S_2} \mid r_1, r_2\right) \nonumber \\
	&\overset{(b)}{=}  \mathcal{P}_{\bm{x}^*}\left(E_{2, S_1}\mid r_1\right) \cdot \mathcal{P}_{\bm{x}^*}\left(E_{2, S_2}\mid r_2\right) 
	\label{eq:bound_E_2_r_1_2}
\end{align}
\fi
where $(a)$ holds since $E_2$ does not involve $r^*$; $(b)$ comes from the fact that $E_{2, S_1}$ and 
$E_{2, S_2}$ are independent given $r_1$ and $r_2$. As a result, because of the symmetry, we can focus on bounding $\mathcal{P}_{\bm{x}^*}\left(E_{2, S_1}\mid r_1\right)$ and the similar result applies to $\mathcal{P}_{\bm{x}^*}\left(E_{2, S_2}\mid r_2\right)$.

\begin{figure*}[t!]
	\vspace{-50pt}
	\includegraphics[width = 1.0\textwidth]{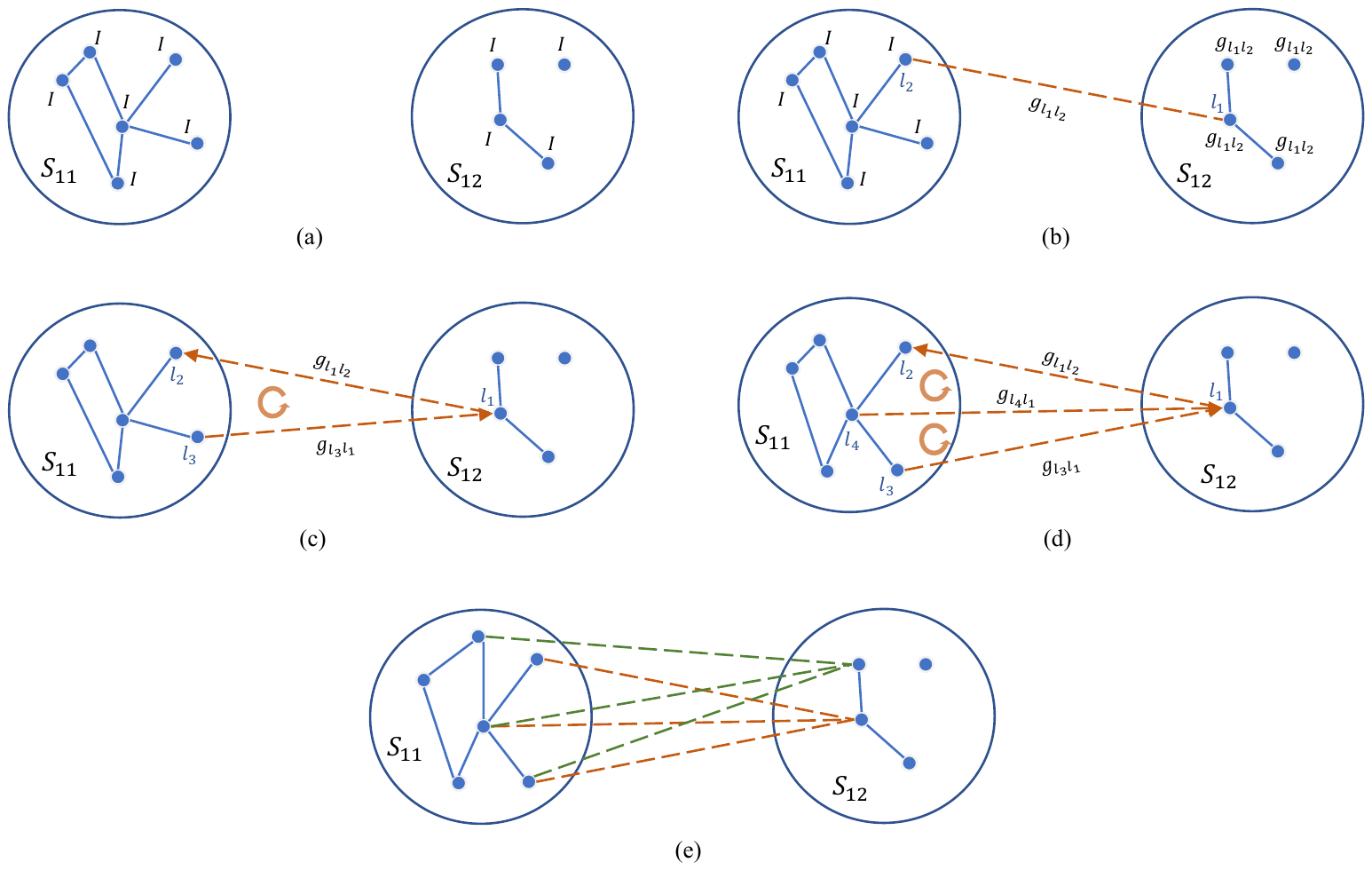}
	\vspace{-50pt}
	\caption{ %
     Some illustrations on $E_{2, S_1}$. We consider five different scenarios: (a) and (b): no edge or only one edge connected between $S_{11}$ and $S_{12}$, then $E_{2, S_1}$ occurs with probability 1; (c): assuming the subgraph formed by $S_{11}$ is connected, and two edges connected to the same node in $S_{12}$ form a cycle, then $E_{2, S_1}$ occurs when the cycle consistency is satisfied with probability $M^{-1}$; (d) and (e): more edges between $S_{11}$ and $S_{12}$ leads to more cycles and further lower probability of $E_{2, S_1}$. See text for details.
	}
	\label{fig:E_2_illustration}
\end{figure*}

To this end, by definition, $E_{2, S_1}$ occurs as long as there exists a hypothesis of the group elements $\bm{g}$ that satisfies the consistency $g_{ij} = g_ig_j^{-1}$ on each edge in $S_1$. However, determining the conditional probability $\mathcal{P}_{\bm{x}^*}\left(E_{2, S_1}\mid r_1\right)$ appears to be challenging as it involves two sources of randomness: 
\begin{enumerate}
	\vspace{0.2cm}
	\item For any pair of node $(i,j)$ within $S_{11}$ or $S_{12}$, they are connected with probability $p = \frac{a \log n}{n}$, and the group transformation observed is {\color{black} $g_{ij} = g_i^*(g_j^*)^{-1} = I$} (recall the assumption \eqref{eq:assumption_theta} that $g_i^* = I, \; i = 1, \ldots, n$) with probability 1. 
	\vspace{0.2cm}
	\item For any pair of nodes $(i,j)$ between $S_{11}$ and  $S_{12}$, the edges are randomly assigned under the constraint that the total number of edges between $S_{11}$ and  $S_{12}$ is $r_1$ (recall the definition of $r_1$ in \eqref{eq:r_1_r_2_definition}), and $g_{ij}$ is uniformly drawn from $\mathcal{G}_M$. 
	\vspace{0.2cm}
\end{enumerate}
Here, our analysis on $\mathcal{P}_{\bm{x}^*}\left(E_{2, S_1}\mid r_1\right)$ starts from several simple examples that correspond to different choices of $r_1$, as shown in Figure~\ref{fig:E_2_illustration} that, which are elaborated as follows:

\begin{itemize}
	\vspace{0.2cm}
	\item [(a)] We start from Figure~\ref{fig:E_2_illustration}(a) with $r_1 = 0$ such that no edges connect between $S_{11}$ and $S_{12}$. In this case, $g_{ij}$ on all edges within $S_{11}$ and $S_{12}$ satisfy $
		g_{ij} = g_i^*(g_j^*)^{-1} = I$.  
	As a result, no matter how the nodes within $S_{11}$ or $S_{12}$ are connected, we have $$\mathcal{P}_{\bm{x}^*}\left(E_{2, S_1}|r_1 = 0\right)= 1, $$ as a hypothesis $g_i = I, \; i = 1, \ldots, n$ always satisfies the consistency.  
	
	\vspace{0.2cm}
	\item [(b)] Now, suppose $r_1 = 1$ where there is an edge $(l_1,l_2)$ connecting $l_1 \in S_{12}$ and $l_2 \in S_{11}$. In this case, the group transformation $g_{l_1l_2}$ is uniformly drawn from $\mathcal{G}_M$. However, for any realization of $g_{l_1l_2}$, there always exists a hypothesis $\bm{g}$ satisfy the consistency by taking $$g_{i} = I, \; \forall i \in S_{11} \quad \text{and} \quad g_{i} = g_{l_1l_2}, \; \forall i \in S_{12}, $$ as shown in  Figure~\ref{fig:E_2_illustration}(b). Therefore, we still have $$\mathcal{P}_{\bm{x}^*}\left(E_{2, S_1}|r_1 = 1\right)= 1.$$ 
	
	\vspace{0.2cm}
	\item [(c)] Next, we consider when $r_1 = 2$ and further assume
	\begin{enumerate}[leftmargin=15pt]
		\vspace{0.2cm}
		\item The subgraph formed by $S_{11}$, denoted as $G(S_{11})$, is connected.
		\vspace{0.2cm}
		\item The $r_1 = 2$ edges between $S_{11}$ and $S_{12}$ connect the same node $l_1 \in S_{12}$.
		\vspace{0.1cm}
	\end{enumerate}
	Let us denote this event as $E_c$. As a result, let $l_2$ and $l_3 \in S_{11}$ denote the two nodes that connect with $l_1$. Then, $l_2$ and $l_3$ are  connected by some path in $G(S_{11})$ due to the assumption that $G(S_{11})$ is connected. This path, together with the two edges $(l_1, l_2)$ and $(l_1, l_3)$, form a \textit{cycle} across $S_{11}$ and $S_{12}$:
	\begin{equation}
		\begin{aligned}
					l_1 \; \xrightarrow{g_{l_1l_2}} \; l_2 \; \xrightarrow{{\color{white} g_{l_1l_2}}} \; \cdots \;   \xrightarrow{{\color{white} g_{l_1l_2}}} \; l_3 \; \xrightarrow{g_{l_3l_1}} \; l_1, 
		\end{aligned}
	\label{eq:def_E_c}
	\end{equation}
	which is shown in Figure~\ref{fig:E_2_illustration}(c).
	In this case, one can notice that $E_2$ occurs only if the cycle consistency is met, %
    i.e., 
	\begin{equation}
	g_{l_1l_2} \cdots g_{l_3l_1} = I, 
	\label{eq:cycle_consistency_proof}
	\end{equation}
	which is so-called the \textit{cycle consistency}. However, since  $g_{l_1l_2}$ and $g_{l_3l_1}$ are both randomly drawn from $\mathcal{G}_M$, no matter what the hypothesis $\bm{g}$ is, the probability that \eqref{eq:cycle_consistency_proof} holds is only $M^{-1}$, Hence, we have $$\mathcal{P}_{\bm{x}^*}\left(E_{2, S_1}|E_c, r_1 = 2\right) = M^{-1},$$
	given that $E_c$ occurs.
	
	\vspace{0.2cm}
	\item [(d)] Now, we consider $r_1 = 3$ and assume an event $E_d$ similar to $E_c$ above occurs, i.e., $G(S_{11})$ is connected and all the three edges connect the same node $l_1 \in S_{22}$. We denote the new node connect to $l_1$ as $l_4 \in S_{11}$. As a result, this gives rise another cycle  
	\begin{equation*}
		l_1 \; \xrightarrow{g_{l_1l_4}} \; l_4 \; \xrightarrow{{\color{white} g_{l_1l_2}}} \; \cdots \;   \xrightarrow{{\color{white} g_{l_1l_2}}} \; l_3 \; \xrightarrow{g_{l_3l_1}} \; l_1, 
	\end{equation*}
	as shown in Figure~\ref{fig:E_2_illustration}(d). Again, by applying the same argument as in part $(c)$, one can see that the new cycle satisfies the cycle consistency with probability $M^{-1}$ and is independent with the old one. Therefore, we have $$\mathcal{P}_{\bm{x}^*}\left(E_{2, S_1}|E_{d}, r_1 = 3\right) = M^{-2}, $$
	given that $E_d$ occurs.
	
	\vspace{0.2cm}
	\item [(e)] Lastly, we consider a more involved example when two nodes in $S_{12}$ connect with $S_{11}$, and the subgraph $G(S_{11})$ is still connected, as shown in Figure~\ref{fig:E_2_illustration}(e). We denote such an event as $E_e$. Here, each of the two nodes has three edge connections, and therefore there are at least four cycles. In this case, we have 
    $$\mathcal{P}_{\bm{x}^*}\left(E_{2, S_1}|E_e, r_1 = 6\right) \leq M^{-4}.$$
    Notably, here we use `$\leq$' instead of `$=$' since we have ignored the edge connections in $S_{12}$, and there will be more cycles exist by considering those edges. However, such an upper bound is sufficient to our analysis.	
\vspace{0.2cm}
\end{itemize}

Now, based on the previous discussions, our analysis in Figure~\ref{fig:E_2_illustration} can be generalized to the following result

\begin{lemma}
Let $E_{\mathrm{con}}$ denotes the event that the subgraph $G(S_{11})$ is connected. Then, 
\begin{equation}
	\mathcal{P}_{\bm{x}^*}\left(E_{2, S_1}\mid r_1, E_{\mathrm{con}}\right) \leq \min \left\{1, \; M^{\Delta - r_1}\right\}, 
	\label{eq:bound_E_2_S_1_connect}
\end{equation} 
where $\Delta := |S_{12}| = |S_{21}|$ is defined in \eqref{eq:delta_definition}. 
\label{lemma:bound_E_2_S_1_connect}
\end{lemma}

\begin{proof}
For each node $i \in S_{12}$, let $r_{1, i}$ denotes the number of edges it connects to $S_{11}$. Importantly, since $G(S_{11})$ is assumed to be connected, there are always at least $\max\{0, \; r_{1, i}-1\}$ number of cycles formed by $i$ like \eqref{eq:def_E_c}. Therefore, let 
$E_{\text{cycle}, i}$ denotes the event that all of them satisfy the cycle consistency such as \eqref{eq:cycle_consistency_proof}. Then,
	\if\arXivver1
\begin{align*}
	\mathcal{P}_{\bm{x}^*}\left(E_{2, S_1}\mid r_1, E_{\mathrm{con}}\right)  &\leq \mathcal{P}_{\bm{x}^*}\left(\bigcap_{i \in S_{12}} E_{\text{cycle}, i} \mid r_1, E_{\mathrm{con}}\right) \overset{(a)}{=} \prod_{i \in S_{12}} \mathcal{P}_{\bm{x}^*}\left(E_{\text{cycle}, i} \mid r_1, E_{\mathrm{con}}\right)\\
	&= \prod_{i \in S_{12}} M^{-\max\{0, \; r_{1, i}-1\}} = \prod_{i \in S_{12}} \min \left\{1, \; M^{1 - r_{1,i} }\right\} \\
	&\leq \min \left\{1, \; \prod_{i \in S_{12}} M^{1 - r_{1, i}}\right\} = \min\left\{1, \; M^{\sum_{i \in S_{12}} \left(r_{1,i}-1\right)}\right\}\\
	&\overset{(b)}{=} \min \left\{1, \; M^{\Delta - r_1}\right\}, 
\end{align*}
\else
\begin{align*}
	&\mathcal{P}_{\bm{x}^*}\left(E_{2, S_1}\mid r_1, E_{\mathrm{con}}\right)  \\
	&\leq \mathcal{P}\left(\bigcap_{i \in S_{12}} E_{\text{cycle}, i}\right) \overset{(a)}{=} \prod_{i \in S_{12}} \mathcal{P}\left(E_{\text{cycle}, i}\right)\\
	&= \prod_{i i \in S_{12}} M^{-\max\{0, \; r_{1, i}-1\}} = \prod_{i \in S_{12}} \min \left\{1, \; M^{1 - r_{1,i} }\right\} \\
	&\leq \min \left\{1, \; \prod_{i \in S_{12}} M^{1 - r_{1, i}}\right\} = \min\left\{1, \; M^{\sum_{i \in S_{12}} \left(r_{1,i}-1\right)}\right\}\\
	&\overset{(b)}{=} \min \left\{1, \; M^{\Delta - r_1}\right\}
\end{align*}
\fi
where $(a)$ uses the fact that the set of events $\{E_{\text{cycle}, i}\}_{i \in S_{12}}$ are mutually independent as they involve disjoint sets of edges; $(b)$ holds since $\sum_{i \in S_{12}} r_{1, i} = r_1$.
\end{proof}

However, Lemma~\ref{lemma:bound_E_2_S_1_connect} is still one step away from the desired bound as it requires $G(S_{11})$ to be connected. Fortunately, as $G(S_{11})$ follows an Erdős–Rényi model, when $p = \Theta\left(\frac{\log n}{n}\right)$, we can show that at least a giant component of size $(1-o(1)) \cdot |S_{11}|$ exists in $G(S_{11})$ w.h.p.~(see Theorem~\ref{the:bound_Z_n_3}). In this way, an upper bound similar to \eqref{eq:bound_E_2_S_1_connect} can be obtained as follows:

\begin{lemma}
Given $p = \Theta\left(\frac{\log n}{n}\right)$, it satisfies
	\if\arXivver1
\begin{equation}
	\mathcal{P}_{\bm{x}^*}\left(E_{2, S_1}\mid r_1\right) \leq n^{-\Theta(\log n)} + \min\left\{1, M^{\Delta - (1-o(1))r_1}\right\}, 
\end{equation}
\else
\begin{equation}
	\begin{aligned}
		\mathcal{P}_{\bm{x}^*}\left(E_{2, S_1}\mid r_1\right) &\leq n^{-\Theta(\log n)} \\
		&\quad+ \min\left\{1, M^{\Delta - (1-o(1))r_1}\right\}, 
	\end{aligned}
\end{equation}
\fi
which holds for $\Delta \leq \ceil*{n/4}$. 
\label{lemma:bound_E_2_S_1}
\end{lemma}

The proof is deferred to Appendix~\ref{sec:proof_bound_E_2_S_1}. As a result, Lemma~\ref{lemma:bound_E_2_S_1} relaxes the condition of the subgraph $G(S_{11})$ being connected made in~Lemma~\ref{lemma:bound_E_2_S_1_connect}, and thus can be applied to \eqref{eq:bound_E_2_r_1_2} directly. By symmetry we have the same result for $S_2$ as
\begin{equation*}
	\mathcal{P}_{\bm{x}^*}\left(E_{2, S_2}\mid r_2\right) \leq n^{-\Theta(\log n)} + \min\left\{1, M^{\Delta - (1-o(1))r_2}\right\}.
\end{equation*}
Plugging these into \eqref{eq:bound_E_2_r_1_2} yields
\if\arXivver1
\begin{align}
\mathcal{P}_{\bm{x}^*}\left(E_2 \mid r_1, r_2, r^*\right) &\leq \left(n^{-\Theta(\log n)} + \min\left\{1, M^{\Delta - (1-o(1))r_1}\right\}\right) \cdot \left(n^{-\Theta(\log n)} + \min\left\{1, M^{\Delta - (1-o(1))r_2}\right\}\right) \nonumber\\
&\leq n^{-\Theta(\log n)} + M^{2\Delta - (1 - o(1))(r_1 + r_2)} \nonumber \\
&= n^{-\Theta(\log n)} + M^{2\Delta - (1 - o(1))r}.
\label{eq:bound_E_2_r_1_r_2}
\end{align}
\else
\begin{align}
	&\mathcal{P}_{\bm{x}^*}\left(E_2 \mid r_1, r_2, r^*\right) \nonumber \\
	&\leq \left(n^{-\Theta(\log n)} + \min\left\{1, M^{\Delta - (1-o(1))r_1}\right\}\right) \nonumber \\
	&\quad \cdot \left(n^{-\Theta(\log n)} + \min\left\{1, M^{\Delta - (1-o(1))r_2}\right\}\right) \nonumber\\
	&\leq n^{-\Theta(\log n)} + M^{2\Delta - (1 - o(1))(r_1 + r_2)} \nonumber \\
	&= n^{-\Theta(\log n)} + M^{2\Delta - (1 - o(1))r}.
	\label{eq:bound_E_2_r_1_r_2}
\end{align}
\fi
As a result, the upper bound \eqref{eq:bound_E_2_r_1_r_2} does not depend on $r_1$ and $r_2$ but only their sum $r$. 

Now, by plugging \eqref{eq:bound_E_2_r_1_r_2} into \eqref{eq:bound_P_Delta_middle} and taking the summation over $r$ and $r^*$, we obtain the following upper bound on $\mathcal{P}_\Delta$. The details are left in Appendix~\ref{sec:proof_bound_P_Delta}. 
\begin{lemma}
$\mathcal{P}_\Delta$ in \eqref{eq:bound_P_Delta_middle} satisfies
\begin{equation*}
\mathcal{P}_\Delta \leq M^{2\Delta} \cdot \exp\left(-\frac{N_\Delta \log n}{n} \left(a + b - (2 +o(1))\sqrt{\frac{ab}{M}}\right)\right),
\end{equation*}
where $N_\Delta:= 2\Delta\left(\frac{n}{2} - \Delta \right)$.
\label{lemma:bound_P_Delta}
\end{lemma}

\if\arXivver1
\subsection{Step 4: Sum over $\Delta$}
\else
\vspace{0.2cm}
\subsubsection{Step 4: Sum over $\Delta$}
\fi
Now, we are ready to bound the error probability $\mathcal{P}_{e, \text{c}}(\psi_{\text{MLE}})$ in \eqref{eq:bound_error_MLE_delta}, which is given as 
\if\arXivver1
\begin{align}
\mathcal{P}_{e, \text{c}}(\psi_{\text{MLE}}) &\leq 2\sum_{\Delta = 1}^{\ceil*{n/4}} \binom{n/2}{\Delta}^2 \cdot  M^{2\Delta} \cdot \exp\left(-\frac{N_\Delta \log n}{n} \left(a + b - (2 +o(1))\sqrt{\frac{ab}{M}}\right)\right) \nonumber \\
&\overset{(a)}{\leq} 2 \sum_{\Delta = 1}^{\ceil*{n/4}} \underbrace{\left(\frac{en}{2\Delta}\right)^{2\Delta} \cdot M^{2\Delta} \cdot \exp\left(-\frac{N_\Delta \log n}{n} \left(a + b - (2 +o(1))\sqrt{\frac{ab}{M}}\right)\right) }_{=: g(\Delta)} \nonumber \\
&\leq 2\sum_{\Delta = 1}^{\ceil*{n/4}} g(\Delta), 
\label{eq:bound_P_e_g_Delta}
\end{align}
\else
\begin{align}
	&\mathcal{P}_{e, \text{c}}(\psi_{\text{MLE}}) \nonumber \\
	&\leq 2\sum_{\Delta = 1}^{\ceil*{n/4}} \Bigg[{n/2 \choose \Delta}^2 \cdot  M^{2\Delta} \nonumber \\
	&\quad \cdot \exp\left(-\frac{N_\Delta \log n}{n} \left(a + b - (2 +o(1))\sqrt{\frac{ab}{M}}\right)\right)\Bigg] \nonumber \\
	&\overset{(a)}{\leq} 2 \sum_{\Delta = 1}^{\ceil*{n/4}} \Bigg[\left(\frac{en}{2\Delta}\right)^{2\Delta} \cdot M^{2\Delta} \nonumber \\
	&\quad \cdot \exp\left(-\frac{N_\Delta \log n}{n} \left(a + b - (2 +o(1))\sqrt{\frac{ab}{M}}\right)\right) \Bigg]\nonumber \\
	&=: 2\sum_{\Delta = 1}^{\ceil*{n/4}} g(\Delta)
	\label{eq:bound_P_e_g_Delta}
\end{align}
\fi
where $(a)$ uses the property of binomial coefficient that $\binom{n}{k} \leq \left(e\frac{n}{k}\right)^k$ and we define each term in the summation as $g(\Delta)$. To proceed, we take the logarithm of $g(\Delta)$ as 
	\if\arXivver1
\begin{align*}
	\log g(\Delta) &= 2\Delta\log \left(\frac{en}{2\Delta}\right) + 2\Delta \log M  - \frac{N_\Delta \log n}{n} \cdot \left(a + b - (2 +o(1))\sqrt{\frac{ab}{M}}\right)\\
	&= 2\Delta \left[ \log\left(\frac{en}{2\Delta}\right) + \log M - \left(\frac{1}{2} - \frac{\Delta}{n}\right) \cdot \left(a + b - (2 +o(1))\sqrt{\frac{ab}{M}}\right) \log n\right].
\end{align*}
\else
\begin{align*}
	&\log g(\Delta) \\
	&= 2\Delta\log \left(\frac{en}{2\Delta}\right) + 2\Delta \log M  \\
	&\quad- \frac{N_\Delta \log n}{n} \cdot \left(a + b - (2 +o(1))\sqrt{\frac{ab}{M}}\right)\\
	&= 2\Delta \Bigg[ \log\left(\frac{en}{2\Delta}\right) + \log M \\
	&\quad -\left(\frac{1}{2} - \frac{\Delta}{n}\right) \cdot \left(a + b - (2 +o(1))\sqrt{\frac{ab}{M}}\right) \log n\Bigg].
\end{align*}
\fi
Then, we study $g(\Delta)$ on the following two cases depending $\Delta$: 
\begin{itemize}[leftmargin=*]
	\item \textbf{Case 1:} When $1 \leq \Delta \leq \delta_1 n$, where $\delta_1 = \frac{c - c_1}{2(1 + c)}$ for any $c_1 \in \left[\frac{2c}{3}, c\right]$ (recall that $c$ is given in the condition \eqref{eq:cond_upper_clus}). Then, 
		\if\arXivver1 
	\begin{align}
	\log g(\Delta) &\overset{(a)}{\leq} 2\Delta\left(\log n - \left(\frac{1}{2} - \delta_1\right) \cdot \left(a + b - (2 +o(1))\sqrt{\frac{ab}{M}}\right) \log n + O(1)\right) \nonumber \\
	&\overset{(b)}{\leq} 2\Delta \left(\log n - (1 - 2\delta_1)(1+c) \log n + O(1)\right) \nonumber \\
	&= -2(1 - o(1))\cdot c_1\Delta \log n,  
	\label{eq:bound_g_Delta_case_1}
	\end{align}
\else
\begin{align}
	&\log g(\Delta) \nonumber \\
	&\overset{(a)}{\leq} 2\Delta\Bigg(\log n - \left(\frac{1}{2} - \delta_1\right) \nonumber \\
	&\quad \cdot \left(a + b - (2 +o(1))\sqrt{\frac{ab}{M}}\right) \log n + O(1)\Bigg) \nonumber \\
	&\overset{(b)}{\leq} 2\Delta \left(\log n - (1 - 2\delta_1)(1+c) \log n + O(1)\right) \nonumber \\
	&= -2(1 - o(1))\cdot c_1\Delta \log n,  
	\label{eq:bound_g_Delta_case_1}
\end{align}
\fi
	where $(a)$ comes from the range $1 \leq \Delta \leq \delta_1 n$ and $(b)$ holds by plugging in the condition \eqref{eq:cond_upper_clus}. As a result, $g(\Delta)$ converges to zero exponentially fast in this case.
	
	\vspace{0.2cm}
	\item \textbf{Case 2:} When $\delta_1 n < \Delta \leq \ceil*{\frac{n}{4}}$, similarly
		\if\arXivver1
	\begin{align}
	\log g(\Delta)  &\leq 2\Delta\left(\log\left(\frac{2e}{\delta_1}\right) + \log M - \frac{1}{4} \cdot  \left(a + b - (2 +o(1))\sqrt{\frac{ab}{M}}\right) \log n\right)\nonumber \\
	&\leq 2\Delta \left(-\left(\frac{1 + c}{2}\right)\log n + O(1) \right) \nonumber \\
	& = -(1 + c) \Delta \log n , 
	\label{eq:bound_g_Delta_case_2}
	\end{align}
\else
\begin{align}
	\log g(\Delta)  &\leq 2\Delta\Bigg(\log\left(\frac{2e}{\delta_1}\right) + \log M \nonumber\\
	&\quad- \frac{1}{4} \cdot  \left(a + b - (2 +o(1))\sqrt{\frac{ab}{M}}\right) \log n\Bigg)\nonumber \\
	&\leq 2\Delta \left(-\left(\frac{1 + c}{2}\right)\log n + O(1) \right) \nonumber \\
	& = -(1 + c) \Delta \log n 
	\label{eq:bound_g_Delta_case_2}
\end{align}
\fi
	which implies $g(\Delta)$ converges to zero super-exponentially fast in this case.
\end{itemize}

Finally, assembling \eqref{eq:bound_g_Delta_case_1} and \eqref{eq:bound_g_Delta_case_2} and plugging them into \eqref{eq:bound_P_e_g_Delta} yields the upper bound on $\mathcal{P}_{e, \text{c}}(\psi_{\text{MLE}})$ as 
\begin{align*}
\mathcal{P}_{e, \text{c}}(\psi_{\text{MLE}}) &\leq 2\sum_{\Delta = 1}^{\floor*{\delta_1 n}} n^{-2(1-o(1)) c_1\Delta} + 2\sum_{\Delta = \floor{\delta_1 n}+1}^{\ceil*{n/4}} n^{-(1 + c)\Delta} \\ 
&\leq n^{-2(1 - o(1))c_1} \leq n^{-c},
\end{align*}
where the last step uses $c_1 \in \left[\frac{2c}{3}, c\right]$. This completes the proof.

\section{Proof of Theorem~\ref{the:cond_lower_clus}}
\label{sec:proof_clus_lower_bound}
Our proof technique is mainly inspired by the one of \cite[Theorem 1]{abbe2015exact}. Similar to the proof of Theorem~\ref{sec:proof_clus_upper_bound}, we still follow the assumption \eqref{eq:assumption_theta}. To begin with, recall the MLE \eqref{eq:MLE} fails to exactly recover the communities $\bm{\kappa}^*$ if and only if there exists an alternative hypothesis $\bm{x} = \{\bm{\kappa}, \bm{g}\}$  that 
\begin{equation*}
	\text{dist}_{\mathrm{c}}(\bm{\kappa}, \bm{\kappa}^*) > 0 \quad \text{and} \quad \mathcal{P}(\bm{y}|\bm{x}) \geq \mathcal{P}(\bm{y}|\bm{x}^*).
\end{equation*}
Also, recall the likelihood ratio given in \eqref{eq:likelihood_ratio}:
	\if\arXivver1
\begin{equation*}
	\frac{\mathcal{P}(\bm{y}|\bm{x})}{\mathcal{P}(\bm{y}|\bm{x}^*)} = \left(\frac{Mp(1-q)}{q(1-p)}\right)^{r - r^*} \cdot 	\frac{\mathbbm{1}(g_{ij} = g_ig_j^{-1}; \; \forall (i,j) \in \mathcal{E}_{\text{inner}})}{\mathbbm{1}(g_{ij} = g_i^*(g_j^*)^{-1}; \; \forall (i,j) \in \mathcal{E}_{\text{inner}}^*)}.
\end{equation*}
\else
\begin{equation*}
	\begin{aligned}
		\frac{\mathcal{P}(\bm{y}|\bm{x})}{\mathcal{P}(\bm{y}|\bm{x}^*)} &= \left(\frac{Mp(1-q)}{q(1-p)}\right)^{r - r^*} \\
		&\quad \cdot 	\frac{\mathbbm{1}(g_{ij} = g_ig_j^{-1}; \; \forall (i,j) \in \mathcal{E}_{\text{inner}})}{\mathbbm{1}(g_{ij} = g_i^*(g_j^*)^{-1}; \; \forall (i,j) \in \mathcal{E}_{\text{inner}}^*)}.
	\end{aligned}
\end{equation*}
\fi
Then, depending on the value of $\frac{Mp(1-q)}{q(1-p)}$, we consider the following two cases: 

\if\arXivver1
\subsection{Case 1}
\else
\vspace{0.2cm}
\subsubsection{Case 1}
\fi
We first consider when $\frac{Mp(1-q)}{q(1-p)} \geq 1$, and the MLE fails only if $r \geq r^*$.  
For ease of presentation, without specification, we denote $\mathcal{P}(\cdot)$ as the probability conditioning on the ground truth $\bm{x}^*$ and omit the subscript. Then, let us define the following events:
\begin{itemize}[leftmargin=40pt]
	\vspace{0.2cm}
	\item [$F$:] the MLE \eqref{eq:MLE} fails;
	\vspace{0.2cm}
	\item [$E_A^{(i)}$:] node $i \in S_1^*$ (resp. $S_2^*$) is connecting to more nodes in $S_2^*$ (resp. $S_1^*$) than in $S_1^*$ (resp. $S_2^*$),  i.e.,  $|\mathcal{E}(i, S_1^* \backslash i)| < |\mathcal{E}(i, S_2^*)| $ (\text{resp. $|\mathcal{E}(i, S_2^* \backslash i)| < |\mathcal{E}(i, S_1^*)|$}); 
	\vspace{0.2cm}
	\item [$E_B^{(i, g_i)}$:] node $i \in S_1^*$ (resp. $S_2^*$) satisfies $g_{ij} = g_i$, for any $(i,j) \in \mathcal{E}(i, S_2^*)$ (resp. $\mathcal{E}(i, S_1^*)$). 
	\vspace{0.2cm}
\end{itemize}
We further define 
\begin{align}
E^{(i)} &:= E_A^{(i)} \bigcap \left( \bigcup_{g_i \in \mathcal{G}_M} E_B^{(i, g_i)}\right), \nonumber \\
E_{1} &:= \bigcup_{i \in S_1^*} E^{(i)}, \label{eq:def_E_1_2} \\
 E_{2} &:= \bigcup_{i \in S_2^*} E^{(i)}. \nonumber
\end{align} 
Given the above, we have the following relation between $E_1$ and $F$:
\begin{lemma}
If $\mathcal{P}(E_1) \geq \frac{4}{5}$, then $\mathcal{P}(F) \geq \frac{3}{5}$.
\label{lemma:proof_E_1_F}
\end{lemma}
\begin{proof}
By symmetry, we have $\mathcal{P}(E_1) = \mathcal{P}(E_2) \geq \frac{4}{5}$. Then, consider the event $E = E_1 \cap E_2$ occurs, which indicates there exists a pair of nodes $i \in S_1^*$ and $j \in S_2^*$ with the two events $E^{(i)}$ and $E^{(j)}$ occur with some $g_i$ and $g_j$.  As a result, one can construct an alternative hypothesis $\bm{x}$ from $\bm{x}^*$ by swapping the community assignment of $i$ and $j$,  i.e., $i \in S_2$ and $j \in S_1$, and let the associated group elements be $g_i$ and $g_j$, then $\mathcal{P}(\bm{y}|\bm{x}) > \mathcal{P}(\bm{y}|\bm{x}^*)$, indicating $F$ occurs. Therefore, 
\begin{equation*}
\mathcal{P}(F) \geq \mathcal{P}(E_1 \cap E_2) \geq \mathcal{P}(E_1) + \mathcal{P}(E_2) - 1 \geq \frac{3}{5},
\end{equation*}
which completes the proof.
\end{proof}

To proceed, let us denote the following two quantities:
\begin{equation*}
	\gamma := \log^3 n, \quad \delta := \frac{\log n}{\log\log n}, 
\end{equation*}
and we assume them to be integers (otherwise round them to the nearest ones). Then, let $H$ be a fixed subset of $S_1^*$ of size $\frac{n}{\gamma}$, we further define the following two events:
\begin{itemize}[leftmargin=40pt]
	\vspace{0.2cm}
	\item [$F_H$:] no node in $H$ connects to at least $\delta$ number of other nodes in $H$, 
	\vspace{0.2cm}
	\item [$E_{H, A}^{(i)}$:] node $i \in H$ satisfies $|\mathcal{E}(i, S_1^* \backslash H)| + \delta \leq |\mathcal{E}(i, S_2^*)|$,
	\vspace{0.2cm}
\end{itemize}
and we denote
\begin{equation*}
	\begin{aligned}
		E_H^{(i)} &:= E_{H,A}^{(i)} \bigcap \left( \bigcup_{g_i \in \mathcal{G}_M} E_B^{(i, g_i)}\right),\\
		E_H &:= \bigcup_{i \in H} E_H^{(i)}.
	\end{aligned}
\end{equation*}
Then we have
\begin{lemma}
If $\mathcal{P}(E_H) \geq \frac{9}{10}$, then $\mathcal{P}(F) \geq \frac{3}{5}$. 
\label{lemma:P_E_H_P_F}
\end{lemma}
\begin{proof}
By definition, $E_H$ occurs when there exists a node $i \in H$ with some $g_i$ that such that $E_{H,A}^{(i)}$ and $E_{B}^{(i, g_i)}$ occur. Then one can see that $F_H \cap E_H \subseteq E_1$, as $F_H \cap E_H^{(i)}$ indicates $E^{(i)}$ occurs.  Therefore, by applying $\mathcal{P}(F_H) \geq \frac{9}{10}$ from \cite[Lemma 10]{abbe2015exact}, we have 
\begin{equation*}
\mathcal{P}(E_1) \geq \mathcal{P}(F_H \cap E_H) \geq \mathcal{P}(F_H) + \mathcal{P}(E_H) - 1  \geq \frac{4}{5}.
\end{equation*}
By further applying Lemma~\ref{lemma:proof_E_1_F} we complete the proof.
\end{proof}

Now, we focus on bounding $\mathcal{P}(E_{H})$. Importantly, it satisfies
\if\arXivver1
\begin{align}
\mathcal{P}(E_H) &= \mathcal{P}\left(\bigcup_{i \in H} E_H^{(i)}\right) = 1 - \mathcal{P}\left(\bigcap_{i \in H}\left(E_H^{(i)}\right)^c\right) \overset{(a)}{=} 1 - \prod_{i \in H} \mathcal{P}\left(\left(E_H^{(i)}\right)^c\right) \nonumber  \\
&= 1 - \prod_{i \in H} \left(1 - \mathcal{P}\left(E_H^{(i)}\right)\right) = 1 - \prod_{i \in H} \left(1 - \mathcal{P}\left(E_{H,A}^{(i)} \bigcap \left( \bigcup_{g_i \in \mathcal{G}_M} E_B^{(i, g_i)}\right)\right)\right) \nonumber \\ 
&\overset{(b)}{\geq} 1 - \prod_{i \in H} \left(1 - \mathcal{P}\left(E_{H,A}^{(i)} \bigcap E_B^{(i, 0)}\right)\right), 
\label{eq:bound_P_E_H_start}
\end{align} 
\else
\begin{align}
	\mathcal{P}(E_H) &= \mathcal{P}\left(\bigcup_{i \in H} E_H^{(i)}\right) = 1 - \mathcal{P}\left(\bigcap_{i \in H}\left(E_H^{(i)}\right)^c\right) \nonumber  \\
	&\overset{(a)}{=} 1 - \prod_{i \in H} \mathcal{P}\left(\left(E_H^{(i)}\right)^c\right) = 1 - \prod_{i \in H} \left(1 - \mathcal{P}\left(E_H^{(i)}\right)\right) \nonumber \\
	&= 1 - \prod_{i \in H} \left(1 - \mathcal{P}\left(E_{H,A}^{(i)} \bigcap \left( \bigcup_{g_i \in \mathcal{G}_M} E_B^{(i, g_i)}\right)\right)\right) \nonumber \\ 
	&\overset{(b)}{\geq} 1 - \prod_{i \in H} \left(1 - \mathcal{P}\left(E_{H,A}^{(i)} \bigcap E_B^{(i, 0)}\right)\right), 
	\label{eq:bound_P_E_H_start}
\end{align} 
\fi
where $(a)$ holds since the set of events $\left\{E_H^{(i)}\right\}_{i \in H}$ are independent, as they involve mutually disjoint sets of edge connections, and $(b)$ comes by focusing on the case of $g_i = 0$. Then, our next job is to lower bound $\mathcal{P}\left(E_{H,A}^{(i)} \bigcap E_B^{(i, 0)}\right)$. To this end, let us denote
\begin{equation}
	r_1 = |\mathcal{E}(i, S_1^* \backslash H)|, \quad r_2 = |\mathcal{E}(i, S_2^*)|.
	\label{eq:r_1_r_2_definition_v2}
\end{equation}
Notice that $r_1$ and $r_2$ are two independent binomial random variables such that 
\begin{equation}
	r_1 \sim \textrm{Binom}\left(\frac{n}{2} - \frac{n}{\gamma}, p\right), \quad  r_2  \sim \textrm{Binom}\left(\frac{n}{2}, q\right).
	\label{eq:r_1_r_2_distribution}
\end{equation}
Then, it satisfies
	\if\arXivver1
\begin{align}
	\mathcal{P}\left(E_{H,A}^{(i)} \bigcap E_B^{(i, 0)}\right) &= \sum_{r_1}\sum_{r_2} \mathcal{P}(r_1) \mathcal{P}(r_2)	\mathcal{P}\left(E_{H,A}^{(i)} \bigcap E_B^{(i, 0)} \mid r_1, r_2\right)  \nonumber \\
	&\overset{(a)}{=} \sum_{r_1}\sum_{r_2} \mathcal{P}(r_1) \mathcal{P}(r_2)	\mathcal{P}\left(E_{H,A}^{(i)} \mid r_1, r_2\right) \cdot   \mathcal{P}\left(E_B^{(i, 0)} \mid r_1, r_2\right) \nonumber \\
	&\overset{(b)}{=} \sum_{r_1}\sum_{r_2} \mathcal{P}(r_1) \mathcal{P}(r_2) \cdot \mathbbm{1}\left(r_2 \geq r_1 + \delta\right) \cdot M^{-r_2}.
	\label{eq:bound_E_H_A_E_B}
\end{align}
\else
\begin{align}
	&\mathcal{P}\left(E_{H,A}^{(i)} \bigcap E_B^{(i, 0)}\right) \nonumber \\
	&= \sum_{r_1}\sum_{r_2} \mathcal{P}(r_1) \mathcal{P}(r_2)	\mathcal{P}\left(E_{H,A}^{(i)} \bigcap E_B^{(i, 0)} \mid r_1, r_2\right)  \nonumber \\
	&\overset{(a)}{=} \sum_{r_1}\sum_{r_2} \Big[\mathcal{P}(r_1) \mathcal{P}(r_2)	\mathcal{P}\left(E_{H,A}^{(i)} \mid r_1, r_2\right) \nonumber \\
	&\quad \cdot   \mathcal{P}\left(E_B^{(i, 0)} \mid r_1, r_2\right)\Big] \nonumber \\
	&\overset{(b)}{=} \sum_{r_1}\sum_{r_2} \mathcal{P}(r_1) \mathcal{P}(r_2) \cdot \mathbbm{1}\left(r_2 \geq r_1 + \delta\right) \cdot M^{-r_2}.
	\label{eq:bound_E_H_A_E_B}
\end{align}
\fi
Here, $(a)$ comes from the fact that the two events $E_{H,A}^{(i)}$ and $E_B^{(i, 0)}$ occurs independently given $r_1$ and $r_2$; $(b)$ holds as given $r_1$ and $r_2$, $E_{H,A}^{(i)}$ is deterministic, and by definition $E_B^{(i, 0)}$ occurs only if {\color{black} $g_{ij} = I$ } for all $(i,j) \in \mathcal{E}(i, S_2^*)$, and each {\color{black} $g_{ij} = I$ } occurs independent with probability $M^{-1}$. Following this way with some heavy calculation, we obtain
\begin{lemma}
Given $p = \frac{a\log n}{n}$ and $q = \frac{b \log n}{n}$, \eqref{eq:bound_E_H_A_E_B} satisfies
	\if\arXivver1 
\begin{equation}
	\log \mathcal{P}\left(E_{H,A}^{(i)} \bigcap E_B^{(i, 0)}\right) \geq \left(-\frac{a + b}{2} + \sqrt{\frac{ab}{M}}\right) \log n - o(\log n).
	\label{eq:bound_P_E_H_A_B_final}
\end{equation}
\else
\begin{equation}
	\begin{aligned}
		\log \mathcal{P}\left(E_{H,A}^{(i)} \bigcap E_B^{(i, 0)}\right) &\geq \left(-\frac{a + b}{2} + \sqrt{\frac{ab}{M}}\right) \log n \\
		&\quad- o(\log n).
	\end{aligned}
	\label{eq:bound_P_E_H_A_B_final}
\end{equation}
\fi
\label{lemma:bound_P_E_H_A_B}
\end{lemma}
The proof of Lemma~\ref{lemma:bound_P_E_H_A_B} is deferred to Appendix~\ref{sec:bound_P_E_H_A_B_proof}. As a result, plugging \eqref{eq:bound_P_E_H_A_B_final} into \eqref{eq:bound_P_E_H_start} yields:
\begin{equation}
	\begin{aligned}
	\mathcal{P}(E_H) &\geq 1 - \prod_{i \in H} \left(1 - \mathcal{P}\left(E_{H,A}^{(i)} \bigcap E_B^{(i, 0)}\right)\right) \\
	&\overset{(a)}{\geq} 1 - \prod_{i \in H} \exp\left(- \mathcal{P}\left(E_{H,A}^{(i)} \bigcap E_B^{(i, 0)}\right)\right) \\
	&\overset{(b)}{=} 1 - \exp\left(- |H| \cdot \mathcal{P}\left(E_{H,A}^{(i)} \bigcap E_B^{(i, 0)}\right)\right) \\
	&\geq 1 - \exp\left( - n^{-\frac{a + b}{2} + \sqrt{\frac{ab}{M}} + 1 - o(1)}\right) \\
	&\overset{(c)}{\geq} \frac{9}{10}, 
\end{aligned}
\label{eq:bound_E_H_final}
\end{equation}
which holds as $n$ is sufficiently large. Here, $(a)$ holds by using $1 - x \leq e^{-x}, \; \forall x \in \mathbb{R}$; $(b)$ comes from the fact that $\mathcal{P}\left(E_{H,A}^{(i)} \bigcap E_B^{(i, 0)}\right)$ should be identical to all $i \in H$; $(c)$ applies the condition in \eqref{eq:cond_lower_clus}. Therefore, by further applying Lemma~\ref{lemma:P_E_H_P_F} we complete the proof for this case.

\if\arXivver1
\subsection{Case 2}
\else
\vspace{0.2cm}
\subsubsection{Case 2}
\fi 
Now we consider $\frac{Mp(1-q)}{q(1-p)} < 1$ where MLE fails only if $r \leq r^*$.  The proof is very similar to Case 1 where most of the result still applies here. Therefore, we will skip those repeated analyses. To begin with, we reuse all the notations of events in Case 1 but redefine the following two:
\begin{itemize}[leftmargin=40pt]
	\vspace{0.2cm}
	\item [$E_{A}^{(i)}:$] node $i \in S_1^*$ (resp. $S_2^*$) is connecting to more nodes in $S_1^*$ (resp. $S_2^*$) than in $S_2^*$ (resp. $S_1^*$) i.e., $|\mathcal{E}(i, S_1^* \backslash i)| > |\mathcal{E}(i, S_2^*)|$ (resp. $|\mathcal{E}(i, S_2^* \backslash i)| < |\mathcal{E}(i, S_1^*)|$),
	\vspace{0.2cm}
	\item [$E_{H, A}^{(i)}:$] node $i \in H$ satisfies $|\mathcal{E}(i, S_1^* \backslash H)| > |\mathcal{E}(i, S_2^*)|$.
	\vspace{0.2cm}
\end{itemize}
where $H$ is still the fixed subset of $S_1^*$ of size $\frac{n}{\gamma}$ and $\gamma = \log^3 n$. In this way, Lemmas~\ref{lemma:proof_E_1_F} and \ref{lemma:P_E_H_P_F} still hold in this case.\footnote{A slight difference is that we no longer need the event $F_H$, and therefore $E_{H}$ implies $E_1$ directly.} Furthermore, we follow the steps in \eqref{eq:bound_E_H_A_E_B} and provide a lower bound on $\mathcal{P}\left(E_{H,A}^{(i)} \bigcap E_B^{(i, 0)}\right)$ similar to Lemma~\ref{lemma:bound_P_E_H_A_B}. To this end, by using $r_1$ and $r_2$ defined in \eqref{eq:r_1_r_2_definition_v2}, we have 
	\if\arXivver1
\begin{align}
	\mathcal{P}\left(E_{H,A}^{(i)} \bigcap E_B^{(i, 0)}\right) &= \sum_{r_1}\sum_{r_2} \mathcal{P}(r_1) \mathcal{P}(r_2)	\mathcal{P}\left(E_{H,A}^{(i)} \bigcap E_B^{(i, 0)} \mid r_1, r_2\right)  \nonumber \\
	&= \sum_{r_1}\sum_{r_2} \mathcal{P}(r_1) \mathcal{P}(r_2)	\mathcal{P}\left(E_{H,A}^{(i)} \mid r_1, r_2\right) \cdot   \mathcal{P}\left(E_B^{(i, 0)} \mid r_1, r_2\right) \nonumber \\
	&= \sum_{r_1}\sum_{r_2} \mathcal{P}(r_1) \mathcal{P}(r_2) \cdot \mathbbm{1}\left(r_1 > r_2\right) \cdot M^{-r_2} \nonumber \\
	&= \sum_{r_2} \mathcal{P}(r_2) \mathcal{P}(r_1 > r_2) \cdot M^{-r_2}.
	\label{eq:bound_E_H_A_E_B_v2}
\end{align}
\else
\begin{align}
	&\mathcal{P}\left(E_{H,A}^{(i)} \bigcap E_B^{(i, 0)}\right) \nonumber \\
	&= \sum_{r_1}\sum_{r_2} \mathcal{P}(r_1) \mathcal{P}(r_2)	\mathcal{P}\left(E_{H,A}^{(i)} \bigcap E_B^{(i, 0)} \mid r_1, r_2\right)  \nonumber \\
	&= \sum_{r_1}\sum_{r_2} \mathcal{P}(r_1) \mathcal{P}(r_2)	\mathcal{P}\left(E_{H,A}^{(i)} \mid r_1, r_2\right) \cdot   \mathcal{P}\left(E_B^{(i, 0)} \mid r_1, r_2\right) \nonumber \\
	&= \sum_{r_1}\sum_{r_2} \mathcal{P}(r_1) \mathcal{P}(r_2) \cdot \mathbbm{1}\left(r_1 > r_2\right) \cdot M^{-r_2} \nonumber \\
	&= \sum_{r_2} \mathcal{P}(r_2) \mathcal{P}(r_1 > r_2) \cdot M^{-r_2}.
	\label{eq:bound_E_H_A_E_B_v2}
\end{align}
\fi
Then, similar to Lemma~\ref{lemma:bound_P_E_H_A_B}, we obtain
\begin{lemma}
	Given $p = \frac{a\log n}{n}$ and $q = \frac{b \log n}{n}$, \eqref{eq:bound_E_H_A_E_B_v2} satisfies
	\begin{equation}
		\begin{aligned}
		\log \mathcal{P}\left(E_{H,A}^{(i)} \bigcap E_B^{(i, 0)}\right) &\geq \left(-\frac{a + b}{2} + \sqrt{\frac{ab}{M}}\right) \log n  - o(\log n).
		\end{aligned}
		\label{eq:bound_P_E_H_A_B_final_v2}
	\end{equation}
	\label{lemma:bound_P_E_H_A_B_v2}
\end{lemma}
The proof of Lemma~\ref{lemma:bound_P_E_H_A_B_v2} is deferred to Appendix~\ref{sec:proof_bound_P_E_H_A_B_v2}, which is highly similar to the one of Lemma~\ref{lemma:bound_P_E_H_A_B}. As a result, by following the same steps as in \eqref{eq:bound_E_H_final} we complete the proof.

\section{Proof of Theorem~\ref{the:cond_upper_rot}}
\label{sec:proof_cond_upper_rot}

To begin with, we follow the same path of proving Theorem~\ref{the:info_upper_bound_1} in Section~\ref{sec:proof_clus_upper_bound} by writing out the probability that the MLE~\eqref{eq:MLE} fails to exactly recover the group elements as
	\if\arXivver1
\begin{align}
	\mathcal{P}_{e, \text{g}}(\psi_{\text{MLE}}) &= \mathcal{P}\left(\exists \bm{x} = \{\bm{\kappa}, \bm{g}\}: \text{dist}_{\text{g}}(\bm{g}, \bm{g}^*) > 0,  \;	\frac{\mathcal{P}(\bm{y}|\bm{x})}{\mathcal{P}(\bm{y}|\bm{x}^*)}  \geq 1\right) \nonumber \\
	&\overset{(a)}{\leq} \mathcal{P}\left(\exists \bm{x} = \{\bm{\kappa}, \bm{g}\}: \text{dist}_{\text{c}}(\bm{\kappa}, \bm{\kappa}^*) > 0, \; \text{dist}_{\text{g}}(\bm{g}, \bm{g}^*) > 0,  \;	\frac{\mathcal{P}(\bm{y}|\bm{x})}{\mathcal{P}(\bm{y}|\bm{x}^*)}  \geq 1\right) \nonumber \\
	&+ \mathcal{P}\left(\exists \bm{x} = \{\bm{\kappa}, \bm{g}\}: \text{dist}_{\text{c}}(\bm{\kappa}, \bm{\kappa}^*) = 0, \; \text{dist}_{\text{g}}(\bm{g}, \bm{g}^*) > 0,  \;	\frac{\mathcal{P}(\bm{y}|\bm{x})}{\mathcal{P}(\bm{y}|\bm{x}^*)}  \geq 1\right)  \nonumber \\
	&\overset{(b)}{\leq} n^{-c} + \mathcal{P}\left(\exists \bm{x} = \{\bm{\kappa}^*, \bm{g}\}:  \text{dist}_{\text{g}}(\bm{g}, \bm{g}^*) > 0,  \;	\frac{\mathcal{P}(\bm{y}|\bm{x})}{\mathcal{P}(\bm{y}|\bm{x}^*)}  \geq 1\right).
	\label{eq:bound_rot_upper_1}
\end{align}
\else
\begin{align}
	&\mathcal{P}_{e, \text{g}}(\psi_{\text{MLE}}) \nonumber \\
	&= \mathcal{P}\bigg(\exists \bm{x} = \{\bm{\kappa}, \bm{g}\}: \text{dist}_{\text{g}}(\bm{g}, \bm{g}^*) > 0,  \;	\frac{\mathcal{P}(\bm{y}|\bm{x})}{\mathcal{P}(\bm{y}|\bm{x}^*)}  \geq 1\bigg) \nonumber \\
	&\overset{(a)}{\leq} \mathcal{P}\bigg(\exists \bm{x} = \{\bm{\kappa}, \bm{g}\}: \text{dist}_{\text{c}}(\bm{\kappa}, \bm{\kappa}^*) > 0, \nonumber \\
	&\quad \quad \text{dist}_{\text{g}}(\bm{g}, \bm{g}^*) > 0,  \;	\frac{\mathcal{P}(\bm{y}|\bm{x})}{\mathcal{P}(\bm{y}|\bm{x}^*)}  \geq 1\bigg) \nonumber \\
	&+ \mathcal{P}\bigg(\exists \bm{x} = \{\bm{\kappa}, \bm{g}\}: \text{dist}_{\text{c}}(\bm{\kappa}, \bm{\kappa}^*) = 0, \nonumber \\ 
	&\quad \quad \text{dist}_{\text{g}}(\bm{g}, \bm{g}^*) > 0,  \;	\frac{\mathcal{P}(\bm{y}|\bm{x})}{\mathcal{P}(\bm{y}|\bm{x}^*)}  \geq 1\bigg)  \nonumber \\
	&\overset{(b)}{\leq} n^{-c} + \mathcal{P}\bigg(\exists \bm{x} = \{\bm{\kappa}^*, \bm{g}\}:  \text{dist}_{\text{g}}(\bm{g}, \bm{g}^*) > 0,  \nonumber \\	
	&\quad\quad \frac{\mathcal{P}(\bm{y}|\bm{x})}{\mathcal{P}(\bm{y}|\bm{x}^*)}  \geq 1\bigg).
	\label{eq:bound_rot_upper_1}
\end{align}
\fi
Here, in step $(a)$ we split it into two cases whether the community assignment $\bm{\kappa}$ is correct or not, then apply the union bound; $(b)$ uses the fact that the condition
\begin{equation*}
	\frac{a + b}{2} - (1 + o(1))\sqrt{\frac{ab}{M}} > 1 + c
\end{equation*}
is satisfied and therefore from Theorem~\ref{the:info_upper_bound_1}, $\bm{\kappa}$ is incorrect w.p.~at most $n^{-c}$. It remains to consider when $\bm{\kappa} = \bm{\kappa}^*$. In this case, the likelihood ratio in \eqref{eq:likelihood_ratio} satisfies
\begin{equation*}
	\frac{\mathcal{P}(\bm{y}|\bm{x})}{\mathcal{P}(\bm{y}|\bm{x}^*)} 
	=  \frac{\mathbbm{1}(g_{ij} = g_ig_j^{-1}; \; \forall (i,j) \in \mathcal{E}^*_{\text{inner}})}{\mathbbm{1}(g_{ij} = g_i^*(g_j^*)^{-1}; \; \forall (i,j) \in \mathcal{E}_{\text{inner}}^*)},  
\end{equation*}
where we use the fact that $r = r^*$ and $\mathcal{E}_{\text{inner}}  = \mathcal{E}_{\text{inner}}^* $. 
Therefore, recall the definition $\mathcal{E}_{\text{inner}}^* = \mathcal{E}(S_1^*) \cup \mathcal{E}(S_2^*)$, for any $\bm{x} = \{\bm{\kappa}^*, \bm{g}\}$ with a wrong hypothesis $\bm{g}$,  one can see that $\mathcal{P}(\bm{y}|\bm{x}) \geq \mathcal{P}(\bm{y}|\bm{x}^*)$ occurs only if both subgraphs formed by $S_1^*$ and $S_2^*$ are disconnected (otherwise there exists an edge $(i,j) \in \mathcal{E}_{\text{inner}}^*$ with $g_{ij} \neq g_ig_j^{-1}$). Therefore, by denoting $E_1$ and $E_2$ as the events that the subgraph of $S_1^*$ and $S_2^*$ are connected respectively, we have 
	\if\arXivver1
\begin{align}
	\mathcal{P}\left(\exists \bm{x} = \{\bm{\kappa}^*, \bm{g}\}:  \text{dist}_{\text{g}}(\bm{g}, \bm{g}^*) > 0,  \;	\frac{\mathcal{P}(\bm{y}|\bm{x})}{\mathcal{P}(\bm{y}|\bm{x}^*)}  \geq 1\right) &\leq \mathcal{P}(E_1^c \cap E_2^c) \nonumber \\
	&= 1 -  \mathcal{P}(E_1 \cap E_2) \nonumber \\
	&= 1 - \mathcal{P}(E_1) \cdot \mathcal{P}(E_2) \nonumber  \\
	&\leq 1 - (1 - n^{-c+o(1)}) \cdot (1 - n^{-c+o(1)}) \nonumber \\
	&= n^{-c + o(1)}.
	\label{eq:bound_rot_upper_2}
\end{align}
\else
\begin{align}
	&\mathcal{P}\left(\exists \bm{x} = \{\bm{\kappa}^*, \bm{g}\}:  \text{dist}_{\text{g}}(\bm{g}, \bm{g}^*) > 0,  \;	\frac{\mathcal{P}(\bm{y}|\bm{x})}{\mathcal{P}(\bm{y}|\bm{x}^*)}  \geq 1\right) \nonumber \\
	&\leq \mathcal{P}(E_1^c \cap E_2^c) \nonumber \\
	&= 1 -  \mathcal{P}(E_1 \cap E_2) \nonumber \\
	&= 1 - \mathcal{P}(E_1) \cdot \mathcal{P}(E_2) \nonumber  \\
	&\leq 1 - (1 - n^{-c+o(1)}) \cdot (1 - n^{-c+o(1)}) \nonumber \\
	&= n^{-c + o(1)}.
	\label{eq:bound_rot_upper_2}
\end{align}
\fi
Here, we apply Lemma~\ref{lemma:connectivity} which suggests 
\begin{equation*}
	\mathcal{P}(E_1) = \mathcal{P}(E_2) \geq 1 - n^{-c + o(1)}, 
\end{equation*}
given the condition $\frac{a}{2} > 1 + c$ is satisfied. Plugging  \eqref{eq:bound_rot_upper_2} into \eqref{eq:bound_rot_upper_1} completes the proof.

\section{Proof of Theorem~\ref{the:cond_lower_rot}}
\label{sec:proof_cond_lower_rot}
We will show the necessity of the two conditions for the exact recovery of group elements separately. First, when $\frac{a}{2} < 1$, Lemma~\ref{lemma:connectivity} suggests that the two subgraphs formed by $S_1^*$ and $S_2^*$ are disconnected with prob.~$1 - o(1)$. As a result, when either one is disconnected, there exist at least two components in the subgraph, and furthermore, there is no way to synchronize the group elements between them but only guess and randomly draw from $\mathcal{G}_M$.  Therefore, in this case, the probability of having a wrong $\bm{g}$ estimated is at least $1 - M^{-1} \geq 1/2$, where the equality holds when $M = 2$. 

Secondly, when the condition  
\begin{equation*}
	\frac{a + b}{2} -\sqrt{\frac{ab}{M}} < 1, 
\end{equation*}
is satisfied, Theorem~\ref{the:cond_lower_clus} suggests that the estimated community assignment $\bm{\kappa}$ is inaccurate with probability at least $3/5$. When this occurs, there exists at least one pair of wrong classified vertices $(i,j)$ such that $i \in S_1^*$ and $j \in S_2^*$ but $i \in S_2$ and $j \in S_1$. As a result, the estimated group elements on $i$ and $j$, denoted by $g_i$ and $g_j$, are completely determined by those edges across $S_1^*$ and $S_2^*$, whose group transformations are randomly drawn from $\mathcal{G}_M$. Therefore, the probability of having a wrong estimated $g_i$ or $g_j$  is at least $1 - M^{-2}$. Combining this with the event of having a wrong community assignment $\bm{\kappa}$ yields the total probability of having the wrong group elements $\bm{g}$ at least $\frac{3}{5} \cdot (1 - M^{-2}) \geq 0.45$, where the equality holds when $M = 2$. This completes the proof.

\section*{Acknowledgments}
ZZ and YF acknowledge the support from NSF grant DMS-1854791 and Alfred P. Sloan foundation. We thank Yuehaw Khoo for valuable discussions.

\bibliographystyle{plain}
\bibliography{ref}

\appendix

\section{The giant connected component of Erdős–Rényi  graphs}
\label{sec:erdos_renyi_intro}
\textcolor{black}{In this section, we briefly review several classical properties of the Erdős–Rényi graph~\cite{erdHos1960evolution} that support our theoretical framework. While most of the results presented are not original, we include detailed proofs for completeness. For readers interested in a comprehensive treatment of random graph theory, we refer to the excellent book by Frieze and Karoński~\cite{frieze2016introduction}.}

To begin with, given a network of size $n$, the Erdős–Rényi model denoted by $G(n ,p)$ 
\footnote{The model originally proposed by Erdős and Rényi is slightly different: instead of connecting each pair of vertices independently, it generates the graph by randomly assigning exactly $m$ edges to the network, and the resulting model is denoted by $G(n, m)$. However, the two models $G(n,m)$ and $G(n, p)$ are closely related, see~\cite[Chapter 1]{frieze2016introduction}. } 
generates a random graph (denoted by $\mathscr{G}_n$) by connecting each pair of vertices independently with some probability $p$. In this section, our study focuses on the size of the largest connected component in $\mathscr{G}_n$, denoted by $Z_n$. 

First, when $p = \frac{a\log n}{n}$ follows the scaling $\Theta\left(\frac{\log n}{n}\right)$ for some constant $a > 0$, it is well-known that $a = 1$ serves as the sharp threshold for $\mathscr{G}_n$ being connected. Formally, we have:
\begin{lemma}[Connectivity]
For any $c > 0$,  
\begin{align*}
	\mathcal{P}( \text{ $\mathscr{G}_n$ is connected }) \geq 1 - n^{-c + o(1)}, &\quad \text{when } a > 1+c, \\
	\mathcal{P}( \text{ $\mathscr{G}_n$ is disconnected }) \geq 1 - n^{-c + o(1)}, &\quad \text{when } a < 1-c, 
\end{align*}
where the second one also assumes $c < 1$.
\label{lemma:connectivity}
\end{lemma}
\textcolor{black}{For completeness, a proof of Lemma~\ref{lemma:connectivity} is given in Appendix~\ref{sec:proof_lemma_connectivity}, which is basically built from the analysis provided in \cite[Theorem~4.1]{frieze2016introduction}.} Lemma~\ref{lemma:connectivity} indicates that the largest connected component will be the whole network w.h.p.~when $a > 1$. Furthermore, when $a < 1$, it can be shown that the largest connected component will be giant w.h.p.~such that 
\begin{equation*}
	\mathcal{P}(Z_n > (1 - o(1))n) = 1 - o(1).
	\label{eq:Z_n_bound_1}
\end{equation*}
To be concrete, \textcolor{black}{by following standard techniques commonly used in the study of random graphs (e.g., \cite[Theorem 2.14]{frieze2016introduction})\footnote{\cite[Theorem 2.14]{frieze2016introduction} focuses on the regime where $p = \Theta(1/n)$, whereas our interest lies in the denser regime $p = \Theta(\log n / n)$. Therefore, we cannot directly apply \cite[Theorem 2.14]{frieze2016introduction}, but we adopt the same proof techniques to derive the bound in \eqref{eq:Z_n_bound_3_main}. In particular, we first show that there is no component of intermediate size (i.e., between $\beta_1 \log n$ and $\beta_2 n$ for some constants $\beta_1, \beta_2 > 0$), and then we upper bound the total number of vertices in small components (those of size less than $\beta_1 \log n$.}, we can derive the following result regarding $Z_n$,}
\begin{equation}
	\mathcal{P}\left(Z_n > \left(1 -  n^{-a + \epsilon + o(1)}\right)n\right) > 1 - n^{-\epsilon}, 
	\label{eq:Z_n_bound_3}
\end{equation}
which holds for any $\epsilon < a$. \eqref{eq:Z_n_bound_3} suggests that a giant connected component exists with a probability at least $n^{-\Theta(1)}$. 

Unfortunately, \eqref{eq:Z_n_bound_3} is too loose to be applied for our analysis in Section~\ref{sec:proof_clus_upper_bound} and a tighter result is necessary. For example, we hope
\begin{equation}
	\mathcal{P}(Z_n > (1 - o(1))n) > 1 - n^{-\omega(n)},
	\label{eq:bound_z_n_desire}
\end{equation}
where $\omega(n)$ stands for some slowly growing function such that $\omega(n) \rightarrow \infty$ as $n \rightarrow \infty$. To this end, we derive the following results that meet the requirement in \eqref{eq:bound_z_n_desire}:
\begin{lemma}[No medium component when disconnected]
Given the Erdős–Rényi model $G(n, p)$ with $p = \frac{a\log n}{n}$ for some $a \in (0, 1)$, the largest connected component $Z_n$ satisfies
\begin{equation}
	\mathcal{P}\left(\beta_1 \log n \leq Z_n \leq n - \frac{\beta_2 n}{\log \log n}\right) \leq n^{-\Theta(\log n)}, 
	\label{eq:bound_Z_n_1}
\end{equation}
for any constants $\beta_1, \beta_2 > 0$. 
\label{the:bound_Z_n_1}
\end{lemma}
\begin{lemma}[No small component when disconnected]
Given the Erdős–Rényi model $G(n, p)$ with $p = \frac{a\log n}{n}$ for some $a \in (0, 1)$, the largest connected component $Z_n$ satisfies
	\begin{equation}
		\mathcal{P}\left(Z_n \leq \beta \log n \right) \leq n^{-\frac{an(1-o(1))}{2}}, 
		\label{eq:bound_Z_n_2}
	\end{equation}
	for any constant $\beta > 0$. 
	\label{the:bound_Z_n_2}
\end{lemma}

Due to the complexity, the proofs of Lemmas~\ref{the:bound_Z_n_1} and \ref{the:bound_Z_n_2} are deferred to Appendices~\ref{sec:proof_bound_Z_n_1} and \ref{sec:proof_bound_Z_n_2}, respectively. Remarkably, the proof of Lemma~\ref{the:bound_Z_n_2} is original and involves a significant amount of calculation and optimization, because traditional approaches based on the moment method do not work in this case.

Given the above, combining Lemmas~\ref{the:bound_Z_n_1} and \ref{the:bound_Z_n_2} yields a desired result like \eqref{eq:bound_z_n_desire} as follows:

\begin{theorem}[Giant connected component%
]
	Given the Erdős–Rényi model $G(n, p)$ with $p = \frac{a\log n}{n}$ for some $a \in (0, 1)$, the largest connected component $Z_n$ satisfies
	\begin{equation}
		\mathcal{P}\left(Z_n >  n - \frac{\beta n}{\log \log n}\right) > 1 - n^{- \Theta(\log n)}, 
	\end{equation}
	for any constant $\beta > 0$.
	\label{the:bound_Z_n_3}
\end{theorem}

\begin{proof}[Proof of Theorem~\ref{the:bound_Z_n_3}]
Plugging in \eqref{eq:bound_Z_n_1} and \eqref{eq:bound_Z_n_2} as 
	\if\arXivver1
\begin{align*}
	\mathcal{P}\left(Z_n >  n - \frac{\beta n}{\log \log n}\right)  &= 1 - \mathcal{P}\left(Z_n \leq n - \frac{\beta n}{\log \log n}\right) \\
	&= 1 - 	\mathcal{P}\left(\beta \log n \leq Z_n \leq n - \frac{\beta n}{\log \log n}\right)  -  \mathcal{P}\left(Z_n \leq \beta \log n \right) \\
	&> 1 - n^{-\Theta(\log n)}.
\end{align*}
\else
\begin{align*}
	&\mathcal{P}\left(Z_n >  n - \frac{\beta n}{\log \log n}\right) \nonumber \\
	&= 1 - \mathcal{P}\left(Z_n \leq n - \frac{\beta n}{\log \log n}\right) \\
	&= 1 - 	\mathcal{P}\left(\beta \log n \leq Z_n \leq n - \frac{\beta n}{\log \log n}\right)  \\
	&\quad -  \mathcal{P}\left(Z_n \leq \beta \log n \right) \\
	&> 1 - n^{-\Theta(\log n)}.
\end{align*}
\fi
This completes the proof.
\end{proof}

\subsection{Proof of Lemma~\ref{lemma:connectivity}}
\label{sec:proof_lemma_connectivity}
We first consider the case when $a > 1 + c$ for any $c > 0$. Given the graph $\mathscr{G}_n$, let $r_k$ denote the number of connected components of size $k$. The proof strategy is to sharply bound the expectation $\mathbb{E}[r_k]$ for $1 \leq k \leq n/2$ then apply Markov's inequality. That is 
\begin{align}
	\mathcal{P}(\text{$\mathscr{G}_n$ is connected}) &= \mathcal{P}(r_k = 0, k = 1, \ldots, \floor*{n/2}) \nonumber \\
	&= \mathcal{P}\left(\sum_{k = 1}^{\floor*{n/2}} r_k = 0\right) \nonumber \\
	&= 1 - \mathcal{P}\left(\sum_{k = 1}^{\floor*{n/2}} r_k \geq 1\right) \nonumber \\
	&\geq 1 - \sum_{k = 1}^{\floor*{n/2}}\mathbb{E}\left[r_k\right],
	\label{eq:bound_connected_r_k}
\end{align}
where the last step applies Markov's inequality. To determine $\mathbb{E}[r_k]$, let $\mathscr{G}_{k}$ denotes any subgraph of $\mathscr{G}_n$ with size $k$,  and let $\mathscr{G}_{k}^{c}$ be the subgraph that contains the remaining vertices, then we have
\if\arXivver1
\begin{align}
	\mathbb{E}[r_k] &\overset{(a)}{\leq}\binom{n}{k} \cdot \mathcal{P}\left( \mathscr{G}_k \text{ is connected and isolated to $\mathscr{G}_{k}^{c}$} \right) \nonumber \\
	&\overset{(b)}{\leq} \binom{n}{k} \cdot \mathcal{P}\left( \text{A spanning tree of $\mathscr{G}_k$ exists and is isolated to $\mathscr{G}_{k}^{c}$} \right) \nonumber \\
	&\overset{(c)}{\leq}  \binom{n}{k} \cdot k^{k-2} \cdot \mathcal{P}\left( \text{The spanning tree exists and is isolated to $\mathscr{G}_{k}^{c}$} \right) \nonumber \\
	&\leq \binom{n}{k} \cdot k^{k-2} \cdot p^{k-1} \cdot (1-p)^{n-k}.
	\label{eq:bound_X_k_exp}
\end{align}
\else
\begin{align}
	\mathbb{E}[r_k] &\overset{(a)}{\leq} {n \choose k} \cdot \mathcal{P}\left( \mathscr{G}_k \text{ is connected and isolated to $\mathscr{G}_{k}^{c}$} \right) \nonumber \\
	&\overset{(b)}{\leq} {n \choose k} \cdot \mathcal{P}(\text{A spanning tree of $\mathscr{G}_k$ exists} \nonumber \\
	&\quad\quad \text{and is isolated to $\mathscr{G}_{k}^{c}$}) \nonumber \\
	&\overset{(c)}{\leq}  {n \choose k} \cdot k^{k-2} \cdot \mathcal{P}( \text{The spanning tree exists} \nonumber \\
	&\quad\quad \text{and is isolated to $\mathscr{G}_{k}^{c}$} ) \nonumber \\
	&\leq  {n \choose k} \cdot k^{k-2} \cdot p^{k-1} \cdot (1-p)^{n-k}.
	\label{eq:bound_X_k_exp}
\end{align}
\fi
Here, step $(a)$ holds by applying the union bound over all subgraphs of size $k$; $(b)$ comes from the fact that $\mathscr{G}_k$ is connected only if a spanning tree of $\mathscr{G}_k$ exists; $(c)$ uses Cayley's formula that there are $k^{k-2}$ ways of choosing a tree among $\mathscr{G}_k$, and then applies the union bound over them. Now, by plugging $p = \frac{a\log n}{n} $ into \eqref{eq:bound_X_k_exp} we obtain
\begin{align}
	\mathbb{E}[r_k] &\leq \left(\frac{en}{k}\right)^k \cdot k^{k-2} \cdot \left(\frac{a\log n}{n}\right)^{k-1} \cdot e^{-\frac{a\log n}{n} \cdot k(n-k)} \nonumber \\
	&= \left(\frac{e^k}{k^2}\right) \cdot (a\log n)^{k-1} \cdot n^{-ak\left(1 - \frac{k}{n}\right) + 1}, 
	\label{eq:bound_X_k_exp_2}
\end{align}
where the first step uses $\binom{n}{k} \leq \left(\frac{en}{k}\right)^{k}$ and $(1-p) \leq e^{-p}$.
To proceed, notice that when $a > 1 + c$, \eqref{eq:bound_X_k_exp_2} satisfies 
\begin{equation*}
	\mathbb{E}[r_k] \leq n^{-(1 + c)k\left(1 - \frac{k}{n}\right)(1 -o(1)) + 1},  
\end{equation*}
which holds for any $1 \leq k \leq n/2$. Furthermore, when $k < k_0$ for some large constant $k_0$, say $k_0 = 10$, it satisfies
\begin{equation*}
	\mathbb{E}[r_k] \leq n^{-(1 + c)(k-1)\left(1 - o(1)\right) -c + o(1)} \leq n^{-c + o(1)},
\end{equation*}
and when $k \geq 10$ and $k \leq n/2$, we have
\begin{equation*}
	\mathbb{E}[r_k] \leq n^{-(1+c) \cdot \frac{k(1 - o(1))}{2} + 1} \leq n^{-5(1 + c)(1 - o(1)) + 1} \leq n^{-c + 1}.
\end{equation*}
Given the above, \eqref{eq:bound_connected_r_k} satisfies
\begin{align*}
	\mathcal{P}(\text{$\mathscr{G}_n$ is connected}) &\geq 1 - \sum_{k = 1}^9 \mathbb{E}[r_k] - \sum_{k = 10}^{\floor*{n/2}}\mathbb{E}[r_k] \\
	&\geq 1 - 9 \cdot n^{-c + o(1)} - \frac{n}{2} \cdot n^{-c + 1} \\
	&\geq 1 - n^{-c + o(1)},
\end{align*}
which completes the proof for $a > 1+c$.

\vspace{0.2cm}
Now, we consider when $a < 1 - c$ for any $c \in (0, 1)$. In this case, we have 
\begin{equation}
	\begin{aligned}
		\mathcal{P}(\text{$\mathscr{G}_n$ is disconnected}) &\geq \mathcal{P}(r_1 \geq 0) \\
	&= 1 - \mathcal{P}(r_1 = 0).
	\end{aligned}
	\label{eq:bound_disconnect}
\end{equation}
Furthermore, $\mathcal{P}(r_1 = 0)$ can be upper bounded by 
	\if\arXivver1
\begin{align}
	\mathcal{P}(r_1 = 0) &= \mathcal{P}(r_1 - \mathbb{E}[r_1] = -\mathbb{E}[r_1]) \leq \mathcal{P}\left((r_1 - \mathbb{E}[r_1])^2 = \mathbb{E}^2[r_1]\right)\nonumber \\
	&\leq \mathcal{P}\left((r_1 - \mathbb{E}[r_1])^2 \geq \mathbb{E}^2[r_1]\right) \leq \frac{\text{Var}(r_1)}{\mathbb{E}^2[r_1]},
	\label{eq:bound_r_1_0}
\end{align}
\else
\begin{align}
	\mathcal{P}(r_1 = 0) &= \mathcal{P}(r_1 - \mathbb{E}[r_1] = -\mathbb{E}[r_1]) \nonumber \\
	&\leq \mathcal{P}\left((r_1 - \mathbb{E}[r_1])^2 = \mathbb{E}^2[r_1]\right)\nonumber \\
	&\leq \mathcal{P}\left((r_1 - \mathbb{E}[r_1])^2 \geq \mathbb{E}^2[r_1]\right) \nonumber \\
	&\leq \frac{\text{Var}(r_1)}{\mathbb{E}^2[r_1]}
	\label{eq:bound_r_1_0}
\end{align}
\fi
where the last step comes from Markov's inequality. Our next job is to determine or bound $\mathbb{E}^2[r_1]$ and $\text{Var}(r_1)$. To this end, for each vertex $i$ in $\mathscr{G}_n$,  let us define the binary random variable $x_i$ as
\begin{equation*}
	x_i = \begin{cases}
	1, &\quad \text{$i$ is isolated}\\
	0, &\quad \text{otherwise},
	\end{cases}
\end{equation*}
then immediately we have $r_1 = \sum_{i = 1}^n x_i$ and
	\if\arXivver1
\begin{equation*}
	\mathbb{E}[r_1] = \sum_{i = 1}^n \mathbb{E}[x_i] = n(1-p)^{n-1} = n \cdot e^{-p(n-1)(1-o(1))}  = n \cdot n^{a(1 - o(1))} = n^{c - o(1)},
\end{equation*}
\else
\begin{equation*}
	\begin{aligned}
		\mathbb{E}[r_1] &= \sum_{i = 1}^n \mathbb{E}[x_i] = n(1-p)^{n-1} \\
		&= n \cdot e^{-p(n-1)(1-o(1))}  = n \cdot n^{a(1 - o(1))} \\
		&= n^{c - o(1)}, 
	\end{aligned}
\end{equation*}
\fi
where we plugged in $p = a\log n /n $ and $a = 1+c$. Also, 
	\if\arXivver1
\begin{equation*}
	\mathbb{E}[r_1^2] = \sum_{i = 1}^n \mathbb{E}\left[x_i^2\right]  + \sum_{i}\sum_{j \neq i} \mathbb{E}[x_i x_j] = \mathbb{E}[r_1] + n(n-1) \cdot (1-p)^{2(n-2) + 1}.
\end{equation*} 
\else
\begin{equation*}
	\begin{aligned}
		\mathbb{E}[r_1^2] &= \sum_{i = 1}^n \mathbb{E}\left[x_i^2\right]  + \sum_{i}\sum_{j \neq i} \mathbb{E}[x_i x_j] \\
		&= \mathbb{E}[r_1] + n(n-1) \cdot (1-p)^{2(n-2) + 1}.
	\end{aligned}
\end{equation*} 
\fi
Given the above, \eqref{eq:bound_r_1_0} satisfies
	\if\arXivver1
\begin{align*}
	\mathbb{P}(r_1 = 0) &\leq \frac{\mathbb{E}[r_1^2] - \mathbb{E}^2[r_1]}{\mathbb{E}^2[r_1]} \leq \frac{\mathbb{E}[r_1^2]}{\mathbb{E}^2[r_1]} - 1 \leq \frac{1}{\mathbb{E}[r_1]} + \frac{n(n-1) \cdot (1-p)^{2(n-2) + 1}}{n^2 \cdot (1 - p)^{2(n-1)}} - 1 \\
	&= n^{-c + o(1)} + \left(1 - n^{-1}\right) \cdot (1 - p)-1 = n^{-c + o(1)}.
\end{align*}
\else
\begin{align*}
	\mathbb{P}(r_1 = 0) &\leq \frac{\mathbb{E}[r_1^2] - \mathbb{E}^2[r_1]}{\mathbb{E}^2[r_1]} \leq \frac{\mathbb{E}[r_1^2]}{\mathbb{E}^2[r_1]} - 1 \\
	&\leq \frac{1}{\mathbb{E}[r_1]} + \frac{n(n-1) \cdot (1-p)^{2(n-2) + 1}}{n^2 \cdot (1 - p)^{2(n-1)}} - 1 \\
	&= n^{-c + o(1)} + \left(1 - n^{-1}\right) \cdot (1 - p)-1 \\
	&= n^{-c + o(1)}.
\end{align*}
\fi
Plugging this back into \eqref{eq:bound_disconnect} completes the proof.

\subsection{Proof of Lemma~\ref{the:bound_Z_n_1}}
\label{sec:proof_bound_Z_n_1}
Again, given the graph $\mathscr{G}_n$, we denote $r_k$ as the number of connected component with size $k$, then we have
\if\arXivver1
\begin{align}
\mathcal{P}\left(\beta_1 \log n \leq Z_n \leq n - \frac{\beta_2 n}{\log \log n}\right) &= \sum_{\beta_1 \log n \leq k \leq n - \frac{\beta_2 n}{\log \log n}} \mathcal{P}(Z_n = k) \nonumber \\
&\leq \sum_{\beta_1 \log n \leq k \leq n - \frac{\beta_2 n}{\log \log n}} \mathcal{P}(r_k \geq 1) \nonumber\\
&\leq \sum_{\beta_1 \log n \leq k \leq n - \frac{\beta_2 n}{\log \log n}} \mathbb{E}[r_k], 
\label{eq:bound_z_n_middle}
\end{align}
\else
\begin{align}
	&\mathcal{P}\left(\beta_1 \log n \leq Z_n \leq n - \frac{\beta_2 n}{\log \log n}\right) \nonumber \\
	&= \sum_{\beta_1 \log n \leq k \leq n - \frac{\beta_2 n}{\log \log n}} \mathcal{P}(Z_n = k) \nonumber \\
	&\leq \sum_{\beta_1 \log n \leq k \leq n - \frac{\beta_2 n}{\log \log n}} \mathcal{P}(r_k \geq 1) \nonumber\\
	&\leq \sum_{\beta_1 \log n \leq k \leq n - \frac{\beta_2 n}{\log \log n}} \mathbb{E}[r_k],
	\label{eq:bound_z_n_middle}
\end{align}
\fi
where the last step holds by Markov's inequality. Now, recall the result obtained in \eqref{eq:bound_X_k_exp_2} that 
\begin{align*}
	\mathbb{E}[r_k] = \left(\frac{e^k}{k^2}\right) \cdot (a\log n)^{k-1} \cdot n^{-ak\left(1 - \frac{k}{n}\right) + 1},
\end{align*}
by taking the logarithm we obtain
	\if\arXivver1 
\begin{equation}
	\log \mathbb{E}[r_k] \leq -\left[a k \left(1 - \frac{k}{n}\right) - 1\right] \log n + (k-1) \log\log n + O(k).
	\label{eq:bound_X_k_exp_3}
\end{equation}
\else
\begin{equation}
	\begin{aligned}
		\log \mathbb{E}[r_k] &\leq -\left[a k \left(1 - \frac{k}{n}\right) - 1\right] \log n \\
		&\quad+ (k-1) \log\log n + O(k).
	\end{aligned}
	\label{eq:bound_X_k_exp_3}
\end{equation}
\fi
To proceed, on the one hand, consider any $k$ that satisfies $\beta_1 \log n \leq k \leq \beta n$ for some constant $\beta > 0$, then \eqref{eq:bound_X_k_exp_3} satisfies
\begin{align*}
	\log \mathbb{E}[r_k] &\leq -ak(1-\beta)\log n  + (k-1)\log\log n + O(k) \\
	&\leq -(1-o(1)) \cdot a\beta_1(1-\beta)\log^2 n\\
	&\leq -\Theta(\log^2n),
\end{align*}
which further implies $\mathbb{E}[r_k] \leq n^{-\Theta(\log n)}$. On the other hand, let $\beta n \leq k \leq n - \frac{\beta_2n}{\log\log n}$, similarly we have 
\begin{align*}
\log \mathbb{E}[r_k] &\leq -a \cdot \beta n \cdot \frac{\beta_2}{\log\log n} \cdot \log n + n \log\log n + O(n)\\
&\leq -\Theta\left(\frac{n\log n}{\log\log n}\right),
\end{align*}
which implies $\mathbb{E}[r_k] \leq n^{-\Theta\left(\frac{n}{\log\log n}\right)}$. Finally, plugging the bounds on $\mathbb{E}[r_k]$ into \eqref{eq:bound_z_n_middle} yields
	\if\arXivver1
\begin{align*}
	\mathcal{P}\left(\beta_1 \log n \leq Z_n \leq n - \frac{\beta_2 n}{\log \log n}\right) \leq n \cdot n^{-\Theta(\log n)} = n^{-\Theta(\log n)}, 
\end{align*}
\else
\begin{align*}
	\mathcal{P}\left(\beta_1 \log n \leq Z_n \leq n - \frac{\beta_2 n}{\log \log n}\right) &\leq n \cdot n^{-\Theta(\log n)} \\
	&= n^{-\Theta(\log n)}, 
\end{align*}
\fi
which completes the proof.

\subsection{Proof of Lemma~\ref{the:bound_Z_n_2}}
\label{sec:proof_bound_Z_n_2}

First of all, the technique used for proving Theorem~\ref{the:bound_Z_n_1} does not apply in this case, since it can be shown that $\mathbb{E}[r_k] \geq 1$ when $k$ is small and then the upper bound by Markov's inequality  \eqref{eq:bound_X_k_exp} becomes trivial. Therefore, we have to upper bound $\mathcal{P}(Z_n \leq \beta \log n)$ more delicately. To this end, given the graph $\mathscr{G}_n$, let $r_k$ denotes the number of connected components of size $k$, and $k_{\text{m}} := \beta\log n$ as the largest size\footnote{Here, we assume $\beta\log n$ is an integer, otherwise round it to the nearest one.}.  Then, one can observe the following three constraints on $r_k$ that always hold:
\begin{equation}
\begin{aligned}
&(a): \quad \sum_{l = 1}^{k_\text{m}} r_k  \geq  \ceil *{\frac{n}{k_\text{m}}}, \\
&(b): \quad \sum_{k = 1}^{k_\text{m}} k \cdot r_k = n, \\
&(c): \quad r_k \geq 0, \quad k = 1, \ldots, k_\text{m}.
\end{aligned}
\label{eq:constraints_r_k}
\end{equation}
Next, in order to determine $\mathcal{P}(k \leq \beta \log n)$ by using $r_k$, for each feasible combination $\bm{r} := \{r_{k}\}_{k = 1}^{k_\text{m}}$, the following two questions need to be answered:
\begin{enumerate}[leftmargin=25pt]
	\vspace{0.1cm}
	\item [-] \textit{Q1}: what is the number of arrangements on $\{r_{k}\}_{k = 1}^{k_\text{m}}$?
	\vspace{0.1cm}
	\item [-] \textit{Q2}: what is the probability of each arrangement occurs?
	\vspace{0.1cm}
\end{enumerate}

\if\arXivver1
\vspace{0.2cm}
\paragraph{Answer for Q1.} 
\else
\subsubsection{Answer for Q1}
\fi

Let $m_{1, \bm{r}}$ denotes the number of arrangement, then we claim 
\begin{align}
m_{1, \bm{r}} &= \binom{n}{r_1, 2r_2, \ldots, k_\text{m}r_\text{max}} \cdot \prod_{k = 1}^{k_\text{m}} \left[ \binom{kr_k}{k, k, \ldots, k}\cdot \frac{1}{(r_k)!}\right]\nonumber \\
&= \frac{n!}{\prod_{k = 1}^{k_\text{m}} (k!)^{r_k} \cdot (r_k)!}.
\label{eq:claim_r_k_count}
\end{align}
To interpret \eqref{eq:claim_r_k_count},  first notice that {\color{black} there are} $$\binom{n}{r_1, 2r_2, \ldots, k_\text{m}r_\text{max}}$$ {\color{black} ways to assign }%
$kr_k$ nodes for those $r_k$ components of size $k$. Then, for each $kr_k$ nodes, $\binom{kr_k}{k, k, \ldots, k}$ further splits them into $r_k$ pieces. However, this overcounts the actual number by introducing the permutation over the $r_k$ pieces, therefore we divide it by $(r_k)!$. 

\vspace{0.2cm}
\if\arXivver1
\vspace{0.2cm}
\paragraph{Answer for Q2.} 
\else
\subsubsection{Answer for Q2}
\fi

Notice that a specific arrangement of $\{r_{k}\}_{k = 1}^{k_\text{m}}$ occurs if and only if the following two events happen: 
\begin{itemize}
\vspace{0.1cm}
    \item[-] E1: each component is connected,
    \vspace{0.1cm}
    \item[-] E2: each one is isolated to others.
    \vspace{0.1cm}
\end{itemize}
For E1, let $\mathcal{P}_k$ denotes the probability that a subgraph of size $k$ is connected, then we have 
	\if\arXivver1
\begin{align*}
	\mathcal{P}(\text{E1})  &= \prod_{k = 1}^{k_\text{m}} \left(\mathcal{P}_{k}\right)^{r_k} \leq  \prod_{k = 1}^{k_\text{m}}  \left(\mathcal{P}(\text{A spanning tree of size $k$ exists})\right)^{r_k} \\
	&\leq \prod_{k = 1}^{k_\text{m}}  \left(k^{k-2} \cdot p^{k-1}\right)^{r_k},
\end{align*}
\else
\begin{align*}
	\mathcal{P}(\text{E1})  &= \prod_{k = 1}^{k_\text{m}} \left(\mathcal{P}_{k}\right)^{r_k} \\
	&\leq  \prod_{k = 1}^{k_\text{m}}  \left(\mathcal{P}(\text{A spanning tree of size $k$ exists})\right)^{r_k} \\
	&\leq \prod_{k = 1}^{k_\text{m}}  \left(k^{k-2} \cdot p^{k-1}\right)^{r_k}
\end{align*}
\fi
where the last step comes from the union bound over the total $k^{k-2}$ number of trees as Cayley's formula suggests. For E2, we have 
\begin{equation*}
	\mathcal{P}(\text{E2}) = (1-p)^{m_{2, \bm{r}}},
\end{equation*}
where $m_{2, \bm{r}}$ stands for the number of node pairs across different components that satisfies  
\if\arXivver1 
\begin{align*}
m_{2, \bm{r}}
&\overset{(a)}{=} \sum_{k = 1}^{k_\text{m}} k^2 \cdot \left(\frac{r_k(r_k-1)}{2}\right) + \frac{1}{2}\sum_{k_1 = 1}^{k_\text{m}}\sum_{k_2 = k_1+1}^{k_\text{m}} k_1r_{k_1} \cdot k_2r_{k_2}\\
&= \frac{1}{2}\sum_{k_1 = 1}^{k_\text{m}}\sum_{k_2 = 1}^{k_\text{m}} k_1r_{k_1} \cdot k_2r_{k_2} - \frac{1}{2}\sum_{k = 1}^{k_\text{m}} r_k \cdot k^2\\
&\geq \frac{1}{2}  \left( \sum_{k = 1}^{k_\text{m}} kr_k\right)^2 - \frac{k_\text{m}}{2} \sum_{k=1}^{k_\text{m}} kr_k\\
&\overset{(b)}{=} \frac{n(n-k_\text{m})}{2},
\end{align*}
\else
\begin{align*}
	m_{2, \bm{r}}
	&\overset{(a)}{=} \sum_{k = 1}^{k_\text{m}} k^2 \cdot \left(\frac{r_k(r_k-1)}{2}\right) \\
	&\quad + \frac{1}{2}\sum_{k_1 = 1}^{k_\text{m}}\sum_{k_2 = k_1+1}^{k_\text{m}} k_1r_{k_1} \cdot k_2r_{k_2}\\
	&= \frac{1}{2}\sum_{k_1 = 1}^{k_\text{m}}\sum_{k_2 = 1}^{k_\text{m}} k_1r_{k_1} \cdot k_2r_{k_2} - \frac{1}{2}\sum_{k = 1}^{k_\text{m}} r_k \cdot k^2\\
	&\geq \frac{1}{2}  \left( \sum_{k = 1}^{k_\text{m}} kr_k\right)^2 - \frac{k_\text{m}}{2} \sum_{k=1}^{k_\text{m}} kr_k\\
	&\overset{(b)}{=} \frac{n(n-k_\text{m})}{2}
\end{align*}
\fi
where $(a)$ holds by first counting number of pairs that across the $r_k$ components of size $k$, then adding the pairs that across different $k$; (b) applies the constraint $ \sum_{k = 1}^{k_\text{m}} k \cdot r_k = n$ in \eqref{eq:constraints_r_k}. Given the above, the probability of a specific arrangement 
$\bm{r}$ occurs, denoted by $\mathcal{P}_{\bm{r}}$, is given as
\begin{align*}
\mathcal{P}_{\bm{r}} &= \mathcal{P}(\text{E1} \cap \text{E2}) = \mathcal{P}(\text{E1})\mathcal{P}(\text{E2})\\
&\leq (1-p)^{m_{2,\bm{r}}} \cdot \prod_{k = 1}^{k_\text{m}}  \left(k^{k-2} \cdot p^{k-1}\right)^{r_k},
\end{align*}
where we use the fact that E1 and E2 are independent as they involve two disjoint sets of edges.   

Now, we are able to write out $\mathcal{P}(Z_n \leq \beta \log n)$ as 
\if\arXivver1
\begin{align}
\mathcal{P}(Z_n \leq \beta \log n) &= \sum_{\bm{r}} m_{1,\bm{r}} \cdot \mathcal{P}_{\bm{r}} \nonumber \\
&\leq \sum_{\bm{r}} \left(\frac{n!}{\prod_{k = 1}^{k_\text{m}} (k!)^{r_k} \cdot (r_k)!}\right) \cdot (1-p)^{m_{2,\bm{r}}} \cdot \prod_{k = 1}^{k_\text{m}}  \left(k^{k-2} \cdot p^{k-1}\right)^{r_k} \nonumber\\
&= \sum_{\bm{r}} n!(1-p)^{m_{2, \bm{r}}}  \prod_{k = 1}^{k_\text{m}} \left(\frac{k^{k-2}p^{k-1}}{k!}\right)^{r_k} \cdot \frac{1}{(r_k)!} \nonumber \\
&\leq \sum_{\bm{r}} n!(1-p)^{m_{2, \bm{r}}}  \prod_{k = 1}^{k_\text{m}} \frac{(pk)^{r_k(k-1)}}{(r_k)!} \nonumber \\
&\leq \sum_{\bm{r}} n! \cdot e^{-p \cdot m_{2, \bm{r}}}  \cdot \left( \frac{(pk)^{n-\sum_{k = 1}^{k_\text{m}} r_k}}{\prod_{k = 1}^{k_\text{m}}(r_k)!}\right),
\label{eq:bound_k_logn_1}
\end{align}
\else
\begin{align}
	&\mathcal{P}(Z_n \leq \beta \log n) = \sum_{\bm{r}} m_{1,\bm{r}} \cdot \mathcal{P}_{\bm{r}} \nonumber \\
	&\leq \sum_{\bm{r}} \Bigg[\left(\frac{n!}{\prod_{k = 1}^{k_\text{m}} (k!)^{r_k} \cdot (r_k)!}\right) \cdot (1-p)^{m_{2,\bm{r}}} \nonumber \\
	&\quad \cdot \prod_{k = 1}^{k_\text{m}}  \left(k^{k-2} \cdot p^{k-1}\right)^{r_k}\Bigg] \nonumber\\
	&= \sum_{\bm{r}} n!(1-p)^{m_{2, \bm{r}}}  \prod_{k = 1}^{k_\text{m}} \left(\frac{k^{k-2}p^{k-1}}{k!}\right)^{r_k} \cdot \frac{1}{(r_k)!} \nonumber \\
	&\leq \sum_{\bm{r}} n!(1-p)^{m_{2, \bm{r}}}  \prod_{k = 1}^{k_\text{m}} \frac{(pk)^{r_k(k-1)}}{(r_k)!} \nonumber \\
	&\leq \sum_{\bm{r}} n! \cdot e^{-p \cdot m_{2, \bm{r}}}  \cdot \left( \frac{(pk)^{n-\sum_{k = 1}^{k_\text{m}} r_k}}{\prod_{k = 1}^{k_\text{m}}(r_k)!}\right)
	\label{eq:bound_k_logn_1}
\end{align}
\fi
where in the last step we applies $(1-p) \leq e^{-p}$ and the constraint $ \sum_{k = 1}^{k_\text{m}} k \cdot r_k = n$ in \eqref{eq:constraints_r_k}, and we remark that summation is over all combinations $\bm{r}$ that satisfy the three constraints in \eqref{eq:constraints_r_k}. To proceed, by applying $p = \frac{a\log n}{n}$ and Stirling's approximation that $n! = \sqrt{2\pi n} \left(\frac{n}{e}\right)^n \left(1 + O\left(\frac{1}{n}\right)\right)$, \eqref{eq:bound_k_logn_1} can be rewritten as 
\if\arXivver1
\begin{align}
\mathcal{P}(k \leq \beta \log n)  &\leq \sum_{\bm{r}} \sqrt{2\pi n} \left(\frac{n}{e}\right)^n \left(1 + O\left(\frac{1}{n}\right)\right) \cdot n^{-\frac{a(n-k)}{2}} \cdot \left(\frac{ak\log n}{n}\right)^{n-\sum_{k = 1}^{k_\text{m}} r_k} \cdot \frac{1}{\prod_{k = 1}^{k_\text{m}}(r_k)!} \nonumber \\
&\overset{(a)}{=} \sum_{\bm{r}}\left(\frac{\sqrt{2\pi n}}{e^n}\right) \cdot n^{-\frac{an(1 - o(1))}{2}} \cdot \left(a\beta \log^2 n\right)^{n - \sum_{k = 1}^{k_\text{m}} r_k} \cdot \left(\frac{n^{\sum_{k = 1}^{k_\text{m}} r_k}}{\prod_{k = 1}^{k_\text{m}}(r_k)!}\right) \nonumber \\
&\overset{(b)}{=} \sum_{\bm{r}} n^{-\frac{an(1-o(1))}{2}} \cdot \underbrace{\left(\frac{n^{\sum_{k = 1}^{k_\text{m}} r_k}}{\prod_{k = 1}^{k_\text{m}}(r_k)!}\right)}_{=: f(\bm{r})} \nonumber \\
&=\sum_{\bm{r}} n^{-\frac{an(1-o(1))}{2}} \cdot f(\bm{r}),
\label{eq:bound_k_logn_2}
\end{align} 
\else
\begin{align}
	&\mathcal{P}(k \leq \beta \log n) \nonumber \\
	&\leq \sum_{\bm{r}} \Bigg[\sqrt{2\pi n} \left(\frac{n}{e}\right)^n \left(1 + O\left(\frac{1}{n}\right)\right) \cdot n^{-\frac{a(n-k)}{2}} \nonumber\\
	&\quad \cdot \left(\frac{ak\log n}{n}\right)^{n-\sum_{k = 1}^{k_\text{m}} r_k} \cdot \frac{1}{\prod_{k = 1}^{k_\text{m}}(r_k)!} \Bigg]\nonumber \\
	&\overset{(a)}{=} \sum_{\bm{r}} \Bigg[ \left(\frac{\sqrt{2\pi n}}{e^n}\right) \cdot n^{-\frac{an(1 - o(1))}{2}} \cdot \left(a\beta \log^2 n\right)^{n - \sum_{k = 1}^{k_\text{m}} r_k} \nonumber\\
	&\quad \cdot \left(\frac{n^{\sum_{k = 1}^{k_\text{m}} r_k}}{\prod_{k = 1}^{k_\text{m}}(r_k)!}\right)\Bigg] \nonumber \\
	&\overset{(b)}{=} \sum_{\bm{r}} n^{-\frac{an(1-o(1))}{2}} \cdot \underbrace{\left(\frac{n^{\sum_{k = 1}^{k_\text{m}} r_k}}{\prod_{k = 1}^{k_\text{m}}(r_k)!}\right)}_{=: f(\bm{r})} \nonumber \\
	&=\sum_{\bm{r}} n^{-\frac{an(1-o(1))}{2}} \cdot f(\bm{r})
	\label{eq:bound_k_logn_2}
\end{align} 
	\fi
where in step $(a)$ we plug in $k_\text{m} = \beta \log n$ and $(b)$ holds as $n^{-\frac{an(1-o(1))}{2}}$ is dominant and thus can absorb all other terms but only one denoted as 
\begin{equation}
f(\bm{r}) := \frac{n^{\sum_{k = 1}^{k_\text{m}} r_k}}{\prod_{k = 1}^{k_\text{m}}(r_k)!}.
\label{eq:f_func_def}
\end{equation}
At first glance,  it is not clear if $f(\bm{r})$ can be also absorbed. To confirm this, we derive the following lemma that sharply bounds $f(\bm{r})$:
\begin{lemma}
Given $f(\bm{r})$ defined in \eqref{eq:f_func_def} with $k_\textup{m} = \Theta(\log n)$, it satisfies
\begin{equation}
	f(\bm{r})  \leq n^{o(n)}, 
	\label{eq:bound_f_r}
\end{equation}
which holds for  any $\bm{r} = \{r_k\}_{k = 1}^{k_\text{m}}$ that satisfies the three constraints in \eqref{eq:constraints_r_k}.
\label{lemma:bound_f_r}
\end{lemma}
For ease of presentation, the proof of Lemma~\ref{lemma:bound_f_r} is deferred to Appendix~\ref{sec:proof_bound_f_r}.  By plugging \eqref{eq:bound_f_r} into \eqref{eq:bound_k_logn_2} we get 
\begin{align*}
\mathcal{P}(k \leq \beta \log n)  &\leq \sum_{\bm{r}} n^{-\frac{an(1-o(1))}{2}} \cdot n^{o(n)} =  \sum_{\bm{r}} n^{-\frac{an(1-o(1))}{2}} \\
&\overset{(a)}{\leq} (n+1)^{k_\text{m}} \cdot n^{-\frac{an(1-o(1))}{2}} \\
&= (n +1)^{\beta \log n} \cdot n^{-\frac{an(1-o(1))}{2}} \\
&= n^{-\frac{an(1-o(1))}{2}}.
\end{align*}
Here, $(a)$ applies the union bound over the set of all possible $\bm{r}$ and the size of the set is upper bounded by $(n+1)^{k_\text{m}}$, since each $r_k$ takes values from $0$ to $n$. This completes the proof.

\subsection{Proof of Lemma~\ref{lemma:bound_f_r}}
\label{sec:proof_bound_f_r}
To begin with, we denote 
\begin{equation*}
	m := \sum_{k = 1}^{k_\text{m}} r_k, 
\end{equation*}
that represents the number of isolated components. 
Our proof strategy is to study $f(\bm{r})$ under different choices of $m$. In this way, we can immediately prove Lemma~\ref{lemma:bound_f_r} on two special cases: 
\begin{enumerate}
	\vspace{0.1cm}
	\item when $m = n$, this indicates that only $r_1 = n$ is nonzero and and all other $r_k = 0$ for $k = 2, \ldots, k_\text{m}$.  Then we have $f(\bm{r}) \leq e^{n} = n^{o(n)}$ which satisfies \eqref{eq:bound_f_r},
	\vspace{0.1cm}
	\item when $m = o(n)$, then $f(\bm{r}) \leq n^m = n^{o(n)}$, which is the desired \eqref{eq:bound_f_r}.
	\vspace{0.1cm}
\end{enumerate}

Now, we consider the remaining cases of $m$ that satisfy $m < n$ and $m = \Theta(n)$, which excludes the two special cases mentioned above. First, $f(\bm{r})$ can be upper bounded  as 
\if\arXivver1
\begin{align*}
f(\bm{r}) = \frac{n^{m}}{\prod_{k = 1}^{k_\text{m}}(r_k)!}  \overset{(a)}{\leq} \frac{n^{m}}{\prod_{k = 1}^{k_\text{m}} \frac{r_k^{r_k}}{e^{r_k-1}}} \leq n^m \Big/ \underbrace{\prod_{k = 1}^{k_\text{m}} \left(\frac{r_k}{e}\right)^{r_k}}_{=: g(\bm{r})} = \frac{n^m}{g(\bm{r})},
\end{align*}
\else
\begin{align*}
	f(\bm{r}) &= \frac{n^{m}}{\prod_{k = 1}^{k_\text{m}}(r_k)!}  \overset{(a)}{\leq} \frac{n^{m}}{\prod_{k = 1}^{k_\text{m}} \frac{r_k^{r_k}}{e^{r_k-1}}} \leq n^m \Big/ \underbrace{\prod_{k = 1}^{k_\text{m}} \left(\frac{r_k}{e}\right)^{r_k}}_{=: g(\bm{r})} \\
	&= \frac{n^m}{g(\bm{r})}, 
\end{align*}
	\fi
where step $(a)$ uses the fact that $n! \geq \frac{n^n}{e^{n-1}}$.  Then, under a fixed $m$, we focus on finding the minimum of the defined $g(\bm{r})$ under the constraints of $\bm{r}$ in \eqref{eq:constraints_r_k} and $\sum_{k = 1}^{k_\text{m}} r_k = m$. To this end, by noticing that 
\begin{equation*}
	\log g(\bm{r}) = \sum_{k = 1}^{k_\text{m}} r_k (\log r_k - 1), 
\end{equation*}
we can formulate the following convex optimization program:
\begin{equation}
	\begin{aligned}
	\min_{r_k \in \mathbb{R}_+, k = 1, \ldots, k_\text{m}} \quad &\sum_{k = 1}^{k_\text{m}} r_k (\log r_k - 1) \\
	\mathrm{s.t.} \quad & \sum_{k = 1}^{k_\text{m}} r_k = m \\
	& \sum_{k = 1}^{k_\text{m}} k r_k = n \\
	& r_k \geq 0, \; k = 1, \ldots, k_\text{m}. 
	\end{aligned}
\label{eq:opt_r_program}
\end{equation}
Notably, \eqref{eq:opt_r_program} ignores the constraint that $r_k \in \mathbb{Z}_+$ is an integer and instead extends to real number i.e.,$r_k \in \mathbb{R}_+$, which significantly facilitates our analysis. The resulting optima is still a lower bound of $g(\bm{r})$ under $r_k \in \mathbb{Z}_+$. 

We can solve \eqref{eq:opt_r_program} from its KKT conditions, where the Lagrangian is given as
\if\arXivver1
\begin{equation*}
L\left(\bm{r}, \lambda, \nu, \{\mu_k\}_{k = 1}^{k_\text{m}}\right)  = \sum_{k = 1}^{k_\text{m}} r_k (\log r_k - 1) + \lambda \left(\sum_{k = 1}^{k_\text{m}}r_k - m\right) + \nu \left( \sum_{k = 1}^{k_\text{m}} k r_k - n\right) - \sum_{k = 1}^{k_\text{m}} \mu_k r_k.
\label{eq:opt_r_lagrangian}
\end{equation*}
\else
\begin{equation*}
	\begin{aligned}
		&L\left(\bm{r}, \lambda, \nu, \{\mu_k\}_{k = 1}^{k_\text{m}}\right)  \\
		&= \sum_{k = 1}^{k_\text{m}} r_k (\log r_k - 1) + \lambda \left(\sum_{k = 1}^{k_\text{m}}r_k - m\right) \\
		&\quad + \nu \left( \sum_{k = 1}^{k_\text{m}} k r_k - n\right) - \sum_{k = 1}^{k_\text{m}} \mu_k r_k.
	\end{aligned}
	\label{eq:opt_r_lagrangian}
\end{equation*}
\fi
Then, by stationarity, the optimum $\bm{r}^* = \{r_k^*\}_{k = 1}^{k_\text{m}}$ must satisfy
	\if\arXivver1
\begin{equation}
	0 = \frac{\partial L}{\partial r_k} \bigg|_{r_k = r_k^*} = \log r_k^* + \lambda + \nu k - \mu_k, \quad \Rightarrow \quad r_k^* = e^{\nu k + \lambda  - \mu_k}, \quad k = 1, \ldots, k_\text{m}.
	\label{eq:opt_r_stationarity}
\end{equation} 
\else
\begin{equation}
	\begin{aligned}
		0 &= \frac{\partial L}{\partial r_k} \bigg|_{r_k = r_k^*} = \log r_k^* + \lambda + \nu k - \mu_k, \\
		\Rightarrow \quad r_k^* &= e^{\nu k + \lambda  - \mu_k}, \quad k = 1, \ldots, k_\text{m}.
	\end{aligned}
	\label{eq:opt_r_stationarity}
\end{equation} 
\fi
Meanwhile, from the complementary slackness we have $\mu_k r_k = 0$ and $\mu_k \geq 0$ for all $k = 1, \ldots, k_\text{m}$. Combining this fact with \eqref{eq:opt_r_stationarity} yields $\mu_k = 0$ and 
\begin{equation}
	r_k^* = c \cdot s^{k}, \quad k = 1, \ldots, k_\text{m},
	\label{eq:opt_r_sol}
\end{equation}
for some $c, s >0$.  As a result, the optima is given as 
	\if\arXivver1 
\begin{align}
	\log g(\bm{r}^*) &= \sum_{k = 1}^{k_\text{m}} r_k \left(\log c s^k - 1\right) = \sum_{k = 1}^{k_\text{m}} r_k\left( \log c  + k\log s - 1\right) \nonumber \\
	&= m\log c + n\log s - m,
	\label{eq:opt_r_optima}
\end{align}
\else
\begin{align}
	\log g(\bm{r}^*) &= \sum_{k = 1}^{k_\text{m}} r_k \left(\log c s^k - 1\right) \nonumber \\
	&= \sum_{k = 1}^{k_\text{m}} r_k\left( \log c  + k\log s - 1\right) \nonumber \\
	&= m\log c + n\log s - m 
	\label{eq:opt_r_optima}
\end{align}
\fi
where in the last step we apply the constraints in \eqref{eq:opt_r_program}. In order to determine $c$ and $s$,  plugging \eqref{eq:opt_r_sol} into the constraints of \eqref{eq:opt_r_program} yields 
\begin{equation}
\begin{aligned}
m &= \sum_{k = 1}^{k_\text{m}} r_k^*  = \sum_{k = 1}^{k_\text{m}} c \cdot s^k, \quad \Rightarrow \quad c = \frac{m(1-s)}{s(1-s^{k_\text{m}})}, \\
n &= \sum_{k = 1}^{k_\text{m}} k r_k^* = c \cdot \sum_{k = 1}^{k_\text{m}}  k s^k, \quad \Rightarrow \quad \frac{n}{m} = \frac{1}{1-s} - \frac{k_\text{m}s^{k_\text{m}}}{1 - s^{k_\text{m}}}.
\end{aligned}
\label{eq:opt_c_s_sol}
\end{equation}
However, finding a  solution of $c$ and $s$  from \eqref{eq:opt_c_s_sol} is nontrivial, and thus some approximation is necessary. To this end, let us define the following function
\begin{equation}
	h(x) := \frac{1}{1 - x} - \frac{k_\text{m} x^{k_\text{m}}}{1 - x^{k_\text{m}}}, 
\end{equation}
and $x = s$ is the solution to $h(x) = \frac{n}{m}$. Before we determine $s$, several importantly properties of $h(x)$ are listed below:
\begin{itemize}[leftmargin=20pt]
\vspace{0.2cm}
 \item   [-] P1: $h(0) = 1$, $h(1) = \frac{k_\text{m}+1}{2}$, and $\lim_{x \rightarrow \infty} h(x) = k_\text{m}$.
 \vspace{0.2cm}
 \item	 [-] P2: $h(x)$ is monotonically increasing as $x \geq 0$. 
 \vspace{0.2cm}
 \item   [-] P3: a unique solution $x= s$ to $h(x) = \frac{n}{m}$ exists for any $ \frac{n}{k_\text{m}} < m < n$. 
 \vspace{0.2cm}
\end{itemize}

Here, we provide some justifications for each property one by one. For P1, $h(0)$ and $\lim_{x \rightarrow \infty} h(x)$ can be observed immediately. The remaining $h(1) = \frac{k_\text{m}+1}{2}$ is obtained from L'Hôpital's rule. For P2,  the derivative of $h(x)$ is given as 
\if\arXivver1
\begin{align*}
h^\prime(x) &=  \frac{1}{(1-x)^2} - \frac{k^{2}_{\text{m}}x^{k_\text{m}-1}}{(1 - x^{k_\text{m}})^2} = \frac{(1 + x +\cdots + x^{k_\text{m}-1})^2 - k_\text{m}^2x^{k_\text{m}-1}}{(1 - x^{k_\text{m}})^2} \\
&= \frac{(1 + x + \cdots x^{k_\text{m}-1} - k_\text{m}x^{k_\text{m}-1})(1 + x + \cdots x^{k_\text{m}-1} + k_\text{m}x^{k_\text{m}-1})}{(1 - x^{k_\text{m}})^2}\\
&\geq 0,
\end{align*}
\else
\begin{align*}
	h^\prime(x) &=  \frac{1}{(1-x)^2} - \frac{k^{2}_{\text{m}}x^{k_\text{m}-1}}{(1 - x^{k_\text{m}})^2} \\
	&= \frac{(1 + x +\cdots + x^{k_\text{m}-1})^2 - k_\text{m}^2x^{k_\text{m}-1}}{(1 - x^{k_\text{m}})^2} \\
	&= \frac{(1 + x + \cdots x^{k_\text{m}-1} - k_\text{m}x^{k_\text{m}-1})}{(1 - x^{k_\text{m}})^2} \\
	&\quad \cdot (1 + x + \cdots x^{k_\text{m}-1} + k_\text{m}x^{k_\text{m}-1})\\
	&\geq 0 
\end{align*}
	\fi
where the last step stems from Jensen's inequality\footnote{By defining the convex function $f_0(k) = x^k$ and a random variable $y$ that uniformly takes values on $0, 1, \ldots, k_\text{m}-1$, then Jensen's inequality suggests $\mathbb{E}[f_0(y)] \geq f_0(\mathbb{E}[y])$, which is the desired inequality.} that $1 + x + \cdots x^{k_\text{m}-1} \geq k_\text{m}x^{k_\text{m}-1}$, and the equality only holds when $x = 1$, which implies $h^\prime(x) = 0$ only when $x = 1$. For P3, one can see that $1 < \frac{n}{m} < k_{\max}$ when $\frac{n}{k_\text{m}} < m < n$, which agrees with the range of $h(x)$ when $x \geq 0$, therefore the solution to $h(x) = \frac{n}{m}$ always exists and P2 suggests its uniqueness. 

\vspace{0.2cm}
Now, we can (approximately) determine the solution $x = s$ and the corresponding optima $g(\bm{r}^*)$ in \eqref{eq:opt_r_optima} for different choices of $m$ as follows:

\begin{itemize}[leftmargin=*]
	\vspace{0.2cm}
	\item \textbf{Case~1:}  When $m =  n - n_0$ for any $n_0 = o(n)$. In this case, $$\frac{n}{m} = 1 + \frac{n_0}{n} + o\left(\frac{n_0}{n}\right).$$ Notice that when $x = o(1)$, it satisfies $h(x) = 1 + x - o(x)$, therefore the solution to $h(x) = \frac{n}{m}$ is given as $$s = \frac{n_0}{n}(1 + o(1)),$$ plugging this back into \eqref{eq:opt_c_s_sol} gives 
	\begin{equation*}
		c = \frac{m(1 - s)}{s(1 - s^{k_\text{m}})} \geq \frac{nm}{n_0} \cdot (1-o(1)).
	\end{equation*}
	Therefore, the optima given in \eqref{eq:opt_r_optima} satisfies
		\if\arXivver1 
	\begin{align*}
		\log g(\bm{r}^*) &= m\log c + n\log s - m  \\
		&\geq m \log \left(\frac{nm}{n_0} \cdot (1-o(1)) \right) + n\log \left(\frac{n_0}{n} \cdot (1+o(1))\right) \\
		&= n(1 - o(1)) \cdot \log \left(\frac{n^2}{n_0} \cdot (1 - o(1))\right) - (1 + o(1)) \cdot n\log n \\
		&\geq  (1 - o(1)) n\log n.
	\end{align*}
\else
\begin{align*}
	\log g(\bm{r}^*) &= m\log c + n\log s - m  \\
	&\geq m \log \left(\frac{nm}{n_0} \cdot (1-o(1)) \right) \\
	&\quad + n\log \left(\frac{n_0}{n} \cdot (1+o(1))\right) \\
	&= n(1 - o(1)) \cdot \log \left(\frac{n^2}{n_0} \cdot (1 - o(1))\right) \\
	&\quad - (1 + o(1)) \cdot n\log n \\
	&\geq  (1 - o(1)) n\log n.
\end{align*}
	\fi
	This further leads to $g(\bm{r}^*) \geq n^{(1 - o(1))n}$ and $f(\bm{r}^*) \leq n^{o(n)}$, which is the desired result.
	
	\vspace{0.2cm}
	\item \textbf{Case~2:} When $m = n - n_0$ for $n_0 = \Theta(n)$ and $n_0 \leq n - \frac{2n}{k_\text{m}+1}$, this indicates that $$1 < \frac{n}{m} \leq \frac{k_\text{m}+1}{2}.$$ To find the solution $s$, recall that $h(1) = \frac{k_\text{m}+1}{2}$ and $h(x)$ is monotonically increasing, then we get $s \leq 1$. Also, notice that when $s \leq 1$, it satisfies
	\begin{equation*}
		h(s) = \frac{1}{1 - s} - \frac{k_\text{m}s^{k_\text{m}}}{1 - s^{k_\text{m}}} \leq \frac{1}{1 - s},
	\end{equation*}
	and $h(s) = \frac{n}{m} = \frac{1}{1 - \frac{n_0}{n}}$, this further leads to $s \geq \frac{n_0}{n}$. Given the above, we obtain the range of $s$ as $$s \in \left[\frac{n_0}{n}, 1\right].$$ Next, to determine $c$ in \eqref{eq:opt_c_s_sol}, notice that $c(s) = \frac{m(1-s)}{s(1 - s^{k_\text{m}})} $
	as a function of $s$ is monotonically decreasing when $s > 0$. To see this, again by taking the derivative we have
	\begin{equation*}
	c^\prime(s) = \frac{(k_\text{m}+1)s^{k_\text{m}} - k_\text{m}s^{k_\text{m}+1} - 1}{s^2(1 - s^{k_\text{m}})^2} \leq 0 , 
	\end{equation*}
	where the last steps holds by Jensen's inequality that $(k_\text{m}+1)s^{k_\text{m}} - k_\text{m}s^{k_\text{m}+1} - 1 \leq 0$. Therefore, we have 
	\begin{equation*}
		c \geq c(s)|_{s = 1} =  \frac{m(1-s)}{s(1 - s^{k_\text{m}})} \bigg|_{s = 1} = \frac{m}{k_\text{m}}.
	\end{equation*}
	Now, we can determine the optima in \eqref{eq:opt_r_optima} as 
		\if\arXivver1
	\begin{align*}
	\log g(\bm{r}^*) &= m\log c + n\log s - m \overset{(a)}{\geq} m \log\left(\frac{m}{k_\text{m}}\right) + n\log \left(\frac{n_0}{n}\right) - m \\
	&\overset{(b)}{\geq} m \log \left(\frac{2n}{k_\text{m}(k_\text{m}+1)}\right) - O(n)\\
	&= (1 - o(1))\cdot m\log n,
	\end{align*}
\else
\begin{align*}
	\log g(\bm{r}^*) &= m\log c + n\log s - m \\
	&\overset{(a)}{\geq} m \log\left(\frac{m}{k_\text{m}}\right) + n\log \left(\frac{n_0}{n}\right) - m \\
	&\overset{(b)}{\geq} m \log \left(\frac{2n}{k_\text{m}(k_\text{m}+1)}\right) - O(n)\\
	&= (1 - o(1))\cdot m\log n
\end{align*}
\fi
	where $(a)$ comes by plugging $c \geq \frac{m}{k_\text{m}}$ and $a \geq \frac{n_0}{n}$; (b) holds by using $m \geq \frac{2}{k_\text{m}+1}$ and $\frac{n_0}{n} = \Theta(1)$. In this way, we further obtain $f(\bm{r}^*) = \frac{n^m}{g(\bm{r}^*)} \leq n^{o(n)}$. 
	
\end{itemize}

\vspace{0.2cm}
Notably, cases 1 and 2 together cover all the scenarios when $m = \Theta(n)$, and we have shown that $f(\bm{r}) = n^{o(n)}$, which completes the proof.

\section{Proof of lemmas}

\subsection{Proof of Lemma~\ref{lemma:bound_E_2_S_1}}
\label{sec:proof_bound_E_2_S_1}
	To begin with, recall the subgraph $G(S_{11})$ with size $|S_{11}| = \frac{n}{2} - \Delta = \Theta(n)$ follows Erdős–Rényi model. Then, let $z$ denote the size of the largest connected component of $G_{S_{11}}$, 
	Theorem~\ref{the:bound_Z_n_3} suggests a giant connected component in $G(S_{11})$ exists with high probability:
	\begin{equation}
		\mathcal{P}\left(z \geq |S_{11}| - \frac{\beta n}{\log \log n}\right) \geq 1 - n^{-\Theta(\log n)}, 
		\label{eq:prob_Z_bound}
	\end{equation}  
	for some constant $\beta > 0$. In this case, let us denote
	\begin{equation*}
		E_\mathrm{giant}: \quad z \geq |S_{11}| -  \frac{\beta n}{\log \log n},
	\end{equation*}
	as the event of a giant connected component  and $E_\mathrm{giant}^c$ as its complement. Then, $\mathcal{P}_{\bm{x}^*}\left(E_{2, S_1}\mid r_1\right)$ can be written as 
			\if\arXivver1
	\begin{align}
		\mathcal{P}_{\bm{x}^*}\left(E_{2, S_1}\mid r_1\right) &= \mathcal{P}_{\bm{x}^*}\left(E_\mathrm{giant}^c\right) \mathcal{P}_{\bm{x}^*}\left(E_{2, S_1}\mid r_1, E_\mathrm{giant}^c\right) + \mathcal{P}_{\bm{x}^*}\left(E_\mathrm{giant}\right) \mathcal{P}_{\bm{x}^*}\left(E_{2, S_1}\mid r_1, E_\mathrm{giant}\right) \nonumber \\
		&\leq n^{-\Theta(\log n)} +  \min\{1, \; \mathcal{P}_{\bm{x}^*}\left(E_{2, S_1}\mid r_1, E_\mathrm{giant}\right)\},
		\label{eq:bound_E_2_S_1_last}
	\end{align}
\else
\begin{align}
	&\mathcal{P}_{\bm{x}^*}\left(E_{2, S_1}\mid r_1\right) \nonumber \\
	&= \mathcal{P}_{\bm{x}^*}\left(E_\mathrm{giant}^c\right) \mathcal{P}_{\bm{x}^*}\left(E_{2, S_1}\mid r_1, E_\mathrm{giant}^c\right) \nonumber \\
	&\quad + \mathcal{P}_{\bm{x}^*}\left(E_\mathrm{giant}\right) \mathcal{P}_{\bm{x}^*}\left(E_{2, S_1}\mid r_1, E_\mathrm{giant}\right) \nonumber \\
	&\leq n^{-\Theta(\log n)} +  \min\{1, \; \mathcal{P}_{\bm{x}^*}\left(E_{2, S_1}\mid r_1, E_\mathrm{giant}\right)\}
	\label{eq:bound_E_2_S_1_last}
\end{align}
\fi
	where in the last step we apply $\mathcal{P}_{\bm{x}^*}\left(E_\mathrm{giant}^c\right) = 1 - \mathcal{P}_{\bm{x}^*}\left(E_\mathrm{giant}\right) \leq n^{-\Theta(\log n)}$ from \eqref{eq:prob_Z_bound}.
	
	Now, we focus on bounding $\mathcal{P}_{\bm{x}^*}\left(E_{2, S_1}\mid r_1, E_\mathrm{giant}\right)$ with a giant connected component exists.  For each node $i \in S_{12}$, let us denote $r_{1, i}$ as the number of edges that it connects to $S_{11}$, and further, let $r_{1, \text{giant}, i} \leq r_{1, i}$ denote the number of edges that are connected to the giant connected component.  Then, under a fixed set of $\left\{r_{1, i}\right\}_{i = 1}^\Delta$ and $\left\{r_{1, \text{giant},i}\right\}_{i = 1}^{\Delta}$, the event $E_{2, S_1}$ occurs with probability
	\if\arXivver1
	\begin{equation}
		\mathcal{P}\left(E_{2, S_1} \mid r_1, \left\{r_{1, i}\right\}_{i = 1}^\Delta,  \left\{r_{1, \text{giant}, i}\right\}_{i = 1}^\Delta, E_\text{giant}\right) \leq \prod_{i = 1}^\Delta M^{1 - r_{1, \text{giant}, i}}.
	\label{eq:bound_E_2_S_1_giant}
	\end{equation}
\else
\begin{equation}
	\begin{aligned}
		&\mathcal{P}\left(E_{2, S_1} \mid r_1, \left\{r_{1, i}\right\}_{i = 1}^\Delta,  \left\{r_{1, \text{giant}, i}\right\}_{i = 1}^\Delta, E_\text{giant}\right) \\
		&\leq \prod_{i = 1}^\Delta M^{1 - r_{1, \text{giant}, i}}.
	\end{aligned}
	\label{eq:bound_E_2_S_1_giant}
\end{equation}
\fi
	\eqref{eq:bound_E_2_S_1_giant} can be obtained by using the same argument as in the proof of Lemma~\ref{lemma:bound_E_2_S_1_connect}. That is, for each number of $r_{1, \text{giant}, i}$ edges that connects to the giant connected component, it creates at least $r_{1, \text{giant}, i} - 1$ number of cycles, thus the corresponding probability that the cycle consistency is satisfied is upper bounded by $M^{1 - r_{1, \text{giant}, i}}$. To proceed, we have
		\if\arXivver1
	\begin{align}
		\mathcal{P}_{\bm{x}^*}\left(E_{2, S_1}\mid r_1,  \left\{r_{1, i}\right\}_{i = 1}^\Delta, E_\mathrm{giant}\right) &\leq \prod_{i = 1}^\Delta \mathbb{E}\left[M^{1 - r_{1, \text{giant}, i}} \mid r_1, \left\{r_{1, i}\right\}_{i = 1}^\Delta, E_\text{giant}\right] \nonumber \\
		&= \prod_{i = 1}^\Delta \mathbb{E}\left[M^{1 - r_{1, \text{giant}, i}} \mid r_{1, i}, E_\text{giant}\right],
		\label{eq:bound_P_E_2_S_1_r_1_i}
	\end{align}
\else
\begin{align}
	&\mathcal{P}_{\bm{x}^*}\left(E_{2, S_1}\mid r_1,  \left\{r_{1, i}\right\}_{i = 1}^\Delta, E_\mathrm{giant}\right) \nonumber \\
	&\leq \prod_{i = 1}^\Delta \mathbb{E}\left[M^{1 - r_{1, \text{giant}, i}} \mid r_1, \left\{r_{1, i}\right\}_{i = 1}^\Delta, E_\text{giant}\right] \nonumber \\
	&= \prod_{i = 1}^\Delta \mathbb{E}\left[M^{1 - r_{1, \text{giant}, i}} \mid r_{1, i}, E_\text{giant}\right] 
	\label{eq:bound_P_E_2_S_1_r_1_i}
\end{align}
\fi
	where the last step holds as $r_{1, \text{giant}, i}$ only depends on $r_{1, i}$. 
	
	Our next job is to upper bound $\mathbb{E}\left[M^{1 - r_{1, \text{giant}, i}} \mid r_{1, i}, E_\text{giant}\right]$. First, recall that $z$ is denoted as the giant connected component size. We also define $z^\prime := |S_{11}| - z$ as the size of the residual. Then,  under a fixed $r_{1, i}$, the distribution of $r_{1, \text{giant}, i}$ is given as 
	\begin{equation*}
		\mathcal{P}\left(r_{1, \text{giant}, i} = k \mid r_{1, i}, E_\text{giant}\right) = \frac{\binom{z}{k} \cdot \binom{z^\prime}{r_{1, i} - k}}{\binom{|S_{11}|}{r_{1, i}}}, 
	\end{equation*}
	where we use the fact that given $r_{1, i}$, each pair of nodes $(i,j)$ for $j \in S_{11}$ is assigned with an edge with the same probability. 
	Then, we have
		\if\arXivver1
	\begin{align}
	\mathbb{E}\left[M^{1 - r_{1, \text{giant}, i}} \mid r_{1, i}, E_\text{giant}\right] &=  \sum_{k = \max\{0, \; r_{1, i} - z^\prime\}}^{\min\{r_{1, i}, \; z\}}	\mathcal{P}\left(r_{1, \text{giant}, i} = k \mid r_{1, i}, E_\text{giant}\right) \cdot M^{1 - k}.
	\label{eq:exp_M_1_r_giant}
	\end{align}
\else
	\begin{align}
	&\mathbb{E}\left[M^{1 - r_{1, \text{giant}, i}} \mid r_{1, i}, E_\text{giant}\right] \nonumber \\
	&=  \sum_{k = \max\{0, \; r_{1, i} - z^\prime\}}^{\min\{r_{1, i}, \; z\}}	\mathcal{P}\left(r_{1, \text{giant}, i} = k \mid r_{1, i}, E_\text{giant}\right) \cdot M^{1 - k}.
	\label{eq:exp_M_1_r_giant}
\end{align}
	\fi
	It is important to note that the summation over $k$ ranges from $\max\{0, \; r_{1, i} - z^\prime\}$ to $\min\{r_{1, i}, \; z\}$ since there is at least $r_{1, i} - z^\prime$ edges connecting to the giant connected component when $r_{1, i} > z^\prime$, and also $r_{1,i}$ cannot exceed $z$. Then, in order to calculate \eqref{eq:exp_M_1_r_giant}, we consider two cases depending on the relation between $r_{1, i}$ and $z^\prime$:
	
	\begin{itemize}[leftmargin=*]
	\vspace{0.2cm}
	\item \textbf{Case 1:} When $z^\prime / r_{1, i} = o(1)$, \eqref{eq:exp_M_1_r_giant} satisfies
		\if\arXivver1
	\begin{align*}
	\mathbb{E}\left[M^{1 - r_{1, \text{giant}, i}} \mid r_{1, i}, E_\text{giant}\right] &= \sum_{k =  r_{1, i} - z^\prime}^{\min\{r_{1, i}, \; z\}}	\mathcal{P}\left(r_{1, \text{giant}, i} = k \mid r_{1, i}, E_\text{giant}\right) \cdot M^{1 - k} \\
	&\leq M^{1 - (r_{1, i} - z^\prime)} = M^{1 - (1 - o(1))r_{1, i}},
	\end{align*}
\else
	\begin{align*}
	&\mathbb{E}\left[M^{1 - r_{1, \text{giant}, i}} \mid r_{1, i}, E_\text{giant}\right] \\
	&= \sum_{k =  r_{1, i} - z^\prime}^{\min\{r_{1, i}, \; z\}}	\mathcal{P}\left(r_{1, \text{giant}, i} = k \mid r_{1, i}, E_\text{giant}\right) \cdot M^{1 - k} \\
	&\leq M^{1 - (r_{1, i} - z^\prime)} = M^{1 - (1 - o(1))r_{1, i}}
\end{align*}
	\fi
	which is a desired upper bound.
	
	\vspace{0.2cm}
	\item \textbf{Case 2:} When $z^\prime / r_{1, i} = \Omega (1)$, let us denote 
	\begin{equation*}
		f(k) := \mathcal{P}\left(r_{1, \text{giant}, i} = k \mid r_{1, i}, E_\text{giant}\right) \cdot M^{1 - k}.
	\end{equation*}
	To proceed, we study the ratio $f(k-1)/f(k)$, which is given as
		\if\arXivver1
	\begin{align}
	\frac{f(k-1)}{f(k)} = \frac{\binom{z}{k-1} \cdot \binom{z^\prime}{{r_{1, i} - k + 1}}}{\binom{z}{k} \cdot \binom{z^\prime}{{r_{1,1} - k}}} \cdot M = \frac{Mk(z^\prime - r_{1, i} + k)}{(z - k+1)(r_{1, i} - k + 1)}.
	\label{eq:ratio_f_and_k}
	\end{align}
\else
	\begin{align}
	\frac{f(k-1)}{f(k)} &= \frac{{z \choose {k-1}} \cdot {z^\prime \choose {r_{1, i} - k + 1}}}{{z \choose k} \cdot {z^\prime \choose {r_{1,1} - k}}} \cdot M \nonumber \\
	&= \frac{Mk(z^\prime - r_{1, i} + k)}{(z - k+1)(r_{1, i} - k + 1)}.
	\label{eq:ratio_f_and_k}
\end{align}
\fi
	Then, we consider $k \leq (1 - \epsilon) r_{1, i}$ for some $\epsilon = o(1)$, plugging this into \eqref{eq:ratio_f_and_k} yields
	\if\arXivver1
	\begin{align}
	\frac{f(k-1)}{f(k)} &\leq \frac{Mk(z^\prime - \epsilon r_{1, i})}{(z - k + 1)(\epsilon r_{1, i} +1)} \overset{(a)}{=} \frac{Mk \cdot z^\prime (1 - o(1))}{z(1 - o(1)) \cdot (\epsilon r_{1, i} + 1)} \nonumber \\
	&\leq \frac{Mkz^\prime}{z \epsilon r_{1, i}} \cdot (1 + o(1)) = \frac{Mz^\prime (1 - \epsilon)}{z \epsilon} (1 + o(1)) \nonumber \\
	&\leq \frac{Mz^\prime}{z\epsilon} (1 + o(1)),
	\label{eq:ratio_f_and_k_bound}
	\end{align}
\else
	\begin{align}
	\frac{f(k-1)}{f(k)} &\leq \frac{Mk(z^\prime - \epsilon r_{1, i})}{(z - k + 1)(\epsilon r_{1, i} +1)} \nonumber \\
	&\overset{(a)}{=} \frac{Mk \cdot z^\prime (1 - o(1))}{z(1 - o(1)) \cdot (\epsilon r_{1, i} + 1)} \nonumber \\
	&\leq \frac{Mkz^\prime}{z \epsilon r_{1, i}} \cdot (1 + o(1)) = \frac{Mz^\prime (1 - \epsilon)}{z \epsilon} (1 + o(1)) \nonumber \\
	&\leq \frac{Mz^\prime}{z\epsilon} (1 + o(1))
	\label{eq:ratio_f_and_k_bound}
\end{align}
\fi
	where $(a)$ comes from the assumption that $z^\prime / r_{1,i} = \Omega(1)$ and therefore $\epsilon r_{1, i} = o(z^\prime)$. 
	Now, given that the event $E_\text{giant}$ occurs, from \eqref{eq:prob_Z_bound} we have $z \geq |S_{11}| - \frac{\beta n}{\log \log n}$ and $z^\prime \leq \frac{\beta n}{\log\log n}$ for some $\beta > 0$. Then, by taking
	\begin{equation*}
		\epsilon = \Theta\left(\frac{1}{\log\log\log n}\right),
	\end{equation*}
	\eqref{eq:ratio_f_and_k_bound} satisfies
	\begin{equation*}
	\frac{f(k-1)}{f(k)} = O\left(\frac{\log\log\log n}{\log\log n}\right) = o(1), 
	\end{equation*}
	which holds for any $k \leq (1 - \epsilon) r_{1, i}$. 
	As a result, let $k_0 := (1 - \epsilon)r_{1, i} = (1 - o(1))r_{1, i}$ and assume it to be an integer, the summation in \eqref{eq:exp_M_1_r_giant} satisfies
	\if\arXivver1
	\begin{align*}
	\mathbb{E}\left[M^{1 - r_{1, \text{giant}, i}} \mid r_{1, i}, E_\text{giant}\right] &= \sum_{k = \min\{0, \; r_{1, i} - 1\}}^{r_{1, i}} f(k) = \sum_{k = \min\{0, \; r_{1, i} - 1\}}^{k_0 - 1} f(k) + \sum_{k = k_0}^{r_{1, i}} f(k) \\
	&= o(1) \cdot f(k_0) + \sum_{k = k_0}^{r_{1, i}} \mathcal{P}\left(r_{1, \text{giant}, i} = k \mid r_{1, i}, E_\text{giant}\right) \cdot M^{1 - k} \\
	&\leq o(1) \cdot \mathcal{P}\left(r_{1, \text{giant}, i} = k_0 \mid r_{1, i}, E_\text{giant}\right) \cdot M^{1 - k_0} + M^{1 - k_0} \\ 
	&\leq M^{1 - (1 - o(1))r_{1,i}},
	\end{align*}
\else
	\begin{align*}
	&\mathbb{E}\left[M^{1 - r_{1, \text{giant}, i}} \mid r_{1, i}, E_\text{giant}\right] = \sum_{k = \min\{0, \; r_{1, i} - 1\}}^{r_{1, i}} f(k) \\
	&= \sum_{k = \min\{0, \; r_{1, i} - 1\}}^{k_0 - 1} f(k) + \sum_{k = k_0}^{r_{1, i}} f(k) \\
	&= o(1) \cdot f(k_0) + \sum_{k = k_0}^{r_{1, i}} \mathcal{P}\left(r_{1, \text{giant}, i} = k \mid r_{1, i}, E_\text{giant}\right) \cdot M^{1 - k} \\
	&\leq o(1) \cdot \mathcal{P}\left(r_{1, \text{giant}, i} = k_0 \mid r_{1, i}, E_\text{giant}\right) \cdot M^{1 - k_0} + M^{1 - k_0} \\ 
	&\leq M^{1 - (1 - o(1))r_{1,i}}
\end{align*}
\fi
	which is the desired upper bound.
	\end{itemize}

Given the above, \eqref{eq:bound_P_E_2_S_1_r_1_i} can be written as 
\if\arXivver1
\begin{align}
\mathcal{P}_{\bm{x}^*}\left(E_{2, S_1}\mid r_1,  \left\{r_{1, i}\right\}_{i = 1}^\Delta, E_\mathrm{giant}\right)  &\leq \prod_{i = 1}^\Delta  M^{1 - (1 - o(1))r_{1,i}} = M^{\sum_{i = 1}^\Delta \left[1 - (1 - o(1))r_{1,i} \right]} \nonumber \\
&= M^{\Delta - (1 - o(1)) r_1},
\label{eq:bound_E_2_S_1_final}
\end{align}
\else
\begin{align}
	&\mathcal{P}_{\bm{x}^*}\left(E_{2, S_1}\mid r_1,  \left\{r_{1, i}\right\}_{i = 1}^\Delta, E_\mathrm{giant}\right)  \nonumber \\
	&\leq \prod_{i = 1}^\Delta  M^{1 - (1 - o(1))r_{1,i}} = M^{\sum_{i = 1}^\Delta \left[1 - (1 - o(1))r_{1,i} \right]} \nonumber \\
	&= M^{\Delta - (1 - o(1)) r_1}
	\label{eq:bound_E_2_S_1_final}
\end{align}
\fi
where the last step applies $\sum_{i = 1}^\Delta r_{1, i} = r_1$. Therefore, since the upper bound in \eqref{eq:bound_E_2_S_1_final} does not depend on the specific set $\{r_{1, i}\}_{i = 1}^\Delta$, we have
\begin{align*}
\mathcal{P}_{\bm{x}^*}\left(E_{2, S_1}\mid r_1,  E_\mathrm{giant}\right)  &\leq \mathbb{E}\left[M^{\Delta - (1 - o(1)) r_1} \mid r_1,  E_\mathrm{giant}\right] \\
&= M^{\Delta - (1 - o(1)) r_1}.
\end{align*}
Plugging this back into \eqref{eq:bound_E_2_S_1_last} completes the proof.

\subsection{Proof of Lemma~\ref{lemma:bound_P_Delta}}
\label{sec:proof_bound_P_Delta}
To begin with, plugging \eqref{eq:bound_E_2_r_1_r_2} into \eqref{eq:bound_P_Delta_middle} yields
\if\arXivver1
\begin{align}
	P_\Delta &\leq \min_{\alpha > 0}  \left(  \sum_{r_1}\sum_{r_2}\sum_{r^*} \mathcal{P}_{\bm{x}^*}(r_1)\mathcal{P}_{\bm{x}^*}(r_2)\mathcal{P}_{\bm{x}^*}(r^*) \cdot  \left(\frac{Ma(1-q)}{b(1-p)}\right)^{\alpha(r_1 + r_2 - r^*)} \cdot \left(n^{-\Theta(\log n)} + M^{2\Delta - (1 - o(1))(r_1 + r_2)}\right)  \right) \nonumber \\
	&=  \min_{\alpha > 0} \left(  \sum_{r}\sum_{r^*} \mathcal{P}_{\bm{x}^*}(r)\mathcal{P}_{\bm{x}^*}(r^*) \cdot  \left(\frac{Ma(1-q)}{b(1-p)}\right)^{\alpha(r - r^*)} \cdot \left(n^{-\Theta(\log n)} + M^{2\Delta - (1 - o(1)) r}\right)  \right) \nonumber \\
	&= \min_{\alpha > 0} \left( n^{-\Theta(\log n)} \cdot \mathbb{E}\left[\left(\frac{Ma(1-q)}{b(1-p)}\right)^{\alpha(r - r^*)}\right] + M^{2\Delta} \cdot \mathbb{E}\left[\left(\frac{Ma(1-q)}{b(1-p)}\right)^{\alpha(r - r^*)} \cdot M^{-(1 - o(1))r}\right]\right).
	\label{eq:bound_P_Delta_2}
\end{align}
\else
\begin{align}
	&P_\Delta \leq \min_{\alpha > 0}  \Bigg[  \sum_{r_1}\sum_{r_2}\sum_{r^*} \Bigg(\mathcal{P}_{\bm{x}^*}(r_1)\mathcal{P}_{\bm{x}^*}(r_2)\mathcal{P}_{\bm{x}^*}(r^*) \nonumber \\
	&\quad \cdot  \left(\frac{Ma(1-q)}{b(1-p)}\right)^{\alpha(r_1 + r_2 - r^*)} \nonumber \\
	&\quad \cdot \left(n^{-\Theta(\log n)} + M^{2\Delta - (1 - o(1))(r_1 + r_2)}\right) \Bigg) \Bigg] \nonumber \\
	&=  \min_{\alpha > 0} \Bigg[  \sum_{r}\sum_{r^*} \Bigg( \mathcal{P}_{\bm{x}^*}(r)\mathcal{P}_{\bm{x}^*}(r^*) \nonumber \\
	&\quad \cdot  \left(\frac{Ma(1-q)}{b(1-p)}\right)^{\alpha(r - r^*)} \cdot \left(n^{-\Theta(\log n)} + M^{2\Delta - (1 - o(1)) r}\right) \Bigg) \Bigg] \nonumber \\
	&= \min_{\alpha > 0} \Bigg[ n^{-\Theta(\log n)} \cdot \mathbb{E}\left[\left(\frac{Ma(1-q)}{b(1-p)}\right)^{\alpha(r - r^*)}\right] \nonumber\\
	&\quad+ M^{2\Delta} \cdot \mathbb{E}\left[\left(\frac{Ma(1-q)}{b(1-p)}\right)^{\alpha(r - r^*)} \cdot M^{-(1 - o(1))r}\right]\Bigg].
	\label{eq:bound_P_Delta_2}
\end{align}
\fi
Recall that $r$ and $r^*$ defined in \eqref{eq:def_r_in_out} are independent and follow binomial distributions as 
	\if\arXivver1
\begin{equation*}
	r \sim \mathrm{Binom}\left(2\Delta\left(\frac{n}{2}  - \Delta\right), q\right), \quad r^* \sim \mathrm{Binom}\left(2\Delta\left(\frac{n}{2}  - \Delta\right), p\right).
\end{equation*}
\else
\begin{equation*}
	\begin{aligned}
		r &\sim \mathrm{Binom}\left(2\Delta\left(\frac{n}{2}  - \Delta\right), q\right), \\
		r^* &\sim \mathrm{Binom}\left(2\Delta\left(\frac{n}{2}  - \Delta\right), p\right).
	\end{aligned}
\end{equation*}
\fi
Then we apply the following standard result:
\begin{lemma}
	Given a binomial random variable $r \sim \mathrm{Binom}(N, p)$, it satisfies
	\begin{equation*}
		\mathbb{E}\left[t^r\right] \leq e^{-pN(1-t)},
	\end{equation*}
	for any constant $t > 0$.	
\end{lemma}
\begin{proof}
	By definition,
			\if\arXivver1
	\begin{align*}
		\mathbb{E}[t^x] &= \sum_{k = 0}^N \binom{N}{k} p^k (1 - p)^{N-k} \cdot t^k = \sum_{k = 0}^N \binom{N}{k} (pt)^k (1- p)^{N-k} \\
		&= (1- p + pt)^N = (1 - (1-t)p)^N\\
		&\leq e^{-(1-t)pN}, 
	\end{align*}
\else
\begin{align*}
	\mathbb{E}[t^x] &= \sum_{k = 0}^N \binom{N}{k} p^k (1 - p)^{N-k} \cdot t^k \\
	&= \sum_{k = 0}^N \binom{N}{k} (pt)^k (1- p)^{N-k} \\
	&= (1- p + pt)^N \\
	&= (1 - (1-t)p)^N\\
	&\leq e^{-(1-t)pN}, 
\end{align*}
\fi
	where the last step uses the $1 - x \leq e^{-x}$ for any $x \in \mathbb{R}$.
\end{proof}

As a result, by denoting $N_\Delta := 2\Delta\left(\frac{n}{2}  - \Delta\right)$, the expectations in \eqref{eq:bound_P_Delta_2} satisfy
	\if\arXivver1
\begin{align}
	\mathbb{E}\left[\left(\frac{Ma(1-q)}{b(1-p)}\right)^{\alpha(r - r^*)}\right]  &= \mathbb{E}\left[\left(\frac{Ma(1-q)}{b(1-p)}\right)^{\alpha r}\right] \cdot \mathbb{E}\left[\left(\frac{Ma(1-q)}{b(1-p)}\right)^{-\alpha r^*}\right] \nonumber \\
	&\leq \exp\left(-qN_\Delta\left(1 - \left(\frac{Ma(1-q)}{b(1-p)}\right)^\alpha\right)\right) \cdot \exp\left(-pN_\Delta\left(1 - \left(\frac{Ma(1-q)}{b(1-p)}\right)^{-\alpha}\right)\right) \nonumber \\
	&= \exp\left(-\frac{N_\Delta\log n}{n}\left(a + b - a\left(\frac{Ma(1-q)}{b(1-p)}\right)^{-\alpha} - b\left(\frac{Ma(1-q)}{b(1-p)}\right)^{\alpha} \right)\right) \nonumber \\
	&= \exp\left(-\left(\frac{N_\Delta\log n}{n}\right) f_1(\alpha)\right),
	\label{eq:f_1_func}
\end{align}
\else
\begin{align}
	&\mathbb{E}\left[\left(\frac{Ma(1-q)}{b(1-p)}\right)^{\alpha(r - r^*)}\right]  \nonumber \\
	&= \mathbb{E}\left[\left(\frac{Ma(1-q)}{b(1-p)}\right)^{\alpha r}\right] \cdot \mathbb{E}\left[\left(\frac{Ma(1-q)}{b(1-p)}\right)^{-\alpha r^*}\right] \nonumber \\
	&\leq \exp\left(-qN_\Delta\left(1 - \left(\frac{Ma(1-q)}{b(1-p)}\right)^\alpha\right)\right) \nonumber \\
	&\quad \cdot \exp\left(-pN_\Delta\left(1 - \left(\frac{Ma(1-q)}{b(1-p)}\right)^{-\alpha}\right)\right) \nonumber \\
	&= \exp\left(-\frac{N_\Delta\log n}{n}\left(a + b - a\left(\frac{Ma(1-q)}{b(1-p)}\right)^{-\alpha} - b\left(\frac{Ma(1-q)}{b(1-p)}\right)^{\alpha} \right)\right) \nonumber \\
	&= \exp\left(-\left(\frac{N_\Delta\log n}{n}\right) f_1(\alpha)\right)
	\label{eq:f_1_func}
\end{align}
\fi
where we use the fact $p = \frac{a\log n}{n}, p = \frac{b\log n}{n}$ and define 
\begin{equation*}
	f_1(\alpha) :=  a + b - a\left(\frac{Ma(1-q)}{b(1-p)}\right)^{-\alpha} - b\left(\frac{Ma(1-q)}{b(1-p)}\right)^{\alpha}.
\end{equation*}
Similarly, 
	\if\arXivver1
\begin{align}
	&\mathbb{E}\left[\left(\frac{Ma(1-q)}{b(1-p)}\right)^{\alpha(r - r^*)} \cdot M^{-(1 - o(1))r}\right] \nonumber \\
    &= \mathbb{E}\left[\left(\frac{Ma(1-q)}{b(1-p)}\right)^{\alpha r} \cdot M^{-(1 - o(1))r}\right] \cdot \mathbb{E}\left[\left(\frac{Ma(1-q)}{b(1-p)}\right)^{-\alpha r^*}\right] \nonumber \\
	&\leq \exp\left(-qN_\Delta\left(1 - \left(\frac{Ma(1-q)}{b(1-p)}\right)^\alpha \cdot M^{-(1 - o(1))}\right)\right) \cdot \exp\left(-pN_\Delta\left(1 - \left(\frac{Ma(1-q)}{b(1-p)}\right)^{-\alpha}\right)\right)\nonumber \\
	&= \exp\left(-\frac{N_\Delta\log n}{n} \left(a + b - a\left(\frac{Ma(1-q)}{b(1-p)}\right)^{-\alpha} - b\left(\frac{Ma(1-q)}{b(1-p)}\right)^\alpha \cdot M^{-(1 - o(1))}\right)\right) \nonumber\\
	&= \exp\left(-\left(\frac{N_\Delta\log n}{n}\right) f_2(\alpha)\right),
	\label{eq:bound_E_Ma_b}
\end{align}  
\else
\begin{align}
	&\mathbb{E}\left[\left(\frac{Ma(1-q)}{b(1-p)}\right)^{\alpha(r - r^*)} \cdot M^{-(1 - o(1))r}\right] \nonumber \\
	&= \mathbb{E}\left[\left(\frac{Ma(1-q)}{b(1-p)}\right)^{\alpha r} \cdot M^{-(1 - o(1))r}\right] \cdot \mathbb{E}\left[\left(\frac{Ma(1-q)}{b(1-p)}\right)^{-\alpha r^*}\right] \nonumber \\
	&\leq \exp\left(-qN_\Delta\left(1 - \left(\frac{Ma(1-q)}{b(1-p)}\right)^\alpha \cdot M^{-(1 - o(1))}\right)\right)  \nonumber \\
	&\quad \cdot \exp\left(-pN_\Delta\left(1 - \left(\frac{Ma(1-q)}{b(1-p)}\right)^{-\alpha}\right)\right)\nonumber \\
	&= \exp\Bigg(-\frac{N_\Delta\log n}{n} \Bigg(a + b - a\left(\frac{Ma(1-q)}{b(1-p)}\right)^{-\alpha} \nonumber \\
	&\quad- b\left(\frac{Ma(1-q)}{b(1-p)}\right)^\alpha \cdot M^{-(1 - o(1))}\Bigg)\Bigg) \nonumber\\
	&= \exp\left(-\left(\frac{N_\Delta\log n}{n}\right) f_2(\alpha)\right)
	\label{eq:bound_E_Ma_b}
\end{align}  
\fi
where we define
{\color{black}
\begin{equation*}
	f_2(\alpha) := a + b - a\left(\frac{Ma(1-q)}{b(1-p)}\right)^{-\alpha} - b\left(\frac{Ma(1-q)}{b(1-p)}\right)^\alpha \cdot M^{-(1 - o(1))}.
	\label{eq:def_func_f}
\end{equation*}
Notably, since $\alpha$ can be chosen arbitrarily,  $f_2(\alpha)$ satisfies 
\if\arXivver1 
\begin{equation}
	f_2(\alpha) \leq a + b  - 2\sqrt{\frac{ab}{M^{1 - o(1)}}} = a + b - (2+ o(1))\sqrt{\frac{ab}{M}} = \Theta(1), 
\label{eq:f_2_upper_bound}
\end{equation}
\else
\begin{equation*}
\begin{aligned}
    f_2(\alpha) &\leq a + b  - 2\sqrt{\frac{ab}{M^{1 - o(1)}}} \\
    &= a + b - (2+ o(1))\sqrt{\frac{ab}{M}} \\
    &= O(1), 
\end{aligned}
\end{equation*}
\fi
where the equality holds when $\alpha = \alpha_0$ such that 
	\if\arXivver1
\begin{equation*}
	a\left(\frac{Ma(1-q)}{b(1-p)}\right)^{-\alpha_0} = b\left(\frac{Ma(1-q)}{b(1-p)}\right)^{\alpha_0} \cdot M^{-(1 - o(1))} \quad \Rightarrow \quad \left(\frac{Ma(1-q)}{b(1-p)}\right)^{\alpha_0} = \sqrt{\frac{a M^{1 - o(1)}}{b}}.
\end{equation*}
\else
\begin{equation*}
	\begin{aligned}
		a\left(\frac{Ma(1-q)}{b(1-p)}\right)^{-\alpha_0} &= b\left(\frac{Ma(1-q)}{b(1-p)}\right)^{\alpha_0} \cdot M^{-(1 - o(1))} \\
		\Rightarrow \quad \left(\frac{Ma(1-q)}{b(1-p)}\right)^{\alpha_0} &= \sqrt{\frac{a M^{1 - o(1)}}{b}}.
	\end{aligned}
\end{equation*}
\fi
Then, plugging $\alpha = \alpha_0$ into $f_1(\alpha)$ yields
\begin{equation}
	f_1(\alpha_0) = -a - b + \sqrt{\frac{ab}{M^{1 - o(1)}}} + \sqrt{abM^{1 - o(1)}} = \Theta(1). 
\label{eq:f_1_upper_bound}
\end{equation}

Now, given \eqref{eq:f_1_upper_bound} and \eqref{eq:f_2_upper_bound}, by letting $\alpha = \alpha_0$ we can upper bound \eqref{eq:bound_P_Delta_2} as 
\if\arXivver1
\begin{align*}
	\mathcal{P}_{\Delta} &\leq n^{-\Theta(\log n)} \cdot \exp\left(-\left(\frac{N_\Delta\log n}{n}\right) f_1(\alpha_0)\right)+ M^{2\Delta} \cdot \exp\left(-\left(\frac{N_\Delta\log n}{n}\right) f_2(\alpha_0)\right) \\
	&= n^{-\Theta(\log n)} \cdot \exp\left(-\left(\frac{N_\Delta\log n}{n}\right) \cdot \Theta(1)\right)+ M^{2\Delta} \cdot \exp\left(-\left(\frac{N_\Delta\log n}{n}\right) f_2(\alpha_0) \right) \\
	&= (1 + o(1)) \cdot M^{2\Delta} \cdot \exp\left(-\left(\frac{N_\Delta\log n}{n}\right) f_2(\alpha_0)\right)\\
	&\leq M^{2\Delta} \cdot \exp\left(-\frac{N_\Delta\log n}{n} \left(a + b - (2 +o(1))\sqrt{\frac{ab}{M}}\right)\right), 
\end{align*}
\else
\begin{align*}
	\mathcal{P}_{\Delta} &\leq n^{-\Theta(\log n)} \cdot \exp\left(-\left(\frac{N_\Delta\log n}{n}\right) f_1(\alpha_0)\right) \\
	&\quad + M^{2\Delta} \cdot \exp\left(-\left(\frac{N_\Delta\log n}{n}\right) f_2(\alpha_0)\right) \\
	&= n^{-\Theta(\log n)} \cdot \exp\left(-\left(\frac{N_\Delta\log n}{n}\right) \cdot O(1)\right) \\
	&\quad+ M^{2\Delta} \cdot \exp\left(-\left(\frac{N_\Delta\log n}{n}\right) \cdot f_2(\alpha_0) \right) \\
	&= (1 + o(1)) \cdot M^{2\Delta} \cdot \exp\left(-\left(\frac{N_\Delta\log n}{n}\right) f_2(\alpha_0)\right)\\
	&= M^{2\Delta} \cdot \exp\left(-\frac{N_\Delta\log n}{n} \left(a + b - (2 +o(1))\sqrt{\frac{ab}{M}}\right)\right), 
\end{align*}
\fi
which completes the proof.}

\subsection{Proof of Lemma~\ref{lemma:bound_P_E_H_A_B}}
\label{sec:bound_P_E_H_A_B_proof}
We start from \eqref{eq:bound_E_H_A_E_B}, and recall that $r_1$ and $r_2$ are two binomial distributions that follow \eqref{eq:r_1_r_2_distribution}. Here, we further denote another binomial random variable $r_1^\prime$ that facilitates our analysis
\begin{equation*}
	r_1^\prime \sim \textrm{Binom}\left(\frac{n}{2}, \; p\right), 
\end{equation*}
which slightly differs from $r_1 \sim \textrm{Binom}\left(\frac{n}{2} - \frac{n}{\gamma}, p\right)$. Then, \eqref{eq:bound_E_H_A_E_B} satisfies
\if\arXivver1
\begin{align}
\mathcal{P}\left(E_{H,A}^{(i)} \bigcap E_B^{(i, 0)}\right)  &= \sum_{r_1}\sum_{r_2} \mathcal{P}(r_1) \mathcal{P}(r_2) \cdot \mathbbm{1}\left(r_2 \geq r_1 + \delta\right) \cdot M^{-r_2} \nonumber \\
&= \sum_{r_2} \mathcal{P}(r_2)  \cdot \mathcal{P}\left(r_1 \leq r_2 - \delta\right) \cdot M^{-r_2} \nonumber \\
&\overset{(a)}{\geq} \sum_{r_2} \mathcal{P}(r_2) \cdot \mathcal{P}\left(r_1^\prime \leq r_2 - \delta\right) \cdot M^{-r_2} \nonumber \\
&\geq \sum_{r_2 = \delta}^{n/2} \mathcal{P}(r_2) \cdot \mathcal{P}\left(r_1^\prime = r_2 - \delta\right) \cdot M^{-r_2}, 
\label{eq:bound_E_H_A_E_B_v4}
\end{align}
\else
\begin{align}
	&\mathcal{P}\left(E_{H,A}^{(i)} \bigcap E_B^{(i, 0)}\right) \nonumber \\
	&= \sum_{r_1}\sum_{r_2} \mathcal{P}(r_1) \mathcal{P}(r_2) \cdot \mathbbm{1}\left(r_2 \geq r_1 + \delta\right) \cdot M^{-r_2} \nonumber \\
	&= \sum_{r_2} \mathcal{P}(r_2)  \cdot \mathcal{P}\left(r_1 \leq r_2 - \delta\right) \cdot M^{-r_2} \nonumber \\
	&\overset{(a)}{\geq} \sum_{r_2} \mathcal{P}(r_2) \cdot \mathcal{P}\left(r_1^\prime \leq r_2 - \delta\right) \cdot M^{-r_2} \nonumber \\
	&\geq \sum_{r_2 = \delta}^{n/2} \mathcal{P}(r_2) \cdot \mathcal{P}\left(r_1^\prime = r_2 - \delta\right) \cdot M^{-r_2} 
	\label{eq:bound_E_H_A_E_B_v4}
\end{align}
\fi
where $(a)$ comes from the fact that $\mathcal{P}\left(r_1 \leq x\right) \geq \mathcal{P}\left(r_1^\prime \leq x\right)$ for any $x$ as $r^{\prime}$ has more trials than $r$.  To proceed, by plugging in the binomial distribution of $r_1^\prime$ and $r_2$, \eqref{eq:bound_E_H_A_E_B_v4} is given as
\if\arXivver1
\begin{align}
\mathcal{P}\left(E_{H,A}^{(i)} \bigcap E_B^{(i, 0)}\right) &\geq \sum_{k = \delta}^{n/2} \mathcal{P}(r_2 = k) \cdot \mathcal{P}(r_1^\prime = k - \delta) \cdot M^{-k} \nonumber\\
&= \sum_{k = \delta}^{n/2}  \binom{n/2}{k} \left(\frac{q}{M}\right)^{k} (1-q)^{\frac{n}{2} - k} \cdot \binom{n/2}{k-\delta} p^{k - \delta} (1 - p)^{\frac{n}{2} - k + \delta} \nonumber \\
&= \sum_{k = 0}^{n/2 - \delta} \binom{n/2}{k} \binom{n/2}{{k + \delta}}  p^k(1-p)^{\frac{n}{2} - k} \left(\frac{q}{M}\right)^{k + \delta} (1-q)^{\frac{n}{2} - k - \delta}.
\label{eq:bound_E_H_A_B_v3}
\end{align}
\else
\begin{align}
	&\mathcal{P}\left(E_{H,A}^{(i)} \bigcap E_B^{(i, 0)}\right) \nonumber \\
	&\geq \sum_{k = \delta}^{n/2} \mathcal{P}(r_2 = k) \cdot \mathcal{P}(r_1^\prime = k - \delta) \cdot M^{-k} \nonumber\\
	&= \sum_{k = \delta}^{n/2} \Bigg[ \binom{n/2}{k} \left(\frac{q}{M}\right)^{k} (1-q)^{\frac{n}{2} - k} \nonumber \\
	&\quad \cdot \binom{n/2}{k-\delta} p^{k - \delta} (1 - p)^{\frac{n}{2} - k + \delta} \Bigg]\nonumber \\
	&= \sum_{k = 0}^{n/2 - \delta} \Bigg[\binom{n/2}{k} \binom{n/2}{{k + \delta}}  p^k(1-p)^{\frac{n}{2} - k} \nonumber \\
	&\quad \cdot \left(\frac{q}{M}\right)^{k + \delta} (1-q)^{\frac{n}{2} - k - \delta}\Bigg].
	\label{eq:bound_E_H_A_B_v3}
\end{align}
\fi
Now, in order to calculate \eqref{eq:bound_E_H_A_B_v3}, let us define 
	\if\arXivver1
\begin{align}
	f(k) &:= \binom{n/2}{k}\binom{n/2}{{k + \delta} } \left(\frac{pq}{M(1 - p)(1 - q)}\right)^{k} \cdot \underbrace{[(1-p)(1-q)]^{\frac{n}{2}} \cdot \left[\frac{q}{M(1 - q)}\right]^\delta}_{=: C} \nonumber\\	
	&= C \binom{n/2}{k}\binom{n/2}{{k + \delta} } \left(\frac{pq}{M(1 - p)(1 - q)}\right)^{k}.
	\label{eq:def_f_k_bound}
\end{align}
\else
\begin{align}
	f(k) &:= \binom{n/2}{k}{n/2 \choose {k + \delta} } \left(\frac{pq}{M(1 - p)(1 - q)}\right)^{k} \nonumber\\
	&\quad \cdot \underbrace{[(1-p)(1-q)]^{\frac{n}{2}} \cdot \left[\frac{q}{M(1 - q)}\right]^\delta}_{=: C} \nonumber\\	
	&= C \binom{n/2}{k}\binom{n/2}{{k + \delta} } \left(\frac{pq}{M(1 - p)(1 - q)}\right)^{k}.
	\label{eq:def_f_k_bound}
\end{align}
\fi
Then, \eqref{eq:bound_E_H_A_B_v3} satisfies
\begin{equation}
\mathcal{P}\left(E_{H,A}^{(i)} \bigcap E_B^{(i, 0)}\right) = \sum_{k = 0}^{n/2 - \delta} f(k) \geq \max_{0 \leq k \leq \frac{n}{2} - \delta} f(k).
\label{eq:bound_P_H_A_B_max_f}
\end{equation}
Our next job is to determine $\max_{k} f(k)$. To this end, we first show the following result:
\begin{lemma}
Given $f(k)$ defined in \eqref{eq:def_f_k_bound}, it satisfies $$ k^* = \argmax\limits_{0 \leq k \leq \frac{n}{2} - \delta} f(k) = \Theta(\log n).$$
\label{lemma:argmax_k_bound}
\end{lemma}  
\begin{proof}
We consider the ratio $f(k+1) / f(k)$, which is given as
\begin{align}
	\frac{f(k+1)}{f(k)} &= \frac{\binom{n/2}{k+1}}{\binom{n/2}{k}} \cdot \frac{\binom{n/2}{ k+\delta + 1}}{\binom{n/2}{k + \delta}} \cdot \frac{pq}{M(1 - p)(1-q)}  \nonumber \\
	&= \frac{\left(\frac{n}{2} - k\right)\left(\frac{n}{2} - k - \delta\right)}{(k+1)(k + \delta +1)} \cdot \left(\frac{ab \cdot\log^2 n}{Mn^2}\right) (1+ o(1)) \nonumber  \\
	&= \frac{\left(\frac{1}{2} - \frac{k}{n}\right)\left(\frac{1}{2} - \frac{k - \delta}{n}\right)}{(k+1)(k + \delta +1)} \cdot \left(\frac{ab \cdot\log^2 n}{M}\right) (1+ o(1)),
	\label{eq:ratio_f_k_1_k}
\end{align}
where we plugged in $p = \frac{a\log n}{n}$ and $q = \frac{b\log n}{n}$. Then, when $k = \omega(\log n)$, \eqref{eq:ratio_f_k_1_k} satisfies
\begin{equation*}
	\frac{f(k+1)}{f(k)} \leq \left(\frac{\log^2 n}{k^2}\right) \cdot \left(\frac{ab}{4M}\right) \cdot (1+o(1)) = o(1), 
\end{equation*}
which indicates $k^* = O(\log n)$. Also, when $k = o(\log n)$, \eqref{eq:ratio_f_k_1_k} satisfies
\begin{equation*}
	\frac{f(k+1)}{f(k)}  \geq \frac{\log^2 n}{(k + \delta + 1)^2} \cdot \left(\frac{ab}{4M}\right) \cdot (1+o(1)) = \omega(1), 
\end{equation*}
which indicates $k^* = \Omega(\log n)$ and completes the proof.
\end{proof}

Lemma~\ref{lemma:argmax_k_bound} suggests us to focus on $f(k)$ with $k = \Theta(\log n)$. To proceed, by taking the logarithm of $f(k)$ we get
\if\arXivver1 
\begin{equation}
	\log f(k) = \log \binom{n/2}{k} + \log \binom{n/2}{k + \delta}  + k\log \left(\frac{pq}{M(1 - p)(1 - q)}\right) + \log C.
	\label{eq:log_f_k_expand}
\end{equation}
\else
\begin{equation}
	\begin{aligned}
		\log f(k) &= \log {n/2 \choose k} + \log {n/2 \choose k + \delta}  \\
		&\quad + k\log \left(\frac{pq}{M(1 - p)(1 - q)}\right) + \log C.
	\end{aligned}
	\label{eq:log_f_k_expand}
\end{equation}
\fi
Next, we consider each term separately. First, by using the binomial coefficient bound~\cite[Lemma 23.1]{frieze2016introduction} that $\binom{n}{k} \geq \frac{n^k}{k!}\left(1 - \frac{k(k-1)}{2n}\right)$, we have
\if\arXivver1
\begin{align}
\log \binom{n/2}{k} &\geq k\log\left(\frac{n}{2}\right) - \log k! + \log \left(1 - \frac{k(k-1)}{2n}\right) \nonumber \\
&\overset{(a)}{\geq}  k\log\left(\frac{n}{2}\right)  - (k+1)\log k + (k-1) - \frac{k(k-1)}{2n} \nonumber \\
&= k\left[\log \left(\frac{n}{2}\right) - \log \left(\frac{k}{e}\right)- o(1)\right],
\label{eq:bound_f_k_coeff_1}
\end{align}
\else
\begin{align}
	&\log\binom{n/2}{k} \nonumber \\
	&\geq k\log\left(\frac{n}{2}\right) - \log k! + \log \left(1 - \frac{k(k-1)}{2n}\right) \nonumber \\
	&\overset{(a)}{\geq}  k\log\left(\frac{n}{2}\right)  - (k+1)\log k + (k-1) - \frac{k(k-1)}{2n} \nonumber \\
	&= k\left[\log \left(\frac{n}{2}\right) - \log \left(\frac{k}{e}\right)- o(1)\right],
	\label{eq:bound_f_k_coeff_1}
\end{align}
	\fi
where $(a)$ applies the factorial bound that $\log k! \leq (k+1)\log k - (k-1)$ and $\log(1 + x) \leq x, \forall x  > -1$. In a similar manner, 
\begin{align}
	\log\binom{n/2}{k+\delta} &\geq (k + \delta) \left[\log\left(\frac{n}{2}\right) - \log \left(\frac{k+ \delta}{e}\right) - o(1)\right].
	\label{eq:bound_f_k_coeff_2} 
\end{align}
Also, we have
\if\arXivver1
\begin{align}
k\log \left(\frac{pq}{M(1 - p)(1 - q)}\right) &= k\log \left(\frac{ab \cdot \log^2 n}{M n^2}\right) + o(k) \nonumber\\
&=  k\left(\log \left(\frac{ab}{M}\right)  + 2\log \log n - 2\log n + o(1) \right), 
\label{eq:bound_f_k_coeff_3}
\end{align}
\else
\begin{align}
	&k\log \left(\frac{pq}{M(1 - p)(1 - q)}\right) \nonumber \\
	&= k\log \left(\frac{ab \cdot \log^2 n}{M n^2}\right) + o(k) \nonumber\\
	&=  k\left(\log \left(\frac{ab}{M}\right)  + 2\log \log n - 2\log n + o(1) \right), 
	\label{eq:bound_f_k_coeff_3}
\end{align}
\fi
and 
	\if\arXivver1
\begin{align}
		\log C &= \left(\frac{n}{2}\right) \cdot \log[(1 - p)(1 - q) ] + \delta \cdot \log\left(\frac{q}{M(1-q)}\right) \nonumber  \\
		&\geq -\left(\frac{n}{2}\right) \cdot \frac{(a + b) \log n}{n} + \delta (\log \log n - \log n) - o(\log n) \nonumber \\
		&= \left(-\frac{a + b}{2} + 1\right)\log n - \delta \log n - o(\log n)
		\label{eq:bound_f_k_coeff_4}.
\end{align}
\else
\begin{align}
	\log C &= \left(\frac{n}{2}\right) \cdot \log[(1 - p)(1 - q) ] + \delta \cdot \log\left(\frac{q}{M(1-q)}\right) \nonumber  \\
	&\geq -\left(\frac{n}{2}\right) \cdot \frac{(a + b) \log n}{n} + \delta (\log \log n - \log n) \nonumber \\
	&\quad - o(\log n) \nonumber \\
	&= \left(-\frac{a + b}{2} + 1\right)\log n - \delta \log n - o(\log n)
	\label{eq:bound_f_k_coeff_4}.
\end{align}
\fi
Now, assembling \eqref{eq:bound_f_k_coeff_1} - \eqref{eq:bound_f_k_coeff_4} and plugging them into \eqref{eq:log_f_k_expand} yields 
\if\arXivver1 
\begin{align}
	\log f(k) &\geq \underbrace{-k\log\left(\frac{2k}{e}\right) - (k + \delta) \log \left(\frac{2(k + \delta)}{e}\right) + k\left[\log \left(\frac{ab}{M}\right) + 2\log\log n\right]}_{=: h(k)} \nonumber \\
	&+ \left(-\frac{a + b}{2} + 1\right)\log n - o(\log n) \nonumber \\
	&= h(k) + \left(-\frac{a + b}{2} + 1\right)\log n - o(\log n).
	\label{eq:bound_f_k_final}
\end{align}
\else
\begin{align}
	\log f(k) &\geq -k\log\left(\frac{2k}{e}\right) - (k + \delta) \log \left(\frac{2(k + \delta)}{e}\right) \nonumber \\
	&\quad + k\left[\log \left(\frac{ab}{M}\right) + 2\log\log n\right] \nonumber \\
	&\quad + \left(-\frac{a + b}{2} + 1\right)\log n - o(\log n) \nonumber \\
	&= h(k) + \left(-\frac{a + b}{2} + 1\right)\log n - o(\log n)
	\label{eq:bound_f_k_final}
\end{align}
where we define 
\begin{align*}
	h(k) &:= -k\log\left(\frac{2k}{e}\right) - (k + \delta) \log \left(\frac{2(k + \delta)}{e}\right) \\
	&\quad + k\left[\log \left(\frac{ab}{M}\right) + 2\log\log n\right].
\end{align*}
\fi
Then, we can estimate $k^*$ by finding the maximum of $h(k)$. Also, since $h(k)$ is concave, this can be done by solving $h^\prime(k) = 0$, which gives us 
\begin{equation*}
	k_0 = \argmax_k h(k) = (1 + o(1))  \left(\sqrt{\frac{ab}{4M}}\right) \log n.
\end{equation*}
Now, by plugging $k = k_0$ into \eqref{eq:bound_f_k_final} we obtain
\begin{align*}
	\max_{k} \log f(k) &\geq h(k_0) + \left(-\frac{a + b}{2} + 1\right)\log n - o(\log n) \\
	&= \left(-\frac{a + b}{2} + \sqrt{\frac{ab}{M}}\right) \log n - o(\log n).
\end{align*}
Plugging this back into \eqref{eq:bound_P_H_A_B_max_f} completes the proof.

\subsection{Proof of Lemma~\ref{lemma:bound_P_E_H_A_B_v2}}
\label{sec:proof_bound_P_E_H_A_B_v2}
The proof is very similar to the one of Lemma~\ref{lemma:bound_P_E_H_A_B} given in Appendix~\ref{sec:bound_P_E_H_A_B_proof}, therefore we will skip those repeated analyses. To begin with, \eqref{eq:bound_E_H_A_E_B_v2} satisfies
	\if\arXivver1
\begin{align*}
	\mathcal{P}\left(E_{H,A}^{(i)} \bigcap E_B^{(i, 0)}\right) &= \sum_{r_2} \mathcal{P}(r_2) \mathcal{P}(r_1 \geq r_2+1) M^{-r_2} \geq \sum_{r_2} \mathcal{P}(r_2) \mathcal{P}(r_1  = r_2+1) M^{-r_2} \\
	&= \sum_{k = 0}^{\frac{n}{2} - \frac{n}{r} - 1} \binom{n/2}{k} q^k (1 - p)^{\frac{n}{2} - k} \cdot \binom{\frac{n}{2} - \frac{n}{r}}{k+1} p^{k+1} (1-p)^{\frac{n}{2} - \frac{n}{r} - k - 1} \cdot M^{-k} \\
	&=  \sum_{k = 0}^{\frac{n}{2} - \frac{n}{r} - 1}  \binom{n/2}{k} \cdot \binom{\frac{n}{2} - \frac{n}{r}}{k+1} \cdot \left( \frac{pq}{M(1 - p)(1_q)}\right)^k \cdot \underbrace{[(1 - p)(1 - q)]^{\frac{n}{2}} \cdot \frac{p}{(1-p)^{\frac{n}{r} + 1}}}_{=: C} \\
	&=: \sum_{k = 0}^{\frac{n}{2} - \frac{n}{r} - 1} f(k).
\end{align*}
\else
\begin{align*}
	&\mathcal{P}\left(E_{H,A}^{(i)} \bigcap E_B^{(i, 0)}\right) \\
	&= \sum_{r_2} \mathcal{P}(r_2) \mathcal{P}(r_1 \geq r_2+1) M^{-r_2} \\
	&\geq \sum_{r_2} \mathcal{P}(r_2) \mathcal{P}(r_1  = r_2+1) M^{-r_2} \\
	&= \sum_{k = 0}^{\frac{n}{2} - \frac{n}{r} - 1} \Bigg[ \binom{n/2}{k} q^k (1 - p)^{\frac{n}{2} - k} \\
	&\quad \cdot \binom{\frac{n}{2} - \frac{n}{r}}{k+1} p^{k+1} (1-p)^{\frac{n}{2} - \frac{n}{r} - k - 1} \cdot M^{-k} \Bigg]\\
	&=  \sum_{k = 0}^{\frac{n}{2} - \frac{n}{r} - 1}  \binom{n/2}{k} \cdot \binom{\frac{n}{2} - \frac{n}{r}}{k+1} \cdot \left( \frac{pq}{M(1 - p)(1_q)}\right)^k \\
	&\quad \cdot \underbrace{[(1 - p)(1 - q)]^{\frac{n}{2}} \cdot \frac{p}{(1-p)^{\frac{n}{r} + 1}}}_{=: C} \\
	&=: \sum_{k = 0}^{\frac{n}{2} - \frac{n}{r} - 1} f(k).
\end{align*}
\fi
To proceed, we take the logarithm over $f(k)$,
\if\arXivver1
\begin{equation*}
\log f(k) = \log\binom{n/2}{k} + \log\binom{\frac{n}{2} - \frac{n}{r}}{k + 1} + k\log\left(\frac{pq}{M(1- p)(1 - q)}\right) + \log C.
\end{equation*}
\else
\begin{equation*}
	\begin{aligned}
		\log f(k) &= \log\binom{n/2}{k} + \binom\log{\frac{n}{2} - \frac{n}{r}}{k + 1} \\
		&\quad+ k\log\left(\frac{pq}{M(1- p)(1 - q)}\right) + \log C.
	\end{aligned}
\end{equation*}
\fi
By considering $k = \Theta(\log n)$ and following the similar analysis in Appendix~\ref{sec:bound_P_E_H_A_B_proof}, we obtain
\begin{align*}
\log \binom{\frac{n}{2} - \frac{n}{r}}{k + 1} &\geq (k+1) \left[\log\left(\frac{n}{2}\right) - \log\left(\frac{k}{e}\right)\right] - o(\log n)
\end{align*}
and 
\begin{align*}
\log C = \left(-\frac{a + b}{2} + 1\right) \log n  - o(\log n).
\end{align*}
The other two terms are identical to \eqref{eq:bound_f_k_coeff_1} and \eqref{eq:bound_f_k_coeff_3}. Given the above, we have 
	\if\arXivver1
\begin{align*}
	\log f(k) &\geq \underbrace{k\left[2\log\log n + \log\left(\frac{ab}{M}\right) - 2k\log\left(\frac{k}{e}\right) - 2\log 2  \right]}_{=: h(k)}  - 
	\left(\frac{a + b}{2}\right) \log n - o(\log n) \\
	&\geq \left(-\frac{a + b}{2} + \sqrt{\frac{ab}{M}}\right) \log n - o(\log n),
\end{align*}
\else
\begin{align*}
	&\log f(k) \\
	&\geq \underbrace{k\left[2\log\log n + \log\left(\frac{ab}{M}\right) - 2k\log\left(\frac{k}{e}\right) - 2\log 2  \right]}_{=: h(k)}  \\
	&\quad - 
	\left(\frac{a + b}{2}\right) \log n - o(\log n) \\
	&\geq \left(-\frac{a + b}{2} + \sqrt{\frac{ab}{M}}\right) \log n - o(\log n),
\end{align*}
	\fi
where the last step comes by finding the minimum of $h(k)$. This completes the proof.

\end{document}